\documentclass[runningheads]{llncs}
\usepackage{a4wide}
\usepackage{wrapfig}
\usepackage{graphicx}
\usepackage{amssymb,amsmath}
\usepackage{xspace,enumerate}
\usepackage[dvipsnames]{xcolor}
\usepackage[colorlinks=true,urlcolor=Blue,citecolor=Green,linkcolor=BrickRed]{hyperref}
\usepackage[capitalise]{cleveref}
\usepackage[utf8]{inputenc}
\usepackage{thmtools}
\usepackage{thm-restate}
\usepackage{authblk}
\usepackage{verbatim}
\usepackage{mathtools}
\usepackage{sidecap}
\usepackage{todonotes}

\spnewtheorem{problem2}{Problem}{\bfseries}{\itshape}
{\itshape}{\rmfamily}
\usepackage[font=small,labelfont=bf,skip=1pt]{caption}

\let\llncssubparagraph\subparagraph
\let\subparagraph\paragraph
\usepackage[compact]{titlesec}
\let\subparagraph\llncssubparagraph

\usepackage[ruled, noline, noend]{algorithm2e}
\usepackage{subcaption}
\usepackage{algpseudocode}
\def\dd{\mathinner{.\,.}}
\newcommand{\cO}{\mathcal{O}}

\makeatletter
\def\smallunderbrace#1{\mathop{\vtop{\m@th\ialign{##\crcr
   $\hfil\displaystyle{#1}\hfil$\crcr
   \noalign{\kern3\p@\nointerlineskip}%
   \tiny\upbracefill\crcr\noalign{\kern3\p@}}}}\limits}
\makeatother

\def\TFS{\textsc{TFS}}
\def\ETFS{\textsc{ETFS}}
\def\PFS{\textsc{PFS}}
\def\MCSR{\textsc{MCSR}}
\def\ETFSA{\textsc{ETFS}-\textsc{ALGO}}
\def\TFSA{\textsc{TFS}-\textsc{ALGO}}
\def\PFSA{\textsc{PFS}-\textsc{ALGO}}
\def\MCSRA{\textsc{MCSR}-\textsc{ALGO}}
\def\MCSRAI{\textsc{MCSRI}-\textsc{ALGO}}
\def\TPM{\textsc{TPM}}
\def\TM{\textsc{TM}}
\def\TMI{\textsc{TMI}}
\def\BA{\textsc{BA}}
\def\ie{\textit{i.e., }}
\def\eg{\textit{e.g., }}

\def\OLDEN {\textsc{OLD}}
\def\TRU {\textsc{TRU}}
\def\MSN {\textsc{MSN}}
\def\DNA {\textsc{DNA}}
\def\SYN {\textsc{SYN}}
\def\SYNB {\textsc{SYN}$_{\textsc{bin}}$}

\makeatletter
\def\mathcolor#1#{\@mathcolor{#1}}
\def\@mathcolor#1#2#3{%
  \protect\leavevmode
  \begingroup
    \color#1{#2}#3%
  \endgroup
}
\makeatother

\newenvironment{packed_enum}{
\begin{description}
	\setlength{\topsep}{0pt}
	\setlength{\partopsep}{0pt}
  \setlength{\itemsep}{0pt}
  \setlength{\parskip}{0pt}
  \setlength{\parsep}{0pt}
}{\end{description}}
\definecolor{darkgreen}{rgb}{0, .6, 0}
\usepackage{microtype}

\authorrunning{Bernardini et al.}

\title{Combinatorial Algorithms for String Sanitization}
\author{Giulia Bernardini\inst{1} \and Huiping Chen\inst{2} \and Alessio Conte\inst{3} \and Roberto Grossi\inst{3,4} \and\\
\vspace{-2mm} Grigorios Loukides\inst{2} \and Nadia Pisanti\inst{3,4} \and Solon P. Pissis\inst{4,5} \and Giovanna Rosone\inst{3}\and Michelle Sweering\inst{5}}

\institute{
Department of Informatics, Systems and Communication, University of Milano-Bicocca, Milan, Italy, \email{giulia.bernardini@unimib.it}
\and
Department of Informatics, King's College London, London, UK\\
\email{[huiping.chen,grigorios.loukides]@kcl.ac.uk}
\and Department of Computer Science, University of Pisa, Pisa, Italy\\
\email{[conte,grossi,pisanti]@di.unipi.it}, \email{giovanna.rosone@unipi.it}
\and ERABLE Team, INRIA, Lyon, France
\and CWI, Amsterdam, The Netherlands, \email{[solon.pissis,michelle.sweering]@cwi.nl}
}

\begin{document}

\maketitle

\begin{abstract}
  String data are often disseminated to support applications such as location-based service provision or DNA sequence analysis. This dissemination, however, may expose sensitive patterns that model confidential knowledge (\eg trips to mental health clinics from a string representing a user's location history). In this paper, we consider the problem of sanitizing a string by concealing the occurrences of sensitive patterns, while maintaining data utility, in two settings that are relevant to many common string processing tasks.

  ~~~~In the first setting, we aim to generate the  minimal-length string that preserves the order of appearance and frequency of all  non-sensitive patterns. Such a string allows accurately performing tasks based on the sequential nature and pattern frequencies of the string. To construct such a string, we propose a time-optimal algorithm, {\TFSA}. We also  propose another time-optimal algorithm, {\PFSA}, which preserves a partial order of appearance of non-sensitive patterns but produces a much shorter string that can be analyzed more efficiently. The strings produced by either of these algorithms are constructed by concatenating non-sensitive parts of the input string. However, it is possible to detect the sensitive patterns by ``reversing'' the concatenation operations. In response, we propose a heuristic, {\MCSRA}, which replaces letters in the strings output by the algorithms  with carefully selected letters, so that sensitive patterns are not reinstated, implausible patterns are not introduced, and occurrences of spurious patterns are prevented.   In the second setting, we aim to generate a string that is at minimal edit distance from the original string, in addition to preserving the order of appearance and frequency of all non-sensitive patterns. To construct such a string, we propose an algorithm, {\ETFSA}, based on solving specific instances of  approximate regular expression matching.

  ~~~~We implemented our sanitization approach that applies {\TFSA}, {\PFSA} and then {\MCSRA} and experimentally show that it is effective and efficient. We also show that {\TFSA} is nearly as effective at minimizing the edit distance as {\ETFSA}, while being substantially more efficient than {\ETFSA}.
\end{abstract}

\newpage

\section{Introduction}\label{sec:introduction}

A large number of applications, in domains ranging from transportation to web analytics and bioinformatics feature data modeled as \emph{strings}, \ie sequences of letters over some finite alphabet. For instance, a string may represent the history of visited locations of one or more individuals, with each letter corresponding to a location.  Similarly, it may represent the history of search query terms of one or more web users, with letters corresponding to query terms, or a medically important part of the DNA sequence of a patient, with letters corresponding to DNA bases. Analyzing such strings is key in applications including location-based service provision,  product recommendation, and DNA sequence analysis. Therefore, such strings are often disseminated beyond the party that has collected them. For example, location-based service providers often outsource their data to data analytics companies who perform tasks such as similarity evaluation between strings~\cite{anliu}, and retailers outsource their data to marketing agencies who perform tasks such as mining frequent patterns from the strings~\cite{odesa}.

However, disseminating a string intact may result in the exposure of confidential knowledge, such as trips to mental health clinics in transportation data \cite{george}, query terms revealing political beliefs or sexual orientation of individuals in web data \cite{netflix}, or diseases associated  with certain parts of DNA data \cite{brad}. Thus, it may be necessary to sanitize a string prior to its dissemination, so that confidential knowledge is not exposed. At the same time, it is important to preserve the utility of the sanitized string, so that data protection does not outweigh the benefits of disseminating the string to the party that disseminates or analyzes the string, or to the society at large. For example, a retailer should still be able to obtain actionable knowledge in the form of frequent patterns from the marketing agency who analyzed their outsourced data; and researchers should still be able to perform analyses such as identifying significant patterns in DNA sequences.

\subsection{Our Model and Settings} Motivated by the discussion above, we introduce the following model which we call \emph{Combinatorial String Dissemination} (CSD). In CSD, a party has a string $W$ that it seeks to disseminate, while satisfying a set of \emph{constraints} and a set of desirable \emph{properties}. For instance, the constraints aim to capture privacy requirements and the properties aim to capture data utility considerations (\eg posed by some other party based on applications). To satisfy both, $W$ must be transformed to a string  by applying a sequence of edit operations. The computational task is to determine this sequence of edit operations so that the transformed string satisfies the desirable properties subject to the constraints.
Clearly, the constraints and the properties must be specified  based on the application.

Under the CSD model, we consider two specific settings addressing practical considerations in common string processing applications; the  \emph{Minimal String Length} (\textsc{MSL}) setting, in which the goal is to produce a shortest string that satisfies the set of constraints and the set of desirable properties, and the \emph{Minimal Edit Distance} (\textsc{MED}) setting, in which the goal is to produce a string that satisfies the set of constraints and the set of desirable properties and is at minimal edit distance from $W$. In the following, we discuss each setting in more detail.

\paragraph{\textsc{MSL} Setting} In this setting, the sanitized string $X$ must satisfy the following constraint \textbf{C1}: for an  integer $k>0$, no given length-$k$ substring (also called pattern) modeling confidential knowledge should occur in $X$. We call each such length-$k$ substring a \emph{sensitive pattern}. We aim at finding the shortest possible string $X$ satisfying the following desired properties:
(\textbf{P1}) the order of appearance of all other length-$k$ substrings (\emph{non-sensitive patterns}) is the same in $W$ and in $X$; and
(\textbf{P2}) the frequency of these length-$k$ substrings is the same in $W$ and in $X$. The problem of constructing $X$ in this setting is referred to as {\TFS} (Total order, Frequency, Sanitization). Note that it is straightforward to hide substrings of {\em arbitrary} lengths from $X$, by setting $k$ equal to the length of the shortest substring we wish to hide, and then setting, for each of these substrings, any length-$k$ substring as sensitive.

The \textsc{MSL} setting is motivated by real-world applications involving string dissemination. In these applications, a \emph{data custodian} disseminates the sanitized version $X$ of a string $W$ to a \emph{data recipient}, for the purpose of analysis (\eg mining). $W$ contains confidential information that the data custodian needs to hide, so that it does not occur in $X$. Such information is specified by the data custodian based on domain expertise, as in \cite{abul,liyue,sbsh,odesa}. At the same time, the data recipient specifies \textbf{P1} and \textbf{P2} that $X$ must satisfy in order to be useful. These properties map directly to common data utility considerations in string analysis. By satisfying \textbf{P1}, $X$ allows tasks based on the sequential nature of the string, such as blockwise $q$-gram distance computation~\cite{DBLP:journals/almob/GrossiIMPPRV16}, to be performed accurately. By satisfying \textbf{P2}, $X$ allows computing the frequency of length-$k$ substrings and hence mining frequent length-$k$ substrings~\cite{DBLP:journals/bmcbi/Pissis14} with no utility loss. We require that $X$ has minimal length so that it does not contain redundant information. For instance, the string which is constructed by concatenating all non-sensitive length-$k$ substrings in $W$ and separating them with a special letter that does not occur in $W$,  satisfies \textbf{P1} and \textbf{P2} but is not the shortest possible. Such a string $X$ will have a negative impact on the efficiency of any subsequent analysis tasks to be performed on it.

\paragraph{\textsc{MED} Setting} In this setting, the sanitized version $X_{\text{ED}}$ of string $W$ must satisfy the properties \textbf{P1} and \textbf{P2}, subject to the constraint \textbf{C1}, and also be at minimal edit distance from string $W$. Constructing such a string $X_{\text{ED}}$ allows many tasks that are based on edit distance to be performed accurately. Examples of such tasks are frequent pattern mining~\cite{sdm16edit},  clustering~\cite{vldbj08}, entity extraction~\cite{icde19} and range query answering~\cite{tkde14qa}, which are important in domains such as bioinformatics~\cite{sdm16edit}, text mining~\cite{icde19}, and speech recognition~\cite{icassp}.

Note, existing works for sequential data sanitization (\eg  \cite{liyue,sbsh,icdm13,odesa,streamevent}) or anonymization (\eg \cite{condensation,bonomi,chen}) cannot be applied to our settings (see Section~\ref{sec:finale} for details).

\subsection{Our Contributions}

We define the {\TFS} problem for string sanitization and a variant of it, referred to as {\PFS} (Partial order, Frequency, Sanitization), which aims at producing an even shorter string $Y$  by relaxing \textbf{P1} of {\TFS}. We also develop algorithms for {\TFS} and {\PFS}.  Our algorithms construct strings $X$ and $Y$ using a separator letter $\#$, which is not contained in the alphabet of $W$, ensuring that sensitive patterns do not occur in $X$ or $Y$. The algorithms repeat proper substrings of sensitive patterns so that the frequency of non-sensitive patterns overlapping with sensitive ones does not change. For $X$, we give a deterministic construction which may be easily reversible (\ie it may enable a data recipient to construct $W$ from $X$), because the occurrences of $\#$ reveal the exact location of sensitive patterns. For $Y$, we give a construction which breaks several ties arbitrarily, thus being less easily reversible. We further address the reversibility issue by defining the {\MCSR} (Minimum-Cost Separators Replacement) problem and designing an algorithm for dealing with it. In {\MCSR}, we seek to replace all separators, so that the location of sensitive patterns is not revealed, while preserving data utility.  In addition, we define the problem of constructing $X_{\text{ED}}$ in the \textsc{MED} setting, which is referred to as {\ETFS} (Edit-distance, Total order, Frequency, Sanitization), and design an algorithm  to solve it.

Our work makes the following specific contributions:

\vspace{+1mm}
    \noindent\textbf{1.} We design an algorithm, {\TFSA}, for solving the {\TFS} problem in $\cO(kn)$ time,
    where $n$ is the length of $W$. In fact, we prove that $\cO(kn)$ time is worst-case optimal by showing that the length of $X$ is in $\Theta(kn)$ in the worst case.
    The output of {\TFSA} is a string $X$ consisting of a sequence of substrings over the alphabet of $W$ separated by $\#$ (see Example~\ref{ex1} below). An important feature of our algorithm, which is useful in the efficient construction of $Y$ discussed next, is that it can be implemented to produce an $\cO(n)$-sized representation of $X$ with respect to $W$ in $\cO(n)$ time. See Section~\ref{sec:ts}.

\begin{example} \label{ex1}
    Let $W=\texttt{aabaaacbcbbbaabbacaab}$, $k=4$, and the set of sensitive patterns be $\{\texttt{baaa},\texttt{bbaa}\}$.
    The string $X=\texttt{aabaa\#aaacbcbbba\#baabbacaab}$ consists of three substrings over the alphabet $\{\texttt{a},\texttt{b},\texttt{c}\}$ separated by $\#$. Note that no sensitive pattern occurs in $X$, while all non-sensitive substrings of length $k=4$ have the same frequency in $W$ and in $X$ (\eg $\texttt{aaba}$ appears once), and they appear in the same order in $W$ and in $X$ (\eg $\texttt{aaba}$ precedes $\texttt{abaa}$). Also, note that any shorter string than $X$ would either create sensitive patterns or change the frequencies (\eg removing the last letter of $X$ creates a string in which $\texttt{caab}$ no longer appears).
    \qed \end{example}

\noindent\textbf{2.} We define the {\PFS} problem relaxing \textbf{P1} of {\TFS} to produce shorter strings that are more efficient to analyze. Instead of a \emph{total order} (\textbf{P1}), we require a \emph{partial order} ({$\mathbf{\boldsymbol{\Pi} 1}$}) that preserves the order of appearance only for sequences of successive non-sensitive length-$k$ substrings that overlap by $k-1$ letters.
This makes sense because the order of two successive non-sensitive length-$k$ substrings with no length-$(k-1)$ overlap has anyway been ``interrupted'' (by a sensitive pattern). We exploit this observation to shorten the string further. Specifically, we design an algorithm that solves {\PFS} in the optimal $\cO(n +|Y|)$ time, where $|Y|$ is the length of $Y$, using the $\cO(n)$-sized representation of $X$. See Section~\ref{sec:ps}.

\begin{example} \emph{(Cont'd from Example \ref{ex1})}\label{ex2}
    Recall that $W=\texttt{aabaaacbcbbbaabbacaab}$.
    A string $Y$ is $\texttt{aaacbcbbba\#aabaabbacaab}$. The order of \texttt{aaba} and \texttt{abaa} is preserved in $Y$ since they are successive, non-sensitive, and with an overlap of $k-1=3$ letters. The order of \texttt{abaa} and \texttt{aaac}, which are successive non-sensitive, is not preserved since they do not have an overlap of $k-1=3$ letters.
    \qed \end{example}

    \noindent\textbf{3.} We define the {\MCSR} problem, which seeks to produce a string $Z$, by deleting or replacing all separators in $Y$ with letters from the alphabet of $W$ so that: no sensitive patterns are reinstated in $Z$; occurrences of spurious patterns that may not be mined from $W$ but can be mined from $Z$, at  a given support threshold $\tau$, are prevented; and the distortion incurred by the replacements in $Z$ is bounded. The first requirement is to preserve privacy and the next two to preserve data utility. We show that {\MCSR} is NP-hard and propose a heuristic to attack it. We also show how to apply the heuristic, so that letter replacements do not result in \emph{implausible} (\ie{} statistically unexpected) patterns that may reveal the location of sensitive patterns. See Section~\ref{sec:z}.

    \begin{example} \emph{(Cont'd from Example \ref{ex2})}\label{exmcsr}
    Recall that $Y=\texttt{aaacbcbbba\#aabaabbacaab}$. Let $\tau=1$.
    A string $Z=\texttt{aaacbcbbba}\textbf{c}\texttt{aabaabbacaab}$ is produced by replacing letter $\#$ with letter $\texttt{c}$. Note that $Z$ contains no sensitive pattern, nor a non-sensitive pattern of length-$4$ substring that could not be mined from $W$ at a support threshold $\tau$ (i.e., a pattern that does not occur in $W$). In addition, $Z$ contains no implausible pattern,  such as $\texttt{bbab}$, which is not expected to occur in $W$, according to an established  statistical significance measure for strings \cite{brendel,mireille,almir}.
    \qed \end{example}

    \noindent {\bf 4.} We design an algorithm for solving the {\ETFS} problem. The algorithm, called {\ETFSA}, is based on a connection between {\ETFS} and the approximate regular expression matching problem~\cite{regex}. Given a string $W$ and a regular expression $E$, the latter problem seeks to find a string $T$ that matches $E$ and is at minimal edit distance from $W$. {\ETFSA} solves the {\ETFS} problem in $\cO(k|\Sigma|n^2)$ time, where $|\Sigma|$ is the size of the alphabet of $W$. See Section~\ref{sec:etfsa}.

    \begin{example}
     Let $W=\texttt{aaaaaab}$, $k=4$, and the set of sensitive patterns be $\{\texttt{aaaa},\texttt{aaab}\}$. {\TFSA} constructs string $X=\varepsilon$, where $\varepsilon$ is the empty string, with $d_E(W,X)=7$.   On the contrary, {\ETFSA} constructs string $X_{\text{ED}}=\texttt{aaa\#aab}$ with $d_E(W,X_{\text{ED}})=1<7$. Clearly, string $X_{\text{ED}}$ is more suitable for applications, which are based on measuring sequence similarity. \qed
     \end{example}

    \noindent\textbf{5.} For the \textsc{MSL} setting, we implemented our combinatorial approach for sanitizing a string $W$ (\ie the aforementioned algorithms implementing the pipeline  $W \rightarrow X \rightarrow Y \rightarrow Z$) and show its effectiveness and efficiency on real and synthetic data. We also show that it possible to produce a string $Z$ that does not contain implausible patterns,  while incurring insignificant additional utility loss.
    See Section~\ref{sec:exp}.

    \medskip

    \noindent\textbf{6.}  For the \textsc{MED} setting, we implemented {\ETFSA} and experimentally compared it with \TFSA{}. Interestingly, we demonstrate that {\TFSA} constructs optimal or near-optimal solutions to the {\ETFS} problem in practice. This is particularly encouraging because {\TFSA} is linear in the length of the input string $n$, whereas {\ETFSA} is quadratic in $n$. See Section~\ref{sec:exp}.

    \medskip

    A preliminary version of this paper, without the method that avoids implausible patterns and without contributions 4 and 6, appeared in~\cite{pkdd2019}. Furthermore, we include here all proofs omitted from~\cite{pkdd2019}, as well as additional examples and discussion of related work.

\section{Preliminaries, Problem Statements, and Main Results}

\paragraph{Preliminaries} Let $T=T[0]T[1]\ldots T[n-1]$ be a \emph{string} of length $|T|=n$ over a finite ordered alphabet $\Sigma$ of size $|\Sigma|=\sigma$.
By $\Sigma^*$ we denote the set of all strings over $\Sigma$.
By $\Sigma^k$ we denote the set of all length-$k$ strings over $\Sigma$.
For two positions $i$ and $j$ on $T$, we denote by $T[i\dd j]=T[i]\ldots T[j]$ the \emph{substring} of $T$ that starts at position $i$ and ends at position $j$ of $T$. By $\varepsilon$ we denote the \emph{empty string} of length 0. A \emph{prefix} of $T$ is a substring of the form $T[0\dd j]$, and a suffix of $T$ is a substring of the form $T[i\dd n-1]$. A \emph{proper} prefix (suffix) of a string is not equal to the string itself.
By $\text{Freq}_V(U)$ we denote the number of occurrences of string $U$ in string $V$.  Given two strings $U$ and $V$ we say that $U$ has a {\em suffix-prefix overlap} of length $\ell>0$ with $V$ if and only if the length-$\ell$ suffix of $U$ is equal to the
length-$\ell$ prefix of $V$, \ie $U[|U|-\ell \dd |U|-1]=V[0\dd \ell-1]$.

We fix a string $W$ of length $n$ over an alphabet $\Sigma=\{1,\ldots,n^{\cO(1)}\}$ and an integer $0<k<n$. We refer to a length-$k$ string or a {\em pattern} interchangeably. An occurrence of a pattern is uniquely represented by its starting position.
Let $\mathcal{S}$ be a set of positions over $\{0,\ldots, n-k\}$ with the following closure property: for every $i \in \mathcal{S}$, if there exists $j$ such that $W[j\dd j+k-1]=W[i\dd i+k-1]$, then $j\in \mathcal{S}$. That is, if an occurrence of a pattern is in $\mathcal{S}$ all its occurrences are in $\mathcal{S}$. A substring $W[i\dd i+k-1]$ of $W$ is called {\em sensitive} if and only if $i\in \mathcal{S}$. $\mathcal{S}$ is thus the set of occurrences of sensitive patterns. The difference set $\mathcal{I}=\{0,\ldots, n-k\} \setminus \mathcal{S}$ is the set of occurrences of {\em non-sensitive} patterns.

For any string $U$, we denote by $\mathcal{I}_U$ the set of occurrences of non-sensitive length-$k$ strings over $\Sigma$ in $U$. (We have that $\mathcal{I}_W=\mathcal{I}$.) We call an occurrence $i$ the {\em t-predecessor} of another occurrence $j$ in $\mathcal{I}_U$ if and only if $i$ is the largest element in $\mathcal{I}_U$  that is less than $j$. This relation induces a {\em strict total order} on the occurrences in $\mathcal{I}_U$. We call  $i$ the {\em p-predecessor} of  $j$ in $\mathcal{I}_U$ if and only if $i$ is the t-predecessor of $j$ in $\mathcal{I}_U$ {\em and} $U[i\dd i+k-1]$ has a suffix-prefix overlap of length $k-1$ with $U[j\dd j+k-1]$. This relation induces a {\em strict partial order} on the occurrences in $\mathcal{I}_U$. We call a subset $\mathcal{J}$ of $\mathcal{I}_U$ a \textit{t-chain} (resp., \textit{p-chain}) if for all elements in $\mathcal{J}$ except the minimum one, their t-predecessor (resp., p-predecessor) is also in $\mathcal{J}$. For two strings $U$ and $V$, chains $\mathcal{J}_U$ and $\mathcal{J}_V$ are {\em equivalent}, denoted by $\mathcal{J}_U \equiv \mathcal{J}_V$, if and only if $|\mathcal{J}_U|=|\mathcal{J}_V|$ and $U[u\dd u+k-1]=V[v\dd v+k-1]$, where $u$ is the $j$th smallest element of $\mathcal{J}_U$ and $v$ is the $j$th smallest of $\mathcal{J}_V$, for all $j\le |\mathcal{J}_U|$.

Given two strings $U$ and $V$ the {\em edit distance} $d_E(U,V)$ is defined as the minimum number of elementary edit operations (letter insertion, deletion, or substitution) to transform $U$ to $V$.

The set of {\em regular expressions} over an alphabet $\Sigma$ is defined recursively as follows~\cite{regex}:
 (I) $a\in \Sigma \cup \{\varepsilon\}$, where $\varepsilon$ denotes the empty string, is a regular expression. (II) If $E$ and $F$ are regular expressions, then so are  $EF$, $E|F$, and $E^*$, where $EF$ denotes the set of strings obtained by concatenating a string in $E$ and a string in $F$, $E|F$ is the union of the strings in $E$ and $F$, and $E^*$ consists of all strings obtained by concatenating zero or more strings from $E$. Parentheses are used to override the natural precedence of the operators, which places the operator $^*$ highest, the concatenation next, and the operator $|$ last. We state that a string $T$ \emph{matches} a regular expression $E$, if $T$ is equal to one of the strings in $E$.

\paragraph{Problem Statements and Main Results} We define the following problem for the \textsc{MSL} setting.

\begin{problem2}[{\TFS}]\label{prob:t-string}
Given $W$, $k$, $\mathcal{S}$, and $\mathcal{I}_W$ construct the {\em shortest} string $X$:

\begin{packed_enum}
    \item[\textbf{C1}] $X$ does not contain any sensitive pattern.
    \item[\textbf{P1}] $\mathcal{I}_W \equiv \mathcal{I}_X$, \ie the t-chains $\mathcal{I}_W$ and $\mathcal{I}_X$ are equivalent.
    \item[\textbf{P2}] $\text{Freq}_X(U)=\text{Freq}_W(U)$, for all $U\in\Sigma^{k}\setminus \{W[i\dd i+k-1]:i \in \mathcal{S}\}$.
\end{packed_enum}
\end{problem2}

{\TFS} requires constructing the shortest string $X$ in which all sensitive patterns from $W$ are concealed (\textbf{C1}), while preserving the order (\textbf{P1}) and the frequency (\textbf{P2}) of all non-sensitive patterns. Our first result is the following.

\begin{restatable}{theorem}{claimone}\label{the:fastx}
Let $W$ be a string of length $n$ over $\Sigma=\{1,\ldots,n^{\cO(1)}\}$.
Given $k<n$ and $\mathcal{S}$, {\TFSA} solves Problem~\ref{prob:t-string} in $\cO(kn)$ time, which is worst-case optimal. An $\cO(n)$-sized representation of $X$ can be built in $\cO(n)$ time.
\end{restatable}
\textbf{P1} implies \textbf{P2}, but \textbf{P1}
is a strong assumption that may result in long output strings that are inefficient to analyze. We thus relax \textbf{P1} to require that the order of appearance remains the same only for sequences of successive non-sensitive length-$k$ substrings that also overlap by $k-1$ letters (p-chains). This leads to the following problem for the \textsc{MSL} setting.

\begin{problem2}[\PFS]\label{prob:p-string}
Given $W$, $k$, $\mathcal{S}$, and $\mathcal{I}_W$ construct a {\em shortest} string $Y$:

\begin{packed_enum}
        \item[\textbf{C1}] $Y$ does not contain any sensitive pattern.
        \item[\textbf{$\boldsymbol{\Pi}$1}] There exists an injective function $f$ from the p-chains of $\mathcal{I}_W$ to the p-chains of $\mathcal{I}_Y$ such that $f(\mathcal{J}_W)\equiv \mathcal{J}_W$ for any p-chain $\mathcal{J}_W$ of $\mathcal{I}_W$.
    \item[\textbf{P2}] $\text{Freq}_Y(U)=\text{Freq}_W(U)$, for all $U\in\Sigma^{k}\setminus \{W[i\dd i+k-1]:i \in \mathcal{S}\}$.
\end{packed_enum}
\end{problem2}

Our second result, which builds on Theorem~\ref{the:fastx}, is the following.

\begin{restatable}{theorem}{claimtwo}\label{the:fasty}
Let $W$ be a string of length $n$ over $\Sigma=\{1,\ldots,n^{\cO(1)}\}$.
Given $k<n$ and $\mathcal{S}$, {\PFSA} solves Problem~\ref{prob:p-string} in the optimal $\cO(n+|Y|)$ time.
\end{restatable}

To arrive at Theorems~\ref{the:fastx} and~\ref{the:fasty}, we use a special letter (separator) $\#\notin\Sigma$ when required. However, the occurrences of $\#$ may reveal the locations of sensitive patterns. We thus seek to delete or replace the occurrences of $\#$ in $Y$ with letters from $\Sigma$. The new string $Z$ should not reinstate  sensitive patterns or create implausible patterns. Given an integer threshold $\tau>0$, we call a pattern $U\in \Sigma^k$ a $\tau\textit{-ghost}$ in $Z$ if and only if $\text{Freq}_W(U)<\tau$ but $\text{Freq}_{Z}(U)\geq\tau$. Moreover, we seek to prevent {\em $\tau\text{-ghost}$ occurrences} in $Z$ by also bounding the total \emph{weight} of the {\em letter choices} we make to replace the occurrences of $\#$. This is the {\MCSR} problem. We show that already a restricted version of the {\MCSR} problem, namely, the version when $k=1$, is NP-hard via the {\em Multiple Choice Knapsack} (MCK) problem~\cite{Pissinger}.

\begin{restatable}{theorem}{claimthree}\label{the:fasty'}
The {\MCSR} problem is NP-hard.
\end{restatable}

Based on this connection, we propose a non-trivial heuristic algorithm to attack the {\MCSR} problem for the general case of an arbitrary $k$.

We define the following problem for the \textsc{MED} setting.

\begin{problem2}[{\ETFS}]\label{prob:et-string}
Given $W$, $k$, $\mathcal{S}$, and $\mathcal{I}$, construct a string $X_{\text{ED}}$ which is at minimal edit distance from $W$ and satisfies the following:
\begin{packed_enum}
    \item[\textbf{C1}] $X_{\text{ED}}$ does not contain any sensitive pattern.
    \item[\textbf{P1}] $\mathcal{I}_W \equiv \mathcal{I}_{X_{\text{ED}}}$, \ie the t-chains $\mathcal{I}_W$ and $\mathcal{I}_{X_{\text{ED}}}$ are equivalent.
    \item[\textbf{P2}] $\text{Freq}_{X_{\text{ED}}}(U)=\text{Freq}_W(U)$, for all $U\in\Sigma^{k}\setminus \{W[i\dd i+k-1]:i \in \mathcal{S}\}$.
\end{packed_enum}
\end{problem2}

We show how to reduce any instance of the {\ETFS} problem to some instance of the approximate regular expression matching problem.
In particular, the latter instance consists of a string of length $n$ (string $W$) and a regular expression $E$ of length $\cO(k|\Sigma|n)$. We thus prove the claim of Theorem~\ref{the:et-string} by employing the $\cO(|W|\cdot|E|)$-time algorithm of~\cite{regex}.

\begin{restatable}{theorem}{etfsatheorem}\label{the:et-string}
 Let $W$ be a string of length $n$ over an alphabet $\Sigma$. Given $k<n$ and $\mathcal{S}$, {\ETFSA} solves Problem~\ref{prob:et-string} in $\cO(k|\Sigma|n^2)$ time.
\end{restatable}

\section{{\TFSA{}}}\label{sec:ts}

We convert string $W$ into a string $X$ over alphabet $\Sigma\cup\{\#\}$, $\#\notin\Sigma$, by reading the letters of $W$, from left to right, and appending them to $X$ while enforcing the following two rules:
\smallskip

\noindent\textbf{R1}: When the last letter of a sensitive substring $U$ is read from $W$, we append $\#$ to $X$ (essentially replacing this last letter of $U$ with $\#$). Then, we append the succeeding non-sensitive substring (in the t-predecessor order) after $\#$.

\noindent\textbf{R2}: When the $k-1$ letters before $\#$ are the same as the $k-1$ letters after $\#$, we remove $\#$ and the $k-1$ succeeding letters (inspect Fig.~\ref{fig:rules}).

\medskip

\textbf{R1} prevents $U$ from occurring in $X$, and \textbf{R2} reduces the length of $X$ (\ie allows to hide sensitive patterns with fewer extra letters). Both rules leave unchanged the order and frequencies of non-sensitive patterns. It is crucial to observe that applying the idea behind \textbf{R2} on more than $k-1$ letters would decrease the frequency of some pattern, while applying it on fewer than $k-1$ letters would create new patterns. Thus, we need to consider just \textbf{R2} \textit{as-is}.

\medskip

\begin{figure}[!ht]
\centering
\includegraphics[width=5.9cm]{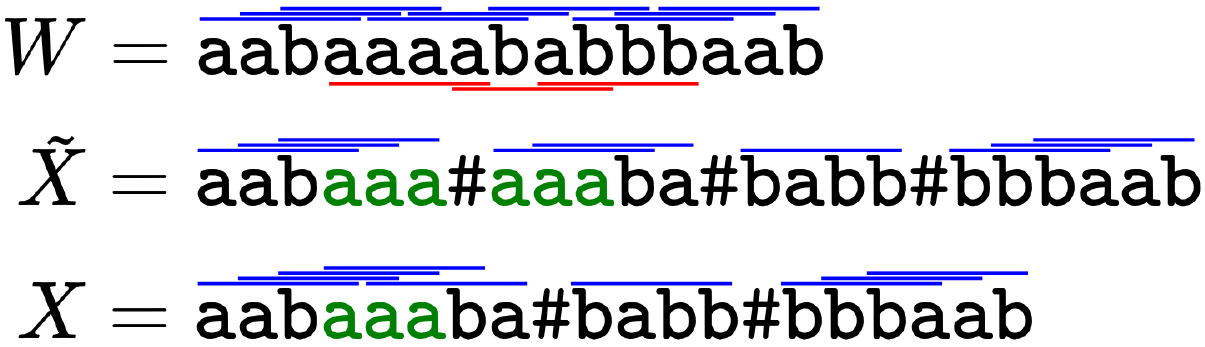}
\caption{Sensitive patterns are underlined in red; non-sensitive patterns are overlined in blue; $\tilde{X}$ is obtained by applying \textbf{R1}; and $X$ by applying \textbf{R1} and \textbf{R2}. In green we highlight an overlap of $k-1=3$ letters.}
\label{fig:rules}
\end{figure}
\medskip

Let $C$ be an array of size $n$ that stores the occurrences of sensitive and non-sensitive patterns: $C[i]=1$ if $i \in \mathcal{S}$ and $C[i]=0$ if $i \in \mathcal{I}$. For technical reasons we set the last $k-1$ values in $C$ equal to $C[n-k]$; \ie~$C[n-k+1]:=\ldots:=C[n-1]:=C[n-k]$.
Note that $C$ is constructible from $\mathcal{S}$ in $\cO(n)$ time.
Given $C$ and $k<n$, {\TFSA} efficiently constructs $X$ by implementing {\bf R1} and {\bf R2} concurrently as opposed to implementing \textbf{R1} and then \textbf{R2} (see the proof of Lemma \ref{lem:properties} for details of the workings of {\TFSA} and Fig.~\ref{fig:rules} for an example). We next show that string $X$ enjoys several properties.

\begin{lemma}
\label{lem:properties}
Let $W$ be a string of length $n$ over $\Sigma$.
Given $k<n$ and array $C$, {\TFSA} constructs the shortest string $X$ such that the following hold:
\begin{itemize}
    \item[(I)] There exists no $W[i\dd i+k-1]$ with $C[i]=1$ occurring in $X$ (\textbf{C1}).
    \item[(II)] $\mathcal{I}_W \equiv \mathcal{I}_X$, \ie the order of substrings $W[i\dd i+k-1]$, for all $i$ such that $C[i]=0$, is the same in $W$ and in $X$; conversely, the order of all substrings $U\in\Sigma^k$ of $X$ is the same in $X$ and in $W$ (\textbf{P1}).
    \item[(III)] $\text{Freq}_X(U) = \text{Freq}_W(U)$, for all $U\in\Sigma^{k}\setminus \{W[i\dd i+k-1]:C[i]=1\}$ (\textbf{P2}).
    \item[(IV)] The occurrences of letter $\#$ in $X$ are at most $\lfloor\frac{n-k+1}{2}\rfloor$ and they are at least $k$ positions apart (\textbf{P3}).
    \item[(V)] $0 \le |X| \le \lceil\frac{n-k+1}{2}\rceil \cdot k + \lfloor\frac{n-k+1}{2}\rfloor$ and these bounds are tight (\textbf{P4}).
\end{itemize}
\end{lemma}
\begin{scriptsize}
\begin{algorithm}[!ht]
 \NoCaptionOfAlgo
 \LinesNotNumbered
 \SetAlgoSkip{}
 \SetKw{report}{report}
 \caption{{\TFSA}$(W\in\Sigma^{n},C,k,\#\notin\Sigma)$}
 \label{alg:ts}
 \nl \label{line1} $X \gets \varepsilon$; $j\gets |W|$; $\ell \gets 0$\;

 \nl \label{line2} $j \gets \min\{i | C[i]=0\}$;\textcolor{gray}{\tcc*[f]{$j$ is the leftmost pos of a non-sens. pattern}}\\
 \nl \label{line3} \If( \textcolor{gray}{\tcc*[f]{Append the first non-sens. pattern to $X$}}){$j+k-1 < |W|$}{

  \nl \label{line4}  $X[0 \dd k-1] \gets W[j \dd j+k-1]$; $j \gets j+k$; $\ell \gets \ell+k$\;
}
  \nl \label{line5} \While( \textcolor{gray}{\tcc*[f]{Examine two consecutive patterns}}){$j < |W|$}{
    \nl \label{line6}   $p \gets j-k$; $c \gets p+1$\;
    \nl\label{line7}    \If(\textcolor{gray}{\tcc*[f]{If both are non-sens., append the last letter of the rightmost one to $X$}}){$C[p]=C[c]=0$}{
    \nl\label{line8}        $X[\ell]\gets W[j]$; $\ell \gets \ell +1$; $j \gets j+1$\;
        }
    \nl \label{line9}   \If(\textcolor{gray}{\tcc*[f]{If the rightmost is sens., mark it and advance $j$}}){$C[p]=0 \land C[c]=1$}{
     \nl  \label{line10}     $f \gets c$;
               $j\gets j+1$\;
        }
     \nl \label{line11}  \lIf(\textcolor{gray}{\tcc*[f]{If both are sens., advance $j$}}){$C[p]=C[c]=1$}{$j\gets j+1$}

     \nl \label{line12}  \If( \textcolor{gray}{\tcc*[f]{If the leftmost is sens. and the rightmost is not}}){$C[p]=1 \land C[c]=0$}{

      \nl\label{line13}      \If(\textcolor{gray}{\tcc*[f]{If the last marked sens. pattern and the current non-sens. overlap by $k-1$, append the last letter of the latter to $X$}}){$W[c \dd c+k-2]= W[f \dd f+k-2]$}{
      \nl \label{line14}        $X[\ell] \gets W[j]$; $\ell \gets \ell +1$; $j \gets j+1$\;
            }

      \nl\label{line15} \Else(\textcolor{gray}{\tcc*[f]{Else append $\#$ and the current non-sens. pattern to $X$}}){
       \nl \label{line16}        $X[\ell]\gets \#$; $\ell \gets \ell +1$\;
       \nl \label{line17}        $X[\ell \dd \ell + k - 1] \gets W[j-k+1 \dd j]$; $\ell \gets \ell+k$; $j \gets j+1$\;
            }
        }
    }
\nl \label{line18}    \report $X$
  \end{algorithm}
\end{scriptsize}
\begin{proof}
    \textbf{C1}: Index $j$ in {\TFSA} runs over the positions of string $W$; at any moment it indicates the ending position of the currently considered length-$k$ substring of $W$. When $C[j-k+1]=1$ (Lines \ref{line9}-\ref{line11}) {\TFSA} never appends $W[j]$, \ie~the last letter of a sensitive length-$k$ substring, implying that, by construction of $C$, no $W[i \dd i+k-1]$ with $C[i]=1$ occurs in $X$.

    \textbf{P1}: When $C[j-k]=C[j-k+1]=0$ (Lines \ref{line7}-\ref{line8}) {\TFSA} appends $W[j]$ to $X$, thus the order of $W[j-k \dd j-1]$ and $W[j-k+1 \dd j]$ is clearly preserved. When $C[j-k]=0$ and $C[j-k+1]=1$, index $f$ stores the starting position on $W$ of the $(k-1)$-length suffix of the last non-sensitive substring appended to $X$ (see also Fig.~\ref{fig:rules}). \textbf{C1} ensures that no sensitive substring is added to $X$ in this case, nor when $C[j-k]=C[j-k+1]=1$. The next letter will thus be appended to $X$ when $C[j-k]=1$ and $C[j-k+1]=0$ (Lines \ref{line12}-\ref{line17}). The condition on Line \ref{line13} is satisfied if and only if the last non-sensitive length-$k$ substring appended to $X$ overlaps with the immediately succeeding non-sensitive one by $k-1$ letters: in this case, the last letter of the latter is appended to $X$ by Line \ref{line14}, clearly maintaining the order of the two. Otherwise, Line \ref{line17} will append $W[j-k+1 \dd j]$ to $X$, once again maintaining the length-$k$ substrings' order.
    Conversely, by construction, any $U \in \Sigma^k$ occurs in $X$ only if it equals a length-$k$ non-sensitive substring of $W$. The only occasion when a letter from $W$ is appended to $X$ more then once is when Line \ref{line17} is executed: it is easy to see that in this case, because of the occurrence of $\#$, each of the $k-1$ repeated letters creates exactly one $U \notin \Sigma^k$, without introducing any new length-$k$ string over $\Sigma$ nor increasing the occurrences of a previous one. Finally, Line \ref{line14} does not introduce any new $U \in \Sigma^k$ except for the one present in $W$, nor any extra occurrence of the latter, because it is only executed when two consecutive non-sensitive length-$k$ substrings of $W$ overlap exactly by $k-1$ letters.

    \textbf{P2}: It follows from the proof for \textbf{C1} and \textbf{P1}.

    \textbf{P3}: Letter $\#$ is added only by Line \ref{line16}, which is executed only when $C[j-k]=1$ and $C[j-k+1]=0$. This can be the case up to $\lceil\frac{n-k+1}{2}\rceil$ times as array $C$ can have alternate values only in the first $n-k+1$ positions. By construction, $X$ cannot start with $\#$ (Lines \ref{line2}-\ref{line4}), and thus the maximal number of occurrences of $\#$ is $\lfloor\frac{n-k+1}{2}\rfloor$. By construction, letter $\#$ in $X$ is followed by at least $k$ letters (Line \ref{line17}): the leftmost non-sensitive substring following a sequence of one or more occurrences of sensitive substrings in $W$.

    \textbf{P4}:{\em Upper bound.} {\TFSA} increases the length of string $X$ by more than one letter only when letter $\#$ is added to $X$ (Line \ref{line16}).
    Every time Lines \ref{line16}-\ref{line17} are executed, the length of $X$ increases by $k+1$ letters.
    Thus the length of $X$ is maximized when the maximal number of occurrences of $\#$ is attained. This length is thus bounded by $\lceil\frac{n-k+1}{2}\rceil \cdot k + \lfloor\frac{n-k+1}{2}\rfloor$.

    {\em Tightness.} For the lower bound, let $W=a^n$ and $a^k$ be sensitive. The condition at Line \ref{line3} is not satisfied because no element in $C$ is set to 0: $j=n$. Then the condition on Line \ref{line5} is also not satisfied because $j=n$, and thus {\TFSA} outputs the empty string. A \emph{de Bruijn sequence} of order $k$ over an alphabet $\Sigma$ is a string in which every possible length-$k$ string over $\Sigma$ occurs exactly once as a substring. For the upper bound, let $W$ be the order-$(k-1)$ de Bruijn sequence over alphabet $\Sigma$, $n-k$ be even, and $\mathcal{S}=\{1,3,5,\ldots,n-k-1\}$. $C[0]=0$ and so Line \ref{line4} will add the first $k$ letters of $W$ to $X$. Then observe that $C[1]=1, C[2]=0; C[3]=1, C[4]=0,\ldots$, and so on; this sequence of values corresponds to satisfying Lines \ref{line12} and \ref{line9} alternately. Line \ref{line9} does not add any letter to $X$. The {\em if} statement on Line \ref{line13} will always fail because of the de Bruijn sequence property. We thus have a sequence of the non-sensitive length-$k$ substrings of $W$ interleaved by occurrences of $\#$ appended to $X$. {\TFSA} thus outputs a string of length $\lceil\frac{n-k+1}{2}\rceil \cdot k + \lfloor\frac{n-k+1}{2}\rfloor$ (see Example~\ref{ex:debruijn}).

 We finally prove that $X$ has minimal length. Let $X_j$ be the prefix of string $X$ obtained by processing $W[0\dd j]$. Let $j_{\min}=\min\{i | C[i]=0\}+k-1$. We will proceed by induction on $j$, claiming that $X_j$ is the shortest string such that \textbf{C1} and \textbf{P1}-\textbf{P4} hold for $W[0 \dd j], \,\,\forall\,j_{\min}\le j \le |W|-1$. We call such a string \emph{optimal}.

\textit{Base case: $j=j_{\min}$.} By Lines \ref{line3}-\ref{line4} of {\TFSA}, $X_j$ is equal to the first non-sensitive length-$k$ substring of $W$, and it is clearly the shortest string such that \textbf{C1} and \textbf{P1}-\textbf{P4} hold for $W[0 \dd j]$.

\textit{Inductive hypothesis and step:} $X_{j-1}$ is optimal for $j>j_{\min}$.
    If $C[j-k]=C[j-k+1]=0$, $X_j=X_{j-1}W[j]$ and this is clearly optimal. If $C[j-k+1]=1$, $X_j=X_{j-1}$ thus still optimal. Finally, if $C[j-k]=1$ and $C[j-k+1]=0$ we have two subcases: if $W[f\dd f+k-2]=W[j-k+1 \dd j-1]$ then $X_j=X_{j-1}W[j]$, and once again $X_j$ is evidently optimal. Otherwise, $X_j=X_{j-1}\#W[j-k+1 \dd j]$. Suppose by contradiction that there exists a shorter $X'_j$ such that \textbf{C1} and \textbf{P1}-\textbf{P4} still hold:  either drop $\#$ or append less than $k$ letters after $\#$. If we appended less than $k$ letters after $\#$, since {\TFSA} will not read $W[j]$ ever again, \textbf{P2}-\textbf{P3} would be violated, as an occurrence of $W[j-k+1 \dd j]$ would be missed.
    Without $\#$, the last $k$ letters of $X_{j-1}W[j-k+1]$ would violate either \textbf{C1} or \textbf{P1} and \textbf{P2} (since we suppose \small$W[f\dd f+k-2]\neq W[j-k+1 \dd j-1]$\normalsize). Then $X_j$ is optimal.~\qed 
\end{proof}

\begin{example}[Illustration of \textbf{P3}]
\label{ex:debruijn}
Let $k=4$. We construct the order-$3$ de Bruijn sequence $W=\texttt{baaabbbaba}$ of length $n=10$ over alphabet $\Sigma=\{\texttt{a},\texttt{b}\}$, and choose $\mathcal{S} = \{1,3,5\}$. {\TFSA} constructs:
\[X=\texttt{baaa\#aabb\#bbba\#baba}.\]
The upper bound of $\lceil\frac{n-k+1}{2}\rceil \cdot k + \lfloor\frac{n-k+1}{2}\rfloor=19$ on the length of $X$ is attained.
\qed \end{example}

Let us now show the main result of this section.

\claimone*
\begin{proof}
For the first part inspect {\TFSA}. Lines \ref{line2}-\ref{line4} can be realized in $\cO(n)$ time. The \emph{while} loop in Line \ref{line5} is executed no more than $n$ times, and every operation inside the loop takes $\cO(1)$ time except for Line \ref{line13} and Line \ref{line17} which take $\cO(k)$ time. Correctness and optimality follow directly from Lemma~\ref{lem:properties} (\textbf{P4}).

For the second part, we assume that $X$ is represented by $W$ and a sequence of pointers $[i,j]$ to $W$ interleaved (if necessary) by occurrences of $\#$. In Line \ref{line17}, we can use an interval $[i,j]$ to represent the length-$k$ substring of $W$ added to $X$. In all other lines (Lines \ref{line4}, \ref{line8} and \ref{line14}) we can use $[i,i]$ as one letter is added to $X$ per one letter of $W$.
By Lemma~\ref{lem:properties} we can have at most $\lfloor\frac{n-k+1}{2}\rfloor$ occurrences of letter $\#$. The check at Line \ref{line13} can be implemented in constant time after linear-time pre-processing of $W$ for longest common extension queries~\cite{DBLP:books/daglib/0020103}. All other operations take in total linear time in $n$. Thus there exists an $\cO(n)$-sized representation of $X$ and it is constructible in $\cO(n)$ time.~\qed
\end{proof}

\section{{\PFSA}}\label{sec:ps}

Lemma~\ref{lem:properties} tells us that $X$ is the shortest string satisfying constraint \textbf{C1} and properties \textbf{P1}-\textbf{P4}. If we were to drop \textbf{P1} and employ the partial order {$\mathbf{\boldsymbol{\Pi} 1}$} (see Problem \ref{prob:p-string}), the length of $X=X_1\#\ldots\#X_N$ would not always be minimal:
if a {\em permutation} of the strings $X_1,\ldots,X_N$ contains pairs $X_i$, $X_j$ with a suffix-prefix overlap of length $\ell=k-1$, we may further apply \textbf{R2}, obtaining a shorter string.

To find such a permutation efficiently and construct a shorter string $Y$ from $W$, we propose {\PFSA}. The crux of our algorithm is an efficient method to solve a variant of the classic NP-complete \emph{Shortest Common Superstring} (SCS) problem~\cite{DBLP:journals/jcss/GallantMS80}. Specifically our algorithm: (I) Computes the string $X$ using Theorem~\ref{the:fastx}. (II) Constructs a collection $\mathcal{B}'$ of strings, each of two letters (two ranks); the first (resp., second) letter is the lexicographic rank of the length-$\ell$ prefix (resp., suffix) of each string in the collection $\mathcal{B}=\{X_1, \ldots, X_N\}$. (III) Computes a shortest string containing every element in $\mathcal{B}'$ as a distinct substring. (IV) Constructs $Y$ by mapping back each element to its distinct substring in $\mathcal{B}$. If there are multiple possible shortest strings, one is selected arbitrarily.

\begin{example}[Illustration of the workings of {\PFSA}]\label{ex:multiplicities}
Let $\ell=k-1=3$ and

$$X=\texttt{aabaa}\#\texttt{aaacbcbbba}\#\texttt{baabbacaab}.$$

The collection $\mathcal{B}$ is comprised of the following substrings: $X_1=\texttt{aabaa}$,    $X_2=\texttt{aaacbcbbba}$, and $X_3=\texttt{baabbacaab}$. The collection $\mathcal{B}'$ is comprised of the following two-letter strings: $\texttt{23},\texttt{14},\texttt{32}$. To construct $B'$, we first find the length-$3$ prefix and the length-$3$ suffix of each $X_i$, $i\in[1,3]$, which leads to a collection $\{\texttt{aab},\texttt{baa},  \texttt{aaa},\texttt{bba}\}$. Then, we sort the collection  lexicographically to obtain $\{\texttt{aaa},\texttt{aab}, \texttt{baa}, \texttt{bba}\}$, and last we replace each $X_i$, $i\in[1,3]$, with the lexicographic ranks of its length-$3$ prefix and length-$3$ suffix. For instance, $X_1$  is replaced by $23$. After that, a shortest string containing all elements of $\mathcal{B}'$ as distinct substrings is computed as: $\texttt{14}\cdot  \texttt{232}$. This shortest string is mapped back to the solution $Y=\texttt{aaacbcbbba}\#\texttt{aabaabbacaab}$.
Note, $Y$ contains one occurrence of $\#$ and has length  $23$, while $X$ contains $2$ occurrences of $\#$ and has length $27$.
\qed \end{example}

We now present the details of {\PFSA}.
We first introduce the {\em Fixed-Overlap Shortest String with Multiplicities} (FO-SSM) problem:  Given a {\em collection} $\mathcal{B}$ of strings $B_1,\ldots,B_{|\mathcal{B}|}$ and an integer $\ell$, with $|B_i|>\ell$, for all $1\leq i \leq |\mathcal{B}|$, FO-SSM seeks to find a shortest string containing each element of $\mathcal{B}$ as a distinct substring using the following operations on any pair of strings $B_i,B_j$:

\begin{itemize}
\item[(I)] $\texttt{concat}(B_i,B_j)=B_i \cdot B_j$;
\item[(II)] $\ell$-$\texttt{merge}(B_i,B_j)=B_i[0\dd |B_i|-1-\ell] B_j[0\dd |B_j|-1] = B_i[0\dd |B_i|-1-\ell] \cdot B_j$.
\end{itemize}

Any solution to FO-SSM with $\ell:=k-1$ and $\mathcal{B}:=X_1,\ldots,X_{N}$ implies a solution to the PFS problem, because $|X_i| > k-1$ for all $i$'s (see Lemma~\ref{lem:properties}, \textbf{P3})

The FO-SSM problem is a variant of the SCS problem. In the SCS problem, we are given a {\em set} of strings and we are asked to compute the shortest common superstring of the elements of this set. The SCS problem is known to be NP-complete, even for binary strings \cite{DBLP:journals/jcss/GallantMS80}. However, if all strings are of length two, the SCS problem admits a linear-time solution~\cite{DBLP:journals/jcss/GallantMS80}. We exploit this crucial detail positively to show a linear-time solution to the FO-SSM problem in Lemma~\ref{lem:fl-scs}. In order to arrive to this result, we first adapt the SCS linear-time solution of~\cite{DBLP:journals/jcss/GallantMS80} to our needs (see Lemma~\ref{lem:2scs}) and plug this solution into Lemma~\ref{lem:fl-scs}.

\begin{lemma}\label{lem:2scs}
Let $\mathcal{Q}$ be a collection of $q$ strings, each of length two, over an alphabet $\Sigma=\{1,\ldots,(2q)^{\cO(1)}\}$. We can compute a shortest string containing every element of $\mathcal{Q}$ as a distinct substring in $\cO(q)$ time.
\end{lemma}
\begin{proof}
We sort the elements of $\mathcal{Q}$ lexicographically in $\cO(q)$ time using radixsort.
We also replace every letter in these strings with their {\em lexicographic rank} from $\{1,\ldots,2q\}$ in $\cO(q)$ time using radixsort. In $\cO(q)$ time we construct the de Bruijn multigraph $G$ of these strings~\cite{CAZAUX2016}. Within the same time complexity, we find all nodes $v$ in $G$ with in-degree, denoted by $\text{IN}(v)$, smaller than out-degree, denoted by $\text{OUT}(v)$. We perform the following two steps:

\paragraph{Step 1} While there exists a node $v$ in $G$ with $\text{IN}(v)<\text{OUT}(v)$, we start an arbitrary path (with possibly repeated nodes) from $v$,
traverse consecutive edges and delete them. Each time we delete an edge, we update the in- and out-degree of the affected nodes.
 We stop traversing edges when a node $v'$ with $\text{OUT}(v')=0$ is reached: whenever $\text{IN}(v')=\text{OUT}(v')=0$, we also delete $v'$ from $G$. Then, we add the traversed path $p=v \ldots v'$ to a set $\mathcal{P}$ of paths. The path can contain the same node $v$ more than once. If $G$ is empty we halt. Proceeding this way, there are no two elements $p_1$ and $p_2$ in $\mathcal{P}$ such that $p_1$ starts with $v$ and $p_2$ ends with $v$; thus this path decomposition is minimal. If $G$ is not empty at the end, by construction, it consists of only cycles.

\paragraph{Step 2} While $G$ is not empty, we perform the following. If there exists a cycle $c$ that {\em intersects} with any path $p$ in $\mathcal{P}$ we splice $c$ into $p$, update $p$ with the result of splicing, and delete $c$ from $G$. This operation can be efficiently implemented by maintaining an array $A$ of size $2q$ of linked lists over the paths in $\mathcal{P}$: $A[\alpha]$ stores a list of  pointers to all occurrences of letter $\alpha$ in the elements of $\mathcal{P}$. Thus in constant time per node of $c$ we check if any such path $p$ exists in $\mathcal{P}$ and splice the two in this case. If no such path exists in $\mathcal{P}$, we add to $\mathcal{P}$ any of the path-linearizations of the cycle, and delete the cycle from $G$. After each change to $\mathcal{P}$, we update $A$ and delete every node $u$ with $\text{IN}(u)=\text{OUT}(u)=0$ from $G$.

  The correctness of this algorithm follows from the fact that $\mathcal{P}$ is a minimal path decomposition of $G$. Thus any concatenation of paths in $\mathcal{P}$ represents a shortest string containing all elements in $\mathcal{Q}$ as distinct substrings.~\qed
\end{proof}

\begin{lemma}\label{lem:fl-scs}
Let $\mathcal{B}$ be a collection of strings over an alphabet $\Sigma=\{1,\ldots,||\mathcal{B}||^{\cO(1)}\}$. Given an integer $\ell$, the FO-SSM problem for $\mathcal{B}$ can be solved in $\cO(||\mathcal{B}||)$ time.
\end{lemma}

\begin{proof}
Consider the following renaming technique. Each length-$\ell$ substring of the collection is assigned a {\em lexicographic rank} from the range $\{1,\ldots,||\mathcal{B}||\}$. Each string in $\mathcal{B}$ is converted to a two-letter string as follows. The first letter is the  lexicographic rank of its length-$\ell$ prefix and the second letter is the lexicographic rank of its length-$\ell$ suffix. We thus obtain a new {\em collection} $\mathcal{B}'$ of two-letter strings. Computing the ranks for all length-$\ell$ substrings in $\mathcal{B}$ can be implemented in $\cO(||\mathcal{B}||)$ time by employing radixsort to sort $\Sigma$ and then the well-known LCP data structure over the concatenation of strings in $\mathcal{B}$~\cite{DBLP:books/daglib/0020103}. The FO-SSM problem is thus solved by finding a shortest string containing every element of $\mathcal{B}'$ as a distinct substring. Since $\mathcal{B}'$ consists of two-letter strings only we can solve the problem in $\cO(|\mathcal{B}'|)$ time by applying Lemma~\ref{lem:2scs}. The statement follows.~\qed
\end{proof}

Thus, {\PFSA} applies Lemma~\ref{lem:fl-scs} on $\mathcal{B}:=X_1,\ldots,X_N$ with $\ell:=k-1$ (recall that $X_1\#\ldots\#X_N=X$). Note that each time the  $\texttt{concat}$ operation is performed, it also places the letter $\#$ in between the two strings.

\begin{lemma}\label{lem:ps}
Let $W$ be a string of length $n$ over an alphabet $\Sigma$.
Given $k<n$ and array $C$, {\PFSA} constructs a shortest string $Y$ with \textbf{C1}, \textbf{$\boldsymbol{\Pi} 1$}, and \textbf{P2-P4}.
\end{lemma}

\begin{proof}
\textbf{C1} and \textbf{P2} hold trivially for $Y$ as no length-$k$ substring over $\Sigma$ is added or removed from $X$. Let $X=X_1\#\ldots\#X_N$.
The order of non-sensitive length-$k$ substrings within $X_i$, for all $i\in[1,N]$, is preserved in $Y$. Thus there exists an injective function $f$ from the p-chains of $\mathcal{I}_W$ to the p-chains of $\mathcal{I}_Y$ such that $f(\mathcal{J}_W)\equiv \mathcal{J}_W$ for any p-chain $\mathcal{J}_W$ of $\mathcal{I}_W$ (\textbf{$\boldsymbol{\Pi} 1$} is preserved). \textbf{P3} also holds trivially for $Y$ as no occurrence of $\#$ is added. Since $|Y|\leq|X|$, for \textbf{P4}, it suffices to note that the construction of $W$ in the proof of tightness in Lemma~\ref{lem:properties} (see also Example~\ref{ex:debruijn}) ensures that there is no suffix-prefix overlap of length $k-1$ between {\em any} pair of length-$k$ substrings of $Y$ over $\Sigma$ due to the property of the order-$(k-1)$ de Bruijn sequence. Thus the upper bound of $\lceil\frac{n-k+1}{2}\rceil \cdot k + \lfloor\frac{n-k+1}{2}\rfloor$ on the length of $X$ is also tight for $Y$.

The minimality on the length of $Y$ follows from
the minimality of $|X|$ and the correctness of Lemma~\ref{lem:fl-scs} that computes a shortest such string.~\qed
\end{proof}

Let us now show the main result of this section.

\claimtwo*

\begin{proof}
We compute the $\cO(n)$-sized representation of string $X$ with respect to $W$ described in the proof of Theorem~\ref{the:fastx}. This can be done in $\cO(n)$ time.
If $X\in\Sigma^*$, then we construct and return $Y:=X$ in time $\cO(|Y|)$ from the representation.
If $X\in(\Sigma\cup\{\#\})^*$, implying $|Y|\leq |X|$, we compute the LCP data structure of string $W$ in  $\cO(n)$ time~\cite{DBLP:books/daglib/0020103}; and implement Lemma~\ref{lem:fl-scs} in $\cO(n)$ time by avoiding to read string $X$ explicitly: we rather rename $X_1,\ldots,X_N$ to a collection of two-letter strings by employing the LCP information of $W$ directly. We then construct and report $Y$ in time $\cO(|Y|)$. Correctness follows directly from Lemma~\ref{lem:ps}.~\qed
\end{proof}

\section{{\MCSR} Problem, {\MCSRA}, and Implausible Pattern Elimination}\label{sec:z}

In the following, we introduce the {\MCSR} problem and prove that it is NP-hard (see Section \ref{sec:z:mcsrproblem}). Then, we introduce {\MCSRA}, a heuristic to address this problem (see Section \ref{sec:z:mcsra}). Finally, we discuss how to configure {\MCSRA} in order to eliminate implausible patterns (see Section \ref{sec:z:implausible}).

\subsection{The {\MCSR} Problem}\label{sec:z:mcsrproblem}

The strings $X$ and $Y$, constructed by {\TFSA} and {\PFSA}, respectively, may contain the separator $\#$, which reveals information about the location of the sensitive patterns in $W$. Specifically, a malicious data recipient can go to the position of a $\#$ in $X$ and ``undo'' Rule {\bf R1} that has been applied by {\TFSA}, removing $\#$ and the $k-1$ letters after $\#$ from $X$. The result could be an occurrence of the sensitive pattern. For example, applying this process to the first $\#$ in $X$ shown in Fig. \ref{fig:rules}, results in recovering the sensitive pattern $\texttt{abab}$. A similar attack is possible on the string $Y$ produced by {\PFSA}, although it is hampered by the fact that substrings within two consecutive $\#$s in $X$ often swap places in $Y$.

To address this issue, we seek to construct a new string $Z$, in which $\#$s are either deleted or replaced by letters from $\Sigma$. To preserve data utility, we favor separator replacements that have a small cost in terms of occurrences of $\tau$-ghosts (patterns with frequency less than $\tau$ in $W$ and at least $\tau$ in $Z)$ and incur a level of distortion bounded by a parameter $\theta$ in $Z$. The cost of an occurrence of a $\tau$-ghost at a certain position is given by function \emph{Ghost}, while function \emph{Sub} assigns a distortion weight to each letter that could replace a $\#$. Both functions will be described in further detail below.

To preserve privacy, we require separator replacements not to reinstate sensitive patterns. This is the {\MCSR} problem, a restricted version of which is presented in Problem \ref{problem:1MCSR}. The restricted version is referred to as $\text{MCSR}_{k=1}$ and differs from {\MCSR} in that it uses $k=1$ for the pattern length instead of an arbitrary value $k>0$.  $\text{MCSR}_{k=1}$ is presented next for simplicity and because it is used in the proof of Lemma \ref{lem:1MCSR}. Lemma \ref{lem:1MCSR} implies Theorem~\ref{the:fasty'}.

\begin{problem2}[$\text{MCSR}_{k=1}$] Given a string $Y$ over an alphabet $\Sigma \cup \{\#\}$ with $\delta>0$ occurrences of letter $\#$, and parameters $\tau$ and $\theta$, construct a new  string $Z$ by substituting the $\delta$ occurrences of $\#$ in $Y$ with letters from $\Sigma$, such that:

\begin{description}
\item[]{\emph{(I)}} $\displaystyle\sum_{\substack{i:Y[i]=\#,\text{ Freq}_{Y}(Z[i])<\tau \\ \text{Freq}_{Z}(Z[i])\geq\tau}}\hspace{-6mm}\text{\emph{Ghost}}(i,Z[i])$ is minimum, and \emph{(II)}  $\displaystyle\sum_{i:Y[i]=\#}\hspace{-2mm}\text{\emph{Sub}}(i,Z[i])\leq \theta$.
\end{description}
\label{problem:1MCSR}
\end{problem2}

\begin{lemma}\label{lem:1MCSR}
The $\text{MCSR}_{k=1}$ problem is NP-hard.
\end{lemma}

\begin{proof}
We reduce the NP-hard \emph{Multiple Choice Knapsack} (MCK) problem \cite{mckp} to $\text{MCSR}_{k=1}$ in polynomial time. In MCK, we are given a set of elements subdivided into $\delta$, mutually exclusive classes, $C_1,\ldots,C_\delta$, and a knapsack. Each class $C_i$ has $|C_i|$ elements. Each element $j \in C_i$ has an arbitrary cost $c_{ij} \geq 0$ and an arbitrary weight $w_{ij}$.
The goal is to minimize the total cost (Eq.~\ref{preq1}) by filling the knapsack with one element from each class (constraint II), such that the weights of the elements in the knapsack satisfy constraint I, where constant $b \geq 0$ represents the minimum allowable total weight of the elements in the knapsack:
\begin{equation}
\min \sum_{i\in[1,\delta]}\sum_{j\in C_i}c_{ij}\cdot x_{ij}\label{preq1}
\end{equation}
subject to the constraints: (I) $ \sum_{i\in[1,\delta]}\sum_{j\in C_i}w_{ij}\cdot x_{ij}\geq b$, (II)$\,\, \sum_{j\in C_i}x_{ij}=1, ~~ i=1,\ldots \delta$, and~ (III) $ x_{ij}\in \{0,1\}, ~~ i=1,\ldots,\delta, ~~ j\in C_i$.

The variable $x_{ij}$ takes value $1$ if the element $j$ is chosen from class $C_i$, $0$ otherwise (constraint III). We reduce any instance $\texttt{I}_{\text{MCK}}$ to an instance $\texttt{I}_{\text{MCSR}_{k=1}}$ in polynomial time, as follows:
\begin{itemize}
\item[(I)] Alphabet $\Sigma$ consists of letters $\alpha_{ij}$, for each $j\in C_i$ and each class $C_i$, $i\in[1,\delta]$.
\item[(II)] We set $Y=\alpha_{11}\alpha_{12}\ldots\alpha_{1|C_1|}\#\ldots\#\alpha_{\delta 1}\alpha_{\delta 2}\ldots\alpha_{\delta|C_\delta|}\#$. Every element of $\Sigma$ occurs exactly once: $\text{Freq}_Y(\alpha_{ij})=1$. Letter $\#$ occurs $\delta$ times in $Y$. For convenience, let us denote by $\mu(i)$ the $i$th occurrence of $\#$ in $Y$.
\item[(III)] We set $\tau=2$ and $\theta = \delta-b$.
\item[(IV)] $\text{Ghost}(\mu(i),\alpha_{ij})=c_{ij}$ and $\text{Sub}(\mu(i),\alpha_{ij})=1-w_{ij}$.
The functions are otherwise {\em not defined}.
\end{itemize}

This is clearly a polynomial-time reduction. We now prove the correspondence between a solution $S_{\texttt{I}_{\text{MCK}}}$ to the given instance $\texttt{I}_{\text{MCK}}$ and a solution $S_{\texttt{I}_{\text{MCSR}_{k=1}}}$ to the instance $\texttt{I}_{\text{MCSR}_{k=1}}$.

We first show that if $S_{\texttt{I}_{\text{MCK}}}$ is a solution to $\texttt{I}_{\text{MCK}}$, then $S_{\texttt{I}_{\text{MCSR}_{k=1}}}$ is a solution to $\texttt{I}_{\text{MCSR}_{k=1}}$. Since the elements in $S_{\texttt{I}_{\text{MCK}}}$ have minimum $\sum_{i\in[1,\delta]}\sum_{j\in C_i}c_{ij}\cdot x_{ij}$, $\text{Freq}_Y(\alpha_{ij})=1$, and $\tau=2$, the letters  $\alpha_1, \ldots, \alpha_\delta$ corresponding to the selected elements lead to a $Z$ that incurs a minimum
\begin{equation}
\sum_{i\in[1,\delta]}\sum_{\substack{j=\mu(i):\text{Freq}_Y(Z[j])<\tau\\ \text{Freq}_{Z}(Z[j])\geq \tau}}\text{Ghost}(j,Z[j]).
\label{theorem6_eq1}
\end{equation}

In addition, each letter $Z[j]$ that is considered by the inner sum of Eq. \ref{theorem6_eq1} corresponds to a single occurrence of $\#$, and these are all the occurrences of $\#$. Thus we obtain that

{\small
\begin{equation}
\sum_{i\in[1,\delta]}\sum_{\substack{j=\mu(i):\text{Freq}_Y(Z[j])<\tau\\ \text{Freq}_{Z}(Z[j])\geq \tau}}\text{Ghost}(j,Z[j])=\displaystyle\sum_{\substack{i:Y[i]=\#,\text{ Freq}_{Y}(Z[i])<\tau \\ \text{Freq}_{Z}(Z[i])\geq\tau}}\text{Ghost}(i,Z[i])
\end{equation}
}

\noindent (\ie condition I in Problem \ref{problem:1MCSR} is satisfied).  Since the elements in $S_{\texttt{I}_{\text{MCK}}}$ have total weight $\sum_{i\in[1,\delta]}\sum_{j\in C_i}w_{ij}\cdot x_{ij}\geq b$, the letters $\alpha_1, \ldots, \alpha_\delta$, they map to, lead to a $Z$ with $\sum_{i\in[1,\delta]}\sum_{j\in C_i}(1 -\text{Sub}(\mu(i),\alpha_i))\cdot x_{ij}\geq \delta-\theta$, which implies

\begin{equation}
    \sum_{i\in[1,\delta]}\sum_{j\in C_i}\text{Sub}(\mu(i),\alpha_{ij})\cdot x_{ij}=
    \displaystyle\sum_{i:Y[i]=\#}\text{Sub}(i,Z[i])\leq \theta
\end{equation}

\noindent (\ie condition II in Problem \ref{problem:1MCSR} is satisfied). $S_{\texttt{I}_{\text{MCSR}_{k=1}}}$ is thus a solution to $\texttt{I}_{\text{MCSR}_{k=1}}$.

We finally show that, if $S_{\texttt{I}_{\text{MCSR}_{k=1}}}$ is a solution to $\texttt{I}_{\text{MCSR}_{k=1}}$, then $S_{\texttt{I}_{\text{MCK}}}$ is a solution to $\texttt{I}_{\text{MCK}}$. Since each $\#_i$, $i\in[1,\delta]$, is replaced by a single letter $\alpha_i$ in $S_{\texttt{I}_{\text{MCSR}_{k=1}}}$, exactly one element will be selected from each class $C_i$ (\ie conditions II-III of MCK are satisfied). Since the letters in $S_{\texttt{I}_{\text{MCSR}_{k=1}}}$ satisfy condition I of Problem~\ref{problem:1MCSR}, every element of $\Sigma$ occurs exactly once in $Y$, and $\tau=2$, their corresponding selected elements $j_1\in C_1, \ldots, j_\delta\in C_\delta$ will have a minimum total cost. Since $S_{\texttt{I}_{\text{MCSR}_{k=1}}}$ satisfies $\sum_{i:Y[i]=\#}\text{Sub}(i,Z[i])=\sum_{i\in[1,\delta]}\sum_{j\in C_i}\text{Sub}(\mu(i),\alpha_{ij})\cdot x_{ij}\leq \theta$, the selected elements $j_1\in C_1, \ldots, j_\delta\in C_\delta$ that correspond to $\alpha_1 \ldots, \alpha_\delta$ will satisfy  $\sum_{i\in[1,\delta]}\sum_{j\in C_i}(1-w_{ij})\cdot x_{ij}\leq \delta-b$, which implies $\sum_{i\in[1,\delta]}\sum_{j\in C_i}w_{ij}\cdot x_{ij}\geq b$ (\ie condition I of MCK is satisfied).  Therefore, $S_{\texttt{I}_{\text{MCK}}}$ is a solution to $\texttt{I}_{\text{MCK}}$. The statement follows.~\qed
\end{proof}

Lemma~\ref{lem:1MCSR} implies the main result of this section.

\claimthree*

The cost of $\tau$-ghosts is captured by a function Ghost. This function assigns a cost to an occurrence of a $\tau\text{-ghost}$, which is caused by a separator replacement at position $i$,  and is specified based on domain knowledge. For example, with a cost equal to $1$ for each gained occurrence of each $\tau\text{-ghost}$, we penalize more heavily a $\tau$-ghost with frequency much below $\tau$ in $Y$  and the penalty increases with the number of gained occurrences. Moreover, we may want to penalize positions towards the end of a temporally ordered string, to avoid spurious patterns that would be deemed important in applications based on time-decaying models \cite{graham1}.

The replacement distortion is captured by a function Sub which assigns a weight to a letter that could replace a $\#$ and is specified based on domain  knowledge. The maximum allowable replacement distortion is $\theta$.  Small weights favor the replacement of separators with desirable letters (\eg letters that reinstate non-sensitive frequent patterns) and letters that reinstate sensitive patterns are assigned a weight larger than $\theta$ that prohibits them from replacing a $\#$. As will be explained in Section \ref{sec:z:implausible}, weights larger than $\theta$ are also assigned to letters which would lead to implausible substrings \cite{icdm13} if they replaced $\#$s.

\subsection{{\MCSRA}}\label{sec:z:mcsra}

We next present {\MCSRA}, a non-trivial heuristic that exploits the connection of the {\MCSR} and MCK~\cite{Pissinger} problems. We start with a high-level description of {\MCSRA}:
\begin{itemize}
    \item[(I)] Construct the set of all candidate $\tau$-ghost patterns (\ie length-$k$ strings  over $\Sigma$ with frequency below $\tau$ in $Y$ that can have frequency at least $\tau$ in $Z$).\\
    \item[(II)] Create an instance of MCK from an instance of {\MCSR}. For this, we map the $i$th occurrence of $\#$ to a class $C_i$ in MCK and each possible replacement of the occurrence with a letter $j$ to a different item in $C_i$. Specifically, we consider all possible replacements with letters in $\Sigma$ and also a replacement with the empty string, which models deleting (instead of replacing) the $i$th occurrence of $\#$.
    In addition, we set the costs and weights that are input to MCK as follows. The cost for replacing the $i$th occurrence of $\#$ with the letter $j$ is set to the sum of the Ghost function for all candidate $\tau$-ghost patterns when the $i$th occurrence of $\#$ is replaced by $j$. That is, we make the worst-case assumption that the replacement forces all candidate $\tau$-ghosts to become $\tau$-ghosts in $Z$. The weight for replacing the $i$th occurrence of $\#$ with letter $j$ is set to $\text{Sub}(i,j)$. \\
    \item[(III)] Solve the instance of MCK and translate the solution back to a (possibly suboptimal) solution of the MCSR problem. For this, we replace the $i$th occurrence of $\#$ with the letter corresponding to the element chosen by the MCK algorithm from class $C_i$, and similarly for each other occurrence of $\#$. If the  instance has no solution (\ie no possible replacement can hide the sensitive patterns), MCSR-ALGO reports that $Z$ cannot be constructed and terminates.
\end{itemize}

Lemma~\ref{lem:mscraeff} below states the running time of an efficient implementation of {\MCSRA}.

\begin{lemma}
{\MCSRA} runs in $\cO(|Y|+k\delta\sigma+\mathcal{T}(\delta,\sigma))$ time, where $\mathcal{T}(\delta,\sigma)$ is the running time of the MCK algorithm for $\delta$ classes with $\sigma+1$ elements each. \label{lem:mscraeff}
\end{lemma}

\begin{proof}
It should be clear that if we conceptually extend $\Sigma$ with the empty string, our approach takes into account the possibility of deleting (instead of replacing) an occurrence of $\#$. To ease comprehension though we only describe the case of letter replacements.

\paragraph{Step 1} Given $Y$, $\Sigma$, $k$, $\delta$, and $\tau$, we construct a set $\mathcal{C}$ of {\em candidate $\tau$-ghosts} as follows. The candidates are at most $(|Y|-k+1-k\delta) + (k\delta\sigma)=\cO(|Y|+k\sigma\delta)$ distinct strings of length $k$. The first term corresponds to all substrings of length $k$ over $\Sigma$ occurring in $Y$ (\ie if $Y$ did not contain $\#$, we would have $|Y|-k+1$ such substrings; each of the $\delta$ $\#$ causes the loss of $k$ such substrings). The second term corresponds to all possible substrings of length $k$ that may be introduced in $Z$ but do not occur in $Y$. For any string $U$ from the set of these $\mathcal{O}(|Y|+k\delta\sigma)$ strings, we want to compute $\text{Freq}_{Y}(U)$ and its {\em maximal frequency} in $Z$, denoted by $\max\text{Freq}_{Z}(U)$, \ie the largest possible frequency that $U$ can have in $Z$, to construct set $\mathcal{C}$. Let $S_{ij}$ denote the string of length $2k-1$, containing the $k$ consecutive length-$k$ substrings, obtained after replacing the $i$th occurrence of $\#$ with letter $j$ in $Y$.

\begin{itemize}
    \item[(I)] If $\text{Freq}_{Y}(U)\geq \tau$, $U$ by definition can never become $\tau$-ghost in $Z$, and we thus exclude it from $\mathcal{C}$. $\text{Freq}_{Y}(U)$, for all $U$ occurring in $Y$, can be computed in $\cO(|Y|)$ total time using the suffix tree of $Y$~\cite{DBLP:books/daglib/0020103}.
    \item[(II)] If $\max\text{Freq}_{Z}(U)<\tau$, $U$ by definition can never become $\tau$-ghost in $Z$, and we thus exclude it from $\mathcal{C}$. $\max\text{Freq}_{Z}(U)$ can be computed by adding to $\text{Freq}_{Y}(U)$, the maximum additional number of occurrences of $U$ caused by a letter replacement among all possible letter replacements. We sum up this quantity for each $U$ and for all replacements of occurrences of $\#$ to obtain $\max\text{Freq}_{Z}(U)$. To do this, we first build the generalized suffix tree of $Y,S_{11},\ldots,S_{\delta\sigma}$ in $\cO(|Y|+k\delta\sigma)$ time~\cite{DBLP:books/daglib/0020103}. We then spell $S_{i1},\ldots,S_{i\sigma}$, for all $i$, in the generalized suffix tree in $\cO(k\sigma)$ time per $i$.
We exploit suffix links to spell the length-$k$ substrings of $S_{ij}$ in $\cO(k)$ time and memorize the maximum number of occurrences of $U$ caused by replacing the $i$th occurrence of $\#$ among all $j$. We represent set $\mathcal{C}$ on the generalized suffix tree by marking the corresponding nodes,
and we denote this representation by $T(\mathcal{C})$. The total size of this representation is $\cO(|Y|+k\sigma\delta)$.
\end{itemize}

\paragraph{Step 2} We now want to construct an instance of the MCK problem using $T(\mathcal{C})$.
We first set letter $j$ as element $\alpha_{ij}$ of class $C_i$.
We then set $c_{ij}$ equal to the sum of the Ghost function cost incurred by replacing the $i$th occurrence of $\#$ by letter $j$ for all (at most $k$) affected length-$k$ substrings that are marked in $T(\mathcal{C})$. The main assumption of our heuristic is precisely the fact that we assume that this letter replacement will force all of these affected length-$k$ substrings becoming $\tau$-ghosts in $Z$.
The computation of $c_{ij}$ is done as follows. For each $(i,j)$, $i\in[1,\delta]$ and $j\in[1,\sigma]$, we have $k$ substrings whose frequency changes, each of length $k$. Let $U$ be one such pattern occurring at position $t$ of $Z$, where $\mu(i)-k+1\leq t \leq \mu(i)$ and $\mu(i)$ is the $i$th occurrence of $\#$ in $Y$. We check if $U$ is marked in $T(\mathcal{C})$ or not. If $U$ is not marked we add nothing to $c_{ij}$. If $U$ is marked, we increment $c_{ij}$ by $\text{Ghost}(t,U)$. We also set $w_{ij}=\text{Sub}(i,j)$ (as stated above, any letter that reinstates a sensitive pattern is assigned a weight $\text{Sub}>\theta$, so that it cannot be selected to replace an occurrence of $\#$ in Step $3$). Similar to Step 1, the total time required for this computation is $\cO(|Y|+k\sigma\delta)$.

\paragraph{Step 3} In Step 2, we have computed $c_{ij}$ and $w_{ij}$, for all $i,j$, $i\in[1,\delta]$ and $j\in[1,\sigma]$.
  We thus have an instance of the MCK problem. We solve it and translate the solution back to a (suboptimal) solution of the MCSR problem: the element $\alpha_{ij}$ chosen by the MCK algorithm from class $C_i$ corresponds to letter $j$ and it is used to replace the $i$th occurrence of $\#$, for all $i\in[1,\delta]$. The cost of solving MCK depends on the chosen algorithm and is given by a function $\mathcal{T}(\delta,\sigma)$.

Thus, the total cost of {\MCSRA} is $\cO(|Y|+k\delta\sigma+
\mathcal{T}(\delta,\sigma))$.~\qed
\end{proof}

\subsection{Eliminating Implausible Patterns}\label{sec:z:implausible}

We present the notion of implausible substring and explain how we can ensure that implausible patterns do not occur in $Z$, as a result of applying the {\MCSRA} algorithm to string $Y$.

Consider, for instance, an input string $Y=\ldots\texttt{a\#c}\ldots$ that models the movement of an individual, and the string $\texttt{abc}$, which is created as a substring of $Z$ when we replace $\#$ with $\texttt{b}$.  Consider further that an individual can, generally, not go from $\texttt{a}$ to $\texttt{c}$ through $\texttt{b}$, or that it is highly unlikely for them to do so. We call a substring such as $\texttt{abc}$ {\em implausible}. Clearly, if $\texttt{abc}$ occurs in $Z$, it may be possible for an attacker to infer that $\texttt{b}$ replaced $\#$, and then infer a sensitive pattern by ``undoing'' {\bf R1} as explained in Section  \ref{sec:z:mcsrproblem}. In order to effectively model this scenario, we define implausible patterns based on a statistical significance  measure for strings \cite{brendel,mireille,almir}. The measure is defined as follows \cite{brendel}:
$$z_{W}(U)=\frac{\text{Freq}_{W}(U)-\mathbb{E}_W[U]}{\max(\sqrt{\mathbb{E}_W[U]},1)},$$

\noindent where $U$ is a string with $|U|>2$, $W$ is the reference string, and
$$\mathbb{E}_W[U]=\begin{cases}\frac{\text{Freq}_W(U[0\dd|U|-2])\cdot\text{Freq}_W(U[1\dd|U|-1])}{\text{Freq}_W(U[1\dd|U|-2])}, & \text{Freq}_W(U[1\dd|U|-2])>0\\
0, \mbox{otherwise}
\end{cases}
$$

\noindent is the expected frequency of $U$ in $W$, computed based on an independence assumption between the event ``$U[0\dd |U|-1]$ occurs in $W$'' and ``$U[1\dd |U|-1]$ occurs in $W$''. The measure $z_W$ is a normalized version of the standard score of $U$, based on the fact that the variance  $\text{Var}_W[U]\approx \sqrt{\mathbb{E}_W[U]}$~\cite{mireille}. A small $z_W(U)$ indicates that $U$ occurs less likely than expected, and hence it can naturally be considered as an artefact of sanitization.

Given a user-defined threshold $\rho<0$, we define a string $U$ as {\em $\rho$-implausible} if $z_W(U)<\rho$. The set of $\rho$-implausible substrings of $W$ can be computed in the optimal $\cO(|\Sigma|\cdot|W|)$ time~\cite{almir}. We use $W$ as the reference string, assuming that it is a good representation of the domain; \eg a trip (substring) that is $\rho$-implausible in $W$ is also implausible in general.  Alternatively, one could use any other string as reference, impose length constraints on implausible patterns~\cite{LoukidesPNAS,manolis}, or even directly specify substrings that should not occur in $Z$ based on domain knowledge.

Given the set $\mathcal{U}$ of ($\rho$-)implausible patterns, we ensure that no $\#$ replacement creates $U=U_1\alpha U_2 \in \mathcal{U}$ in $Z$, where $\alpha$ is the letter that replaces $\#$, by assigning a weight $\text{Sub}(i,Z[i])>\theta$, for each $Z[i]$ such that $Y[i]=\#$ and $U_1 \cdot Z[i] \cdot U_2 \in \mathcal{U}$. This guarantees that no replacement leading to an artefact occurrence of an element of $\mathcal{U}$ is performed by \MCSRA. Note, however, that a $\rho$-implausible pattern may occur in $Z$ as a substring, either because it occurred in a part of $W$ that was copied to $Z$ (\eg a non-sensitive pattern), or due to the change of frequency of some substrings that are created in $Z$ after the replacement of a $\#$. However, since such $\rho$-implausible patterns did not contain a $\#$ in the first place, they cannot be exploited by an attacker seeking to reverse the construction of $Z$.

\section{{\ETFSA}}\label{sec:etfsa}
Let $U$ and $V$ be two non-sensitive length-$k$ substrings of $W$ such that $U$ is the $t$-predecessor of $V$. Since $U$ and $V$ must occur in the same order in the solution string $X_{\text{ED}}$, the main choice we have to make in order to solve the {\ETFS} problem is whether to:
\begin{enumerate}

\item[(I)] ``merge'' $U$ and $V$ when the length-$(k-1)$ suffix of $U$ and the length-$(k-1)$ prefix of $V$ match; or \item[(II)] ``interleave'' $U$ and $V$ with a carefully  selected string over $\Sigma\cup\{\#\}$.
\end{enumerate}

Among operations I and II, for every such pair $U$ and $V$, we must select the operation that \emph{globally} results in the smallest number of edit operations. Operations I and II can naturally be  expressed by means of a regular expression $E$. In particular, this implies that any instance of the {\ETFS} problem can be  reduced to an instance of approximate regular expression matching and thus an algorithm for approximate regular expression matching between $E$ and $W$~\cite{regex} can be employed. More formally, given a string $W$ and a regular expression $E$, the {\em approximate regular expression matching} problem is to find a string $T$ that matches $E$ with minimal $d_E(W,T)$. The following result is known.
\begin{theorem}[\cite{regex}]\label{the:regexp}
Given a string $W$ and a regular expression $E$,
the approximate regular expression matching problem can be solved in $\cO(|W|\cdot|E|)$ time.
\end{theorem}

In the following, we define a specific type of a regular expression $E$.
Let us first define the following regular expression: $$\Sigma^{<k}=(\underbrace{(a_1|a_2|\ldots|a_{|\Sigma|}|\varepsilon)\ldots (a_1|a_2|\ldots|a_{|\Sigma|}|\varepsilon)}_{\text{$k-1$ times}}),$$
\noindent where $\Sigma=\{a_1,a_2,\ldots,a_{|\Sigma|}\}$ is the alphabet of $W$ and $k>1$. We also define the following regular-expression gadgets, for a letter $\# \notin \Sigma$:
\newcommand{\sltk}{\ensuremath{\Sigma^{<k}}}
$$ \oplus = \#(\sltk\#)^*, \quad \ominus = (\sltk\#)^*, \quad \otimes =(\# \sltk)^*.$$

Intuitively, the gadget $\oplus$ represents a string we may choose to include in the output in an effort to minimize the edit distance between $W$ and the solution string $X_{\text{ED}}$. It should be clear that the length of $\oplus$ is in $\cO(k|\Sigma|)$ and that $\oplus$ cannot generate any length-$k$ substring over $\Sigma$. Furthermore, inserting $\oplus$ in $E$ cannot create any sensitive or non-sensitive pattern due to the occurrences of $\#$ on both ends of $\oplus$. The gadgets $\ominus$ and $\otimes$ are similar to $\oplus$. They are added in the beginning and at the end of $E$, respectively. This is because $E$ should not start or end with $\#$ as this would only increase the edit distance to $W$. As it will be explained later, to  construct $E$, we also make use of the $|$ operator.  Intuitively, the $|$ operator represents the choice we make between operation ``merge'' or ``interleave''.

We are now in a position to describe {\ETFSA}, an algorithm for solving the {\ETFS} problem. {\ETFSA} starts by constructing $E$. Let $(N_1,N_2\ldots,N_{|\mathcal{I}|})$ be the sequence of non-sensitive length-$k$ substrings as they occur in $W$ from left to right. We first set $E=\ominus N_1$ and then process the pairs of non-sensitive length-$k$ substrings $N_i$ and $N_{i+1}$, for all $i\in\{1,|\mathcal{I}|-1\}$.
At the $i$th step, we examine whether or not $N_i$ and $N_{i+1}$ can be merged. If they can, we append to $E$ a regular expression $(A|\oplus N_{i+1})$, where $A$ is obtained by chopping-off the length-$(k-1)$ prefix of $N_{i+1}$ (that is, the remainder of $N_{i+1}$ after merging it with $N_i$). Otherwise, we append $\oplus N_{i+1}$ to $E$. Intuitively, using $A$ corresponds to choosing ``merge'' and $\oplus N_{i+1}$ to choosing ``interleave''. After examining each pair $N_i$ and $N_{i+1}$, we append $\otimes$ to $E$. This concludes the construction of $E$. Note how, for any combination of choices, $N_{i+1}$ will always appear in the string obtained.

Next, {\ETFSA} employs Theorem~\ref{the:regexp}  to construct $X_{\text{ED}}$. In particular, it finds a string $T$ that matches $E$ with minimal $d_E(W,T)$. Last, it sets $X_{\text{ED}}=T$. We arrive at the main result of this section.

\etfsatheorem*

\begin{proof} Constructing $E$ can be done in $\cO(n+kn+|E|)=\cO(k|\Sigma|n)$ time, since: (I)
The non-sensitive length-$k$ substrings of $W$ can be obtained in $\cO(n)$ time, by reading $W$ from left to right and checking $\mathcal{S}$. (II)
Checking whether $N_i$ and $N_{i+1}$ are mergeable takes  $\cO(k)$ time via letter comparisons, and it is performed in each of the $\cO(n)$ steps. (III) The length is $|E|=\cO(kn+k|\Sigma|n)=\cO(k|\Sigma|n)$. This is because $E$ contains at most $n$ occurrences of non-sensitive length-$k$ substrings, at most $n$ occurrences of $\oplus$, and one occurrence of each of $\ominus$ and $\otimes$ and because the lengths of $\oplus$, $\ominus$ and $\otimes$ are $\cO(k|\Sigma|)$.

Computing $T$ from $W$ and $E$ can be performed in $\cO(|W|\cdot|E|)=\cO(n\cdot |E|)$ time using Theorem~\ref{the:regexp}. Thus  {\ETFSA} takes $\cO(k|\Sigma|n^2)$ time in total.

The correctness of {\ETFSA} follows from the fact that by construction: (I) $T$ does not contain any sensitive pattern, so {\bf C1} is satisfied; (II) $T$ satisfies {\bf P1} and {\bf P2} as no length-$k$ substring over $\Sigma$ (other than the non-sensitive ones) is inserted in $E$; (III) All strings satisfying {\bf C1}, {\bf P1} and {\bf P2} can be obtained by $E$, since they must have the same t-chain of non-sensitive patterns over $\Sigma^*$ as $W$, interleaved by length-$k$ substrings that are on $(\Sigma\cup\#)^*$ but \textit{not} on $\Sigma^*$; and (IV) the minimality on edit distance is guaranteed by Theorem~\ref{the:regexp}. The statement follows.~\qed
\end{proof}

\begin{example}[Illustration of the workings of {\ETFSA}]\label{example:etfsa}
Let $W=\texttt{aaabbaabaccbbb}$, $k=4$, and
the set of sensitive patterns be $\{\texttt{aabb}, \texttt{abba}, \texttt{bbaa}, \texttt{baab}, \texttt{ccbb}\}$.
The sequence of non-sensitive patterns is thus $(N_1, \ldots, N_6) =(\texttt{aaab}, \texttt{aaba}, \texttt{abac}, \texttt{bacc}, \texttt{accb}, \texttt{cbbb})$.
Given that $k=4$ and $\Sigma=\{\texttt{a},\texttt{b}, \texttt{c}\}$, {\ETFSA} constructs  the following gadgets,
$$\oplus=
\#(\Sigma^{<4}\#)^*=\#(((\texttt{a}|\texttt{b}|\texttt{c}|\varepsilon)(\texttt{a}|\texttt{b}|\texttt{c}|\varepsilon)(\texttt{a}|\texttt{b}|\texttt{c}|\varepsilon))\#)^*$$

\vspace{-2mm}
$$\ominus=(\Sigma^{<4}\#)^*=(((\texttt{a}|\texttt{b}|\texttt{c}|\varepsilon)(\texttt{a}|\texttt{b}|\texttt{c}|\varepsilon)(\texttt{a}|\texttt{b}|\texttt{c}|\varepsilon))\#)^*$$

\vspace{-2mm}
$$\otimes=(\#\Sigma^{<4})^*=(\#((\texttt{a}|\texttt{b}|\texttt{c}|\varepsilon)(\texttt{a}|\texttt{b}|\texttt{c}|\varepsilon)(\texttt{a}|\texttt{b}|\texttt{c}|\varepsilon)))^*$$

\noindent and sets $E=\ominus N_1=\ominus\texttt{aaab}$. Then, it iterates  over each pair of successive non-sensitive length-$k$ substrings in the order they appear in $W$ (\ie pair $(N_i, N_{i+1})$ is considered in Step $i\in [1,5]$) and the regular expression $E$ is  updated, as detailed below.

In Step $1$, {\ETFSA} considers the pair $(N_1, N_2)=(\texttt{aaab},\texttt{aaba})$. Observe that in this case $N_1$ and  $N_2$ can be merged, since the length-$3$ suffix of $N_1$ and the length-$3$ prefix of $N_2$ match.
Thus, $(A|N_2)=(\texttt{a}|\oplus\texttt{aaba})$ is appended to $E$. Recall that when merging, we chop off the length-$(k-1)$ prefix of $N_{i+1}=N_2$ (because we have merged it already) and write down what is left of $N_{2}$ (\texttt{a} in this case) before  $|$. Thus, $E=\ominus \texttt{aaab}(\texttt{a}|\oplus\texttt{aaba})$.

In Step 2, {\ETFSA}  considers $(N_2, N_3)=(\texttt{aaba},\texttt{abac})$. Again, $N_2$ and  $N_3$ can be merged. Thus, $(\texttt{c}| \oplus \texttt{abac})$ is appended into $E$, which leads to
$E=\ominus \texttt{aaab}(\texttt{a}|\oplus\texttt{aaba})(\texttt{c}|\oplus
\texttt{abac})$.

In Steps 3 and 4, {\ETFSA} considers the pairs $(N_3, N_4)=(\texttt{abac}, \texttt{bacc})$ and $(N_4, N_5)=(\texttt{bacc}, \texttt{accb})$, respectively. Since the patterns in each pair can be merged, the algorithm appends into $E$ the regular expression $(\texttt{c}| \oplus \texttt{bacc})$ and $(\texttt{b}| \oplus \texttt{accb})$, for the first and second pair, respectively. This leads to
$E=\ominus \texttt{aaab}(\texttt{a}|\oplus\texttt{aaba})(\texttt{c}|\oplus
\texttt{abac})(\texttt{c}|\oplus
\texttt{bacc})(\texttt{b}|\oplus
\texttt{accb})$.

In Step 5, {\ETFSA} considers the last pair $(N_5, N_6)=(\texttt{accb},\texttt{cbbb})$, which cannot be merged, and appends $\oplus \texttt{cccb}$ to $E$. Since there is no other pair to be considered, $\otimes$ is also appended to $E$, leading to:
$$E=\ominus \texttt{\underline{aaab}}(\texttt{a}|\underline{\oplus\texttt{aaba}})(\texttt{\underline{c}}|\oplus
\texttt{abac})(\texttt{\underline{c}}|\oplus
\texttt{bacc})(\texttt{\underline{b}}|\oplus
\texttt{accb})\underline{\oplus \texttt{cbbb}} \otimes.$$
At this point, {\ETFSA} employs Theorem~\ref{the:regexp} to find the following string $T$ that matches $E$ (the choices that were made in the construction of $T$ are underlined in $E$ and $\ominus$, $\oplus$, $\otimes$ are matched by the empty string):
$$T=\texttt{aaab\#aabaccb\#cbbb},$$

\noindent with minimal $d_E(T,W)=4$. Last, {\ETFSA} returns $X_{\text{ED}}=T$.
\qed
\end{example}

Note that $X_{\text{ED}}=T$ in Example \ref{example:etfsa} does not contain any sensitive pattern and that all non-sensitive patterns of $W$ appear in $T$ in the same order and with the same frequency as they appear in $W$. Note also that, for the same instance,  {\TFSA} would return string $X=$\texttt{aaabaccb\#cbbb} with $d_E(W,X)=5>d_{E}(W,X_{\text{ED}})=4$ and $|X|=13<|X_{\text{ED}}|=17$.

\section{Experimental Evaluation}\label{sec:exp}

We evaluate our algorithms in terms of {\em effectiveness} and {\em efficiency}. Effectiveness is measured based on data utility and number of implausible patterns. Efficiency is measured based on runtime.

\paragraph{Evaluated Algorithms} First, we consider the pipeline {\TFSA}$\rightarrow$ {\PFSA}$\rightarrow${\MCSRA}, referred to as {\TPM}.
Given  a string $W$ over $\Sigma$, {\TPM} sanitizes $W$ by applying {\TFSA}, {\PFSA}, and then {\MCSRA}. {\MCSRA} uses the  $\cO(\delta\sigma\theta)$-time algorithm of~\cite{Pissinger} for solving the MCK instances. The final output is a string $Z$ over $\Sigma$.  {\MCSRA} is configured with an empty  set $\mathcal{U}$ (\ie it may lead to implausible patterns that are created in $Z$ after the replacement of a $\#$).

We did not compare {\TPM} against existing methods, because they are not alternatives to {\TPM} (see Section \ref{sec:finale} for more details on related work). Instead, we compared {\TPM} against a greedy baseline referred to as {\BA}, in terms of data utility and efficiency. {\BA} initializes its output string $Z_{\text{BA}}$ to $W$ and then considers each sensitive pattern $R$ in $Z_{\text{BA}}$, from left to right. For each $R$, {\BA} replaces the letter $r$ of $R$ that has the largest frequency in $Z_{\text{BA}}$ with another letter $r'$ that is not contained in $R$ and has the smallest frequency in $Z_{\text{BA}}$, breaking all ties arbitrarily. Note that this letter replacement should not introduce any other sensitive pattern in $Z_{\text{BA}}$. If no such $r'$ exists, $r$ is replaced by $\#$ to ensure that a solution is produced  (even if it may reveal the location of a sensitive pattern). Each replacement removes the occurrence of $R$ and aims to prevent $\tau$-ghost occurrences by selecting an $r'$ that will not substantially increase the frequency of patterns overlapping with $R$. Note that  {\BA} does not preserve the frequency of non-sensitive patterns, and thus, unlike {\TPM}, it can incur $\tau$-lost patterns. We also implemented a similar baseline  that replaces the  letter in $R$ that has the smallest frequency in $Z_{\text{BA}}$ with another letter that is not contained in $R$ and has the largest frequency in $Z_{\text{BA}}$, but omit its results as it was worse than {\BA}.

In addition, we consider the pipelines {\TFSA}$\rightarrow${\MCSRA} and {\TFSA}$\rightarrow${\MCSRAI},  referred to as {\TM} and {\TMI}, respectively. With {\MCSRAI} we refer to the configuration of {\MCSR} in which there is a non-empty set $\mathcal{U}$ of $\rho$-implausible patterns that must not occur in the output string $Z$.  We omit {\PFSA} from the {\TM} and {\TMI} pipelines to avoid the elimination of some implausible patterns due to re-ordering of blocks of non-sensitive patterns that is performed by {\PFSA}.

Last, we consider {\ETFSA}, which we compare to {\TFSA}, to demonstrate that the latter is a very  effective heuristic for the {\ETFS} problem.

\paragraph{Experimental Data} We considered the following publicly available datasets used in \cite{abul,sbsh,icdm13,odesa}: Oldenburg ({\OLDEN}), Trucks  ({\TRU}), MSNBC ({\MSN}), the complete genome of {\em Escherichia coli} ({\DNA}), and synthetic data (uniformly random strings, the largest of which is referred to as {\SYN}). See Table~\ref{tabdata} for the characteristics of these datasets and the parameter values used in experiments, unless stated otherwise.

\begin{table}[!ht]
\vspace{+1mm}
\scriptsize\centering
\scalebox{1.0}{
\begin{tabular}{|c|| c || c|c|| c | c | c |c|}
\hline {\bf Dataset} & {\bf Data domain} & {\bf Length}  & {\bf Alphabet}  &  {\bf \# sensitive} & {\bf \# sensitive} & {\bf Pattern} & {\bf Implausible pat.} \tabularnewline
~ & ~ & $n$ & {\bf size} $|\Sigma|$ & {\bf patterns} & {\bf positions} $|\mathcal{S}|$ & {\bf length} $k$ & {\bf threshold} $\rho$
\tabularnewline\hline\hline
 \emph{\OLDEN} & Movement & 85,563 & 100 & $[30, 240]$ ~~$\mathbf{(60)}$ & $[600, 6103]$ & ~$[3, 7]$ ~~$\mathbf{(4)}$ & $[-2,-0.1]$ ~$\mathbf{(-1)}$ \tabularnewline \hline
  \emph{\TRU}  & Transportation & 5,763 & 100 & $[30,120]$ ~~$\mathbf{(10)}$ & $[324, 2410]$ & ~$[2, 5]$ ~~$\mathbf{(4)}$ & $[-3,-0.1]$ ~$\mathbf{(-4)}$ \tabularnewline \hline
  \emph{\MSN}  & Web & 4,698,764 & 17 & $[30, 120]$ ~~$\mathbf{(60)}$ & $[6030,320480]$ & ~$[3, 8]$ ~~$\mathbf{(4)}$ & $[-6,-3]$ ~$(\mathbf{-1})$ \tabularnewline \hline
  \emph{\DNA}  & Genomic & 4,641,652 & 4 & ~$[25, 500]$ ~~$\mathbf{(100)}$ & $[163, 3488]$ & $[5, 15]$ $\mathbf{(13)}$ & $[-4.5,-2.5]$ ~$\mathbf{(-2.5)}$ \tabularnewline \hline
  \emph{\SYN} & Synthetic & 20,000,000 & 10 & ~$[10,1000]$ ~$\mathbf{(1000)}$ & $[10724,20171]$& ~$[3,6]$ ~~$\mathbf{(6)}$ & - \tabularnewline \hline
  \emph{\SYNB} & Synthetic & 1,000 & 2 & $[4,32]$ $\mathbf{(16)}$ & $[16,128]$&~$[4,7]$~$\mathbf{(4)}$ & - \tabularnewline \hline
\end{tabular}
}
\caption{Characteristics of datasets and values used (default values are in bold).
}\label{tabdata}
\end{table}

\paragraph{Experimental Setup} The sensitive patterns were selected randomly among the frequent length-$k$ substrings at minimum support $\tau$ following \cite{sbsh,icdm13,odesa}.
We used the fairly low values $\tau=10$, $\tau=20$, $\tau=200$, and $\tau=20$ for {\TRU}, {\OLDEN}, {\MSN}, and {\DNA}, respectively, to have a wider selection of sensitive patterns. In {\MCSRA}, we used a  uniform cost of $1$ for every occurrence of each $\tau$-ghost, a weight of $1$ (resp., $\infty$) for each letter replacement that does not (resp., does) create a sensitive pattern, and we  further set $\theta=\delta$. This setup treats all candidate $\tau$-ghost patterns and all candidate letters for replacement uniformly, to facilitate a fair comparison with {\BA} which cannot distinguish between $\tau$-ghost candidates or favor specific letters. In {\MCSRAI}, we instead set a weight $\infty$  for each letter replacement that does not create a sensitive pattern or an implausible pattern of length $k$.

To capture the utility of sanitized data, we used the \emph{(frequency) distortion} measure  $$\sum_{U}(\text{Freq}_{W}(U)-\text{Freq}_{Z}(U))^2,$$ where $U\in \Sigma^k$ is a non-sensitive pattern. The distortion measure quantifies changes in the frequency of non-sensitive patterns with low values suggesting that $Z$ remains useful for tasks based on pattern frequency (\eg identifying motifs corresponding to functional or conserved DNA~\cite{DBLP:journals/bmcbi/Pissis14}).

We also measured the number of $\tau$-ghost and $\tau$-lost patterns in $Z$ following~\cite{sbsh,icdm13,odesa}, where a pattern $U$ is $\tau\textit{-lost}$ in $Z$ if and only if $\text{Freq}_W(U)\geq \tau$ but $\text{Freq}_{Z}(U)<\tau$. That is, $\tau$-lost patterns model knowledge that can no longer be mined from $Z$ but could be mined from $W$, whereas $\tau$-ghost patterns model knowledge that can be mined from $Z$ but not from $W$. A small number of $\tau$-lost/ghost patterns suggests that frequent pattern mining can be accurately performed on $Z$~\cite{sbsh,icdm13,odesa}.
Unlike BA, by design \TPM{} {\em does not} incur any $\tau$-lost pattern, as \TFSA{} and \PFSA{} preserve frequencies of non-sensitive patterns, and \MCSRA{} can only increase pattern frequencies.

To examine the benefit of using {\MCSRAI} instead of {\MCSRA} when implausible patterns need to be eliminated, we measured the percentage of $\rho$-implausible patterns of length $k$ that may  occur in $Z$, when a letter replaces a $\#$. Clearly, the percentage is $0$ when {\MCSRAI} is used, and a large percentage for {\MCSRA} implies that it is beneficial to use {\MCSRAI} instead.

To capture the effectiveness of {\TFSA} in terms of constructing a string $X$ that is at small edit distance from $W$ (see the {\ETFS} problem), we used the \emph{Edit Distance Relative Error}, defined as $$\frac{d_{E}(W,X)-d_{E}(W,X_{\text{ED}})}{d_{E}(W,X_{\text{ED}})}.$$

All experiments ran on a Desktop PC with an Intel Xeon E5-2640 at 2.66GHz and 16GB RAM. Our source code is written in \texttt{C++}. 
The results presented below have been averaged over $10$ runs.

\subsection{{\TPM} vs.~{\BA}}

\paragraph{Data Utility} We first demonstrate that {\TPM} incurs \emph{very low distortion}, which implies high utility for tasks based on the frequency of patterns (\eg \cite{DBLP:journals/bmcbi/Pissis14}). Fig.~\ref{utility_distortion_sens} shows that, for varying number of sensitive patterns, {\TPM} incurred on average $18.4$ (and up to $95$) times lower distortion than {\BA} over all experiments. Also, Fig.~\ref{utility_distortion_sens} shows that  {\TPM} remains effective even in challenging settings, with many sensitive patterns (\eg the last point in Fig. \ref{utility_distortion_sensb} where about $42\%$ of the positions in $W$ are sensitive). Fig.~\ref{utility_distortion_k} shows that, for varying $k$,  {\TPM} caused on average $7.6$ (and up to $14$) times lower distortion than {\BA} over all experiments.

\begin{figure}\hspace{-6mm}
    \begin{subfigure}[b]{0.24\textwidth}
       \includegraphics[width=1\textwidth]{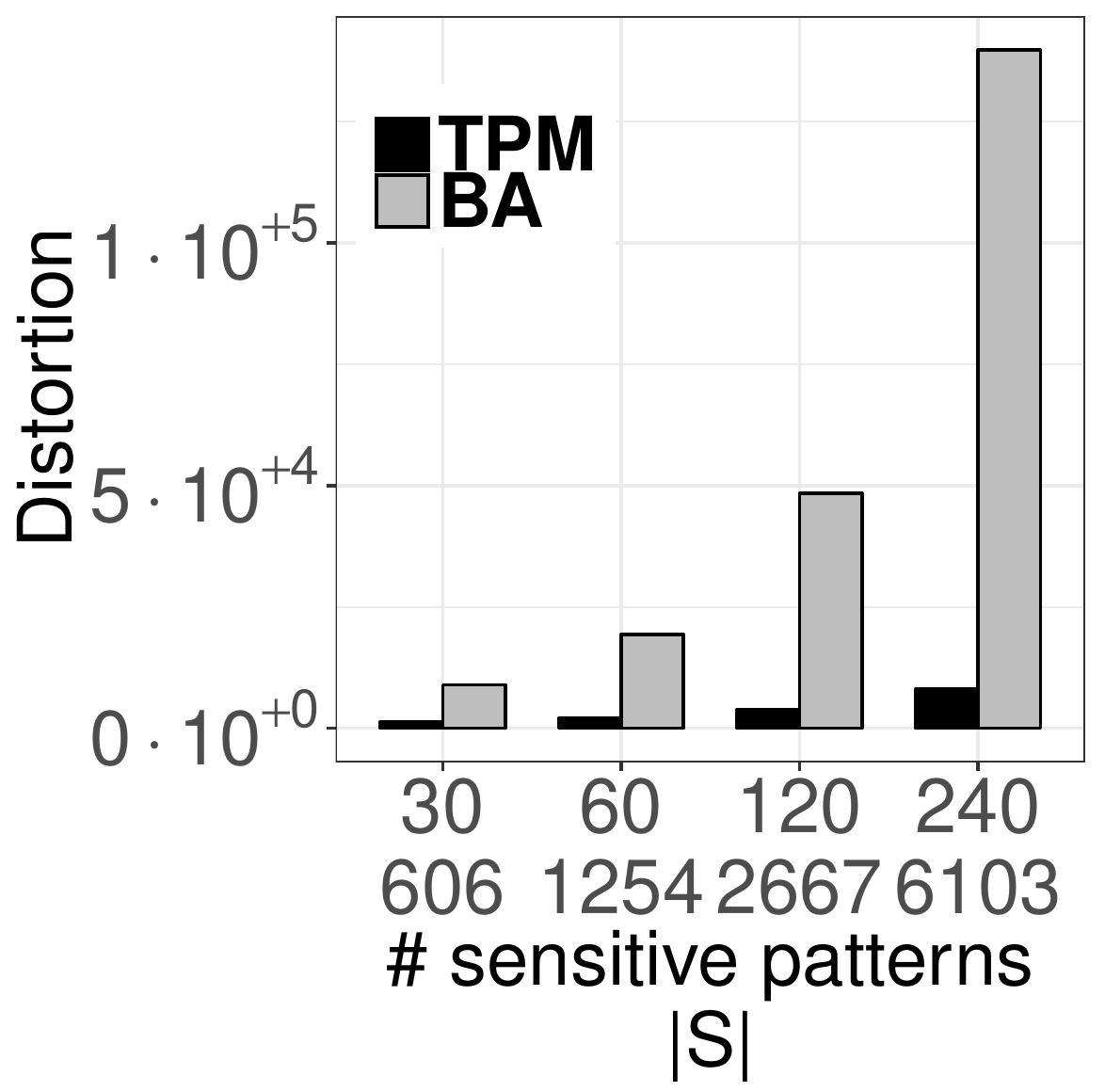}
       \vspace{-6mm}
        \caption{\OLDEN}
        \label{utility_distortion_sensa}
    \end{subfigure}
    ~ 
    \begin{subfigure}[b]{0.24\textwidth}
       \includegraphics[width=1\textwidth]{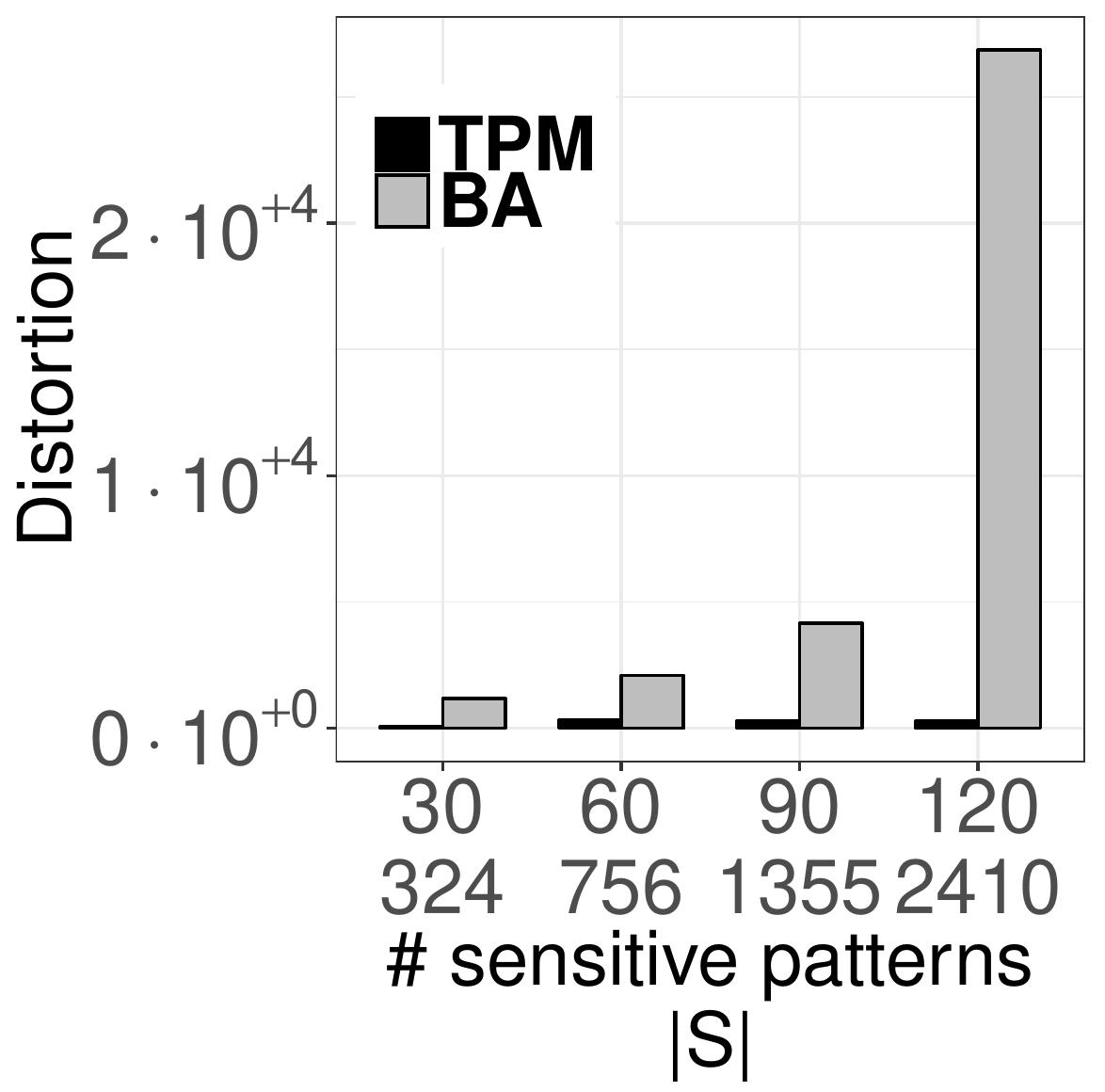}
       \vspace{-6mm}
        \caption{\TRU}
        \label{utility_distortion_sensb}
    \end{subfigure}
    ~ 
    \begin{subfigure}[b]{0.24\textwidth}
       \includegraphics[width=1\textwidth]{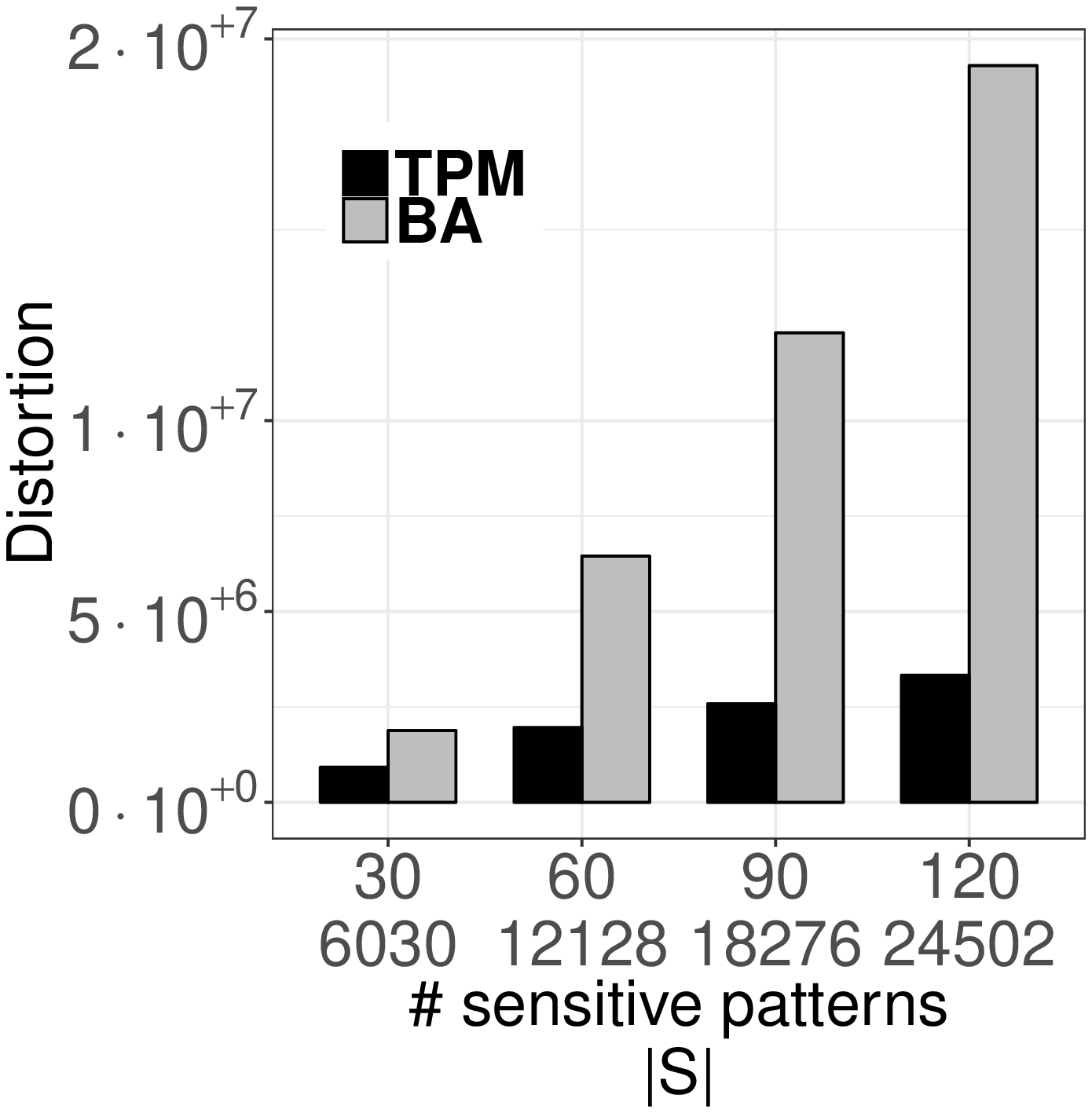}
      \vspace{-6mm}
        \caption{\MSN}
        \label{utility_distortion_sensc}
    \end{subfigure}
    ~ 
      \begin{subfigure}[b]{0.24\textwidth}
       \includegraphics[width=1\textwidth]{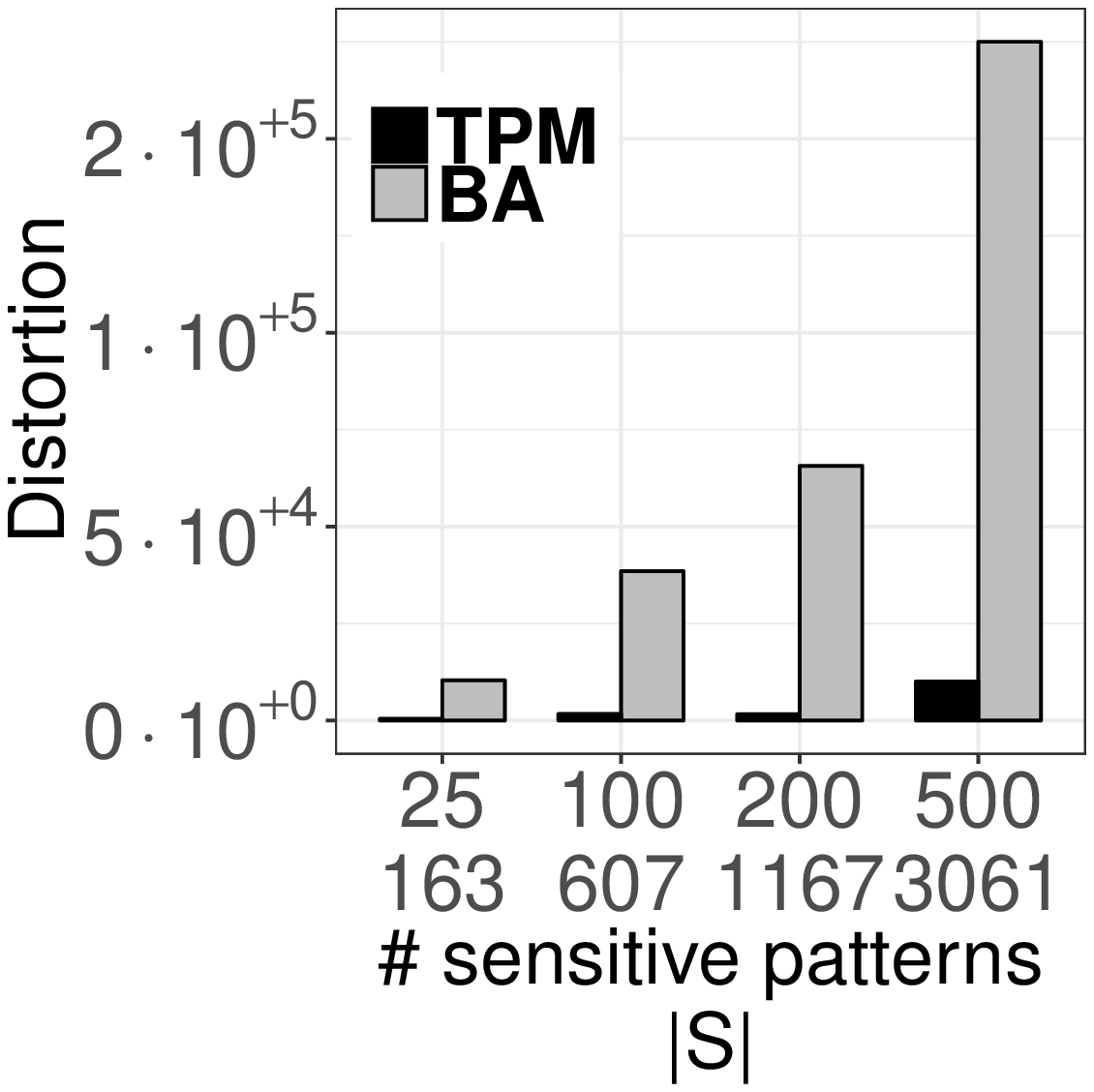}
       \vspace{-6mm}
        \caption{\DNA}
        \label{utility_distortion_sensd}
    \end{subfigure}

    \caption{Distortion vs.~number of sensitive patterns and their total number $|\mathcal{S}|$ of occurrences in $W$~(first two lines on the $X$ axis).} \label{utility_distortion_sens}
\end{figure}
\begin{figure}\hspace{-5mm}
    \centering
    \begin{subfigure}[b]{0.23\textwidth}
       \includegraphics[width=1.05\textwidth]{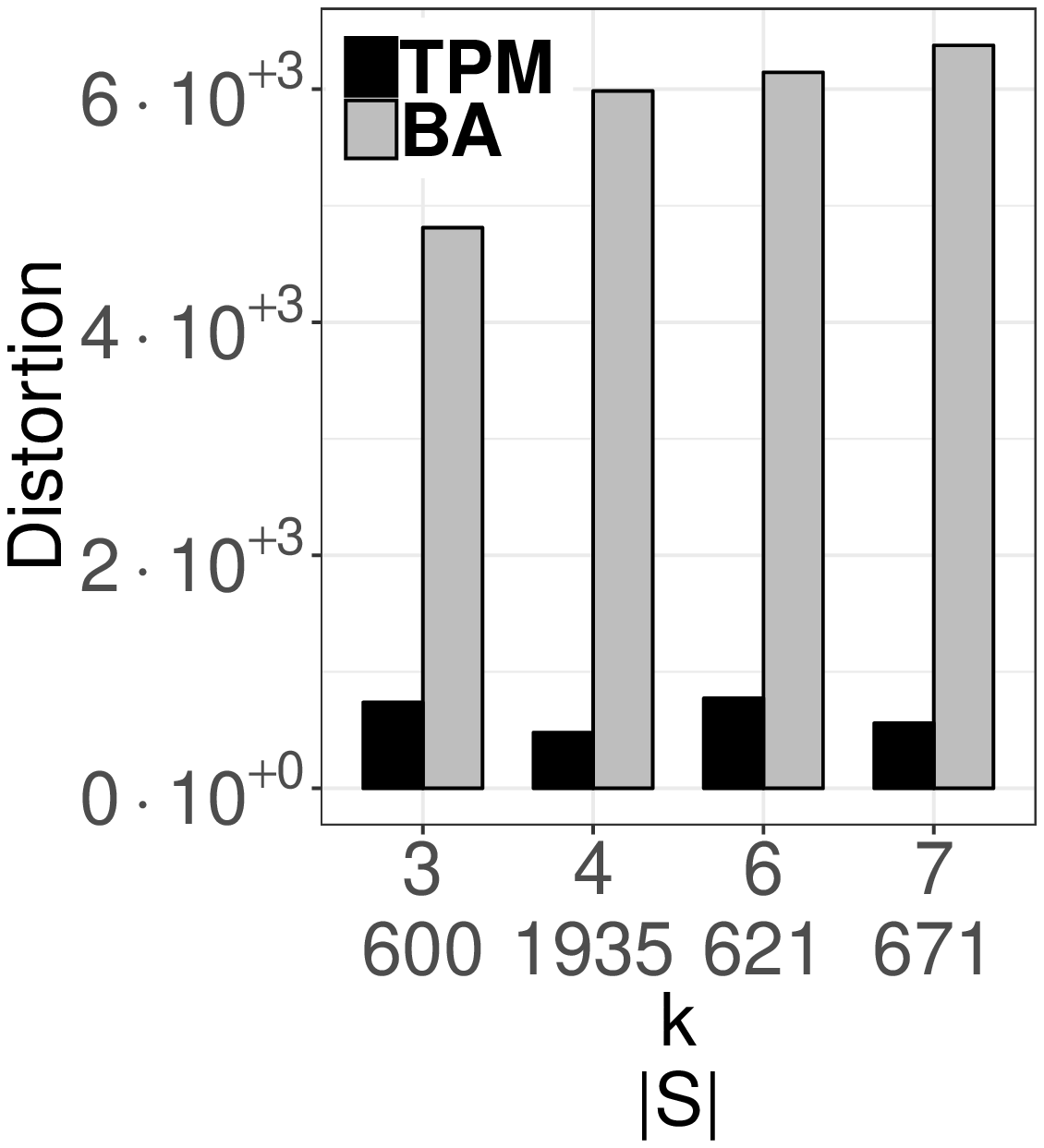}
       \vspace{-7mm}
        \caption{{\OLDEN}}
        \label{utility_distortion_ka}
    \end{subfigure}
    \hspace{-1mm} 
    \begin{subfigure}[b]{0.23\textwidth}
       \includegraphics[width=1.05\textwidth]{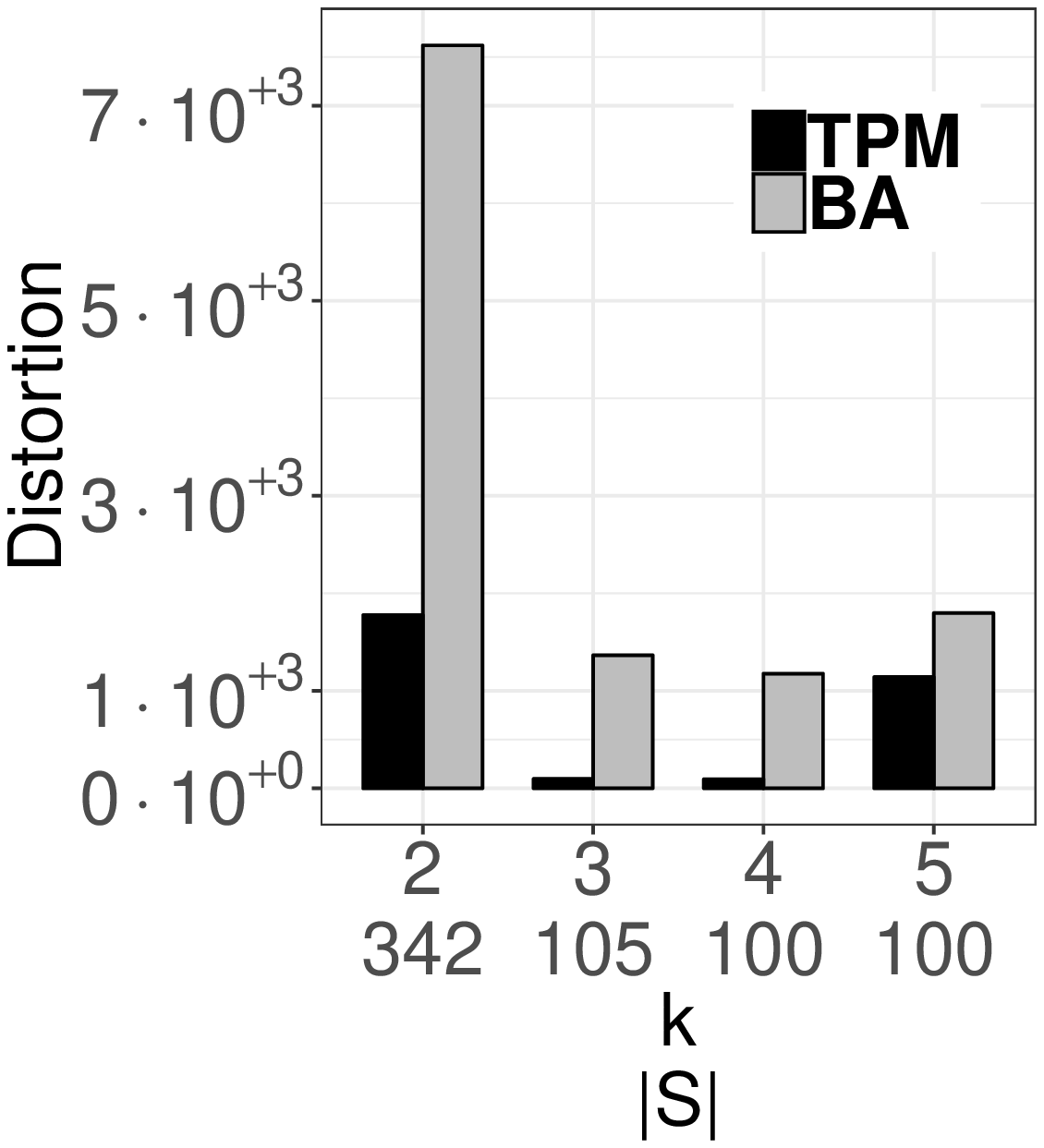}
       \vspace{-7mm}
        \caption{{\TRU}}
        \label{utility_distortion_kb}
    \end{subfigure}
    \hspace{-1mm}
    \begin{subfigure}[b]{0.23\textwidth}
       \includegraphics[width=1.16\textwidth]{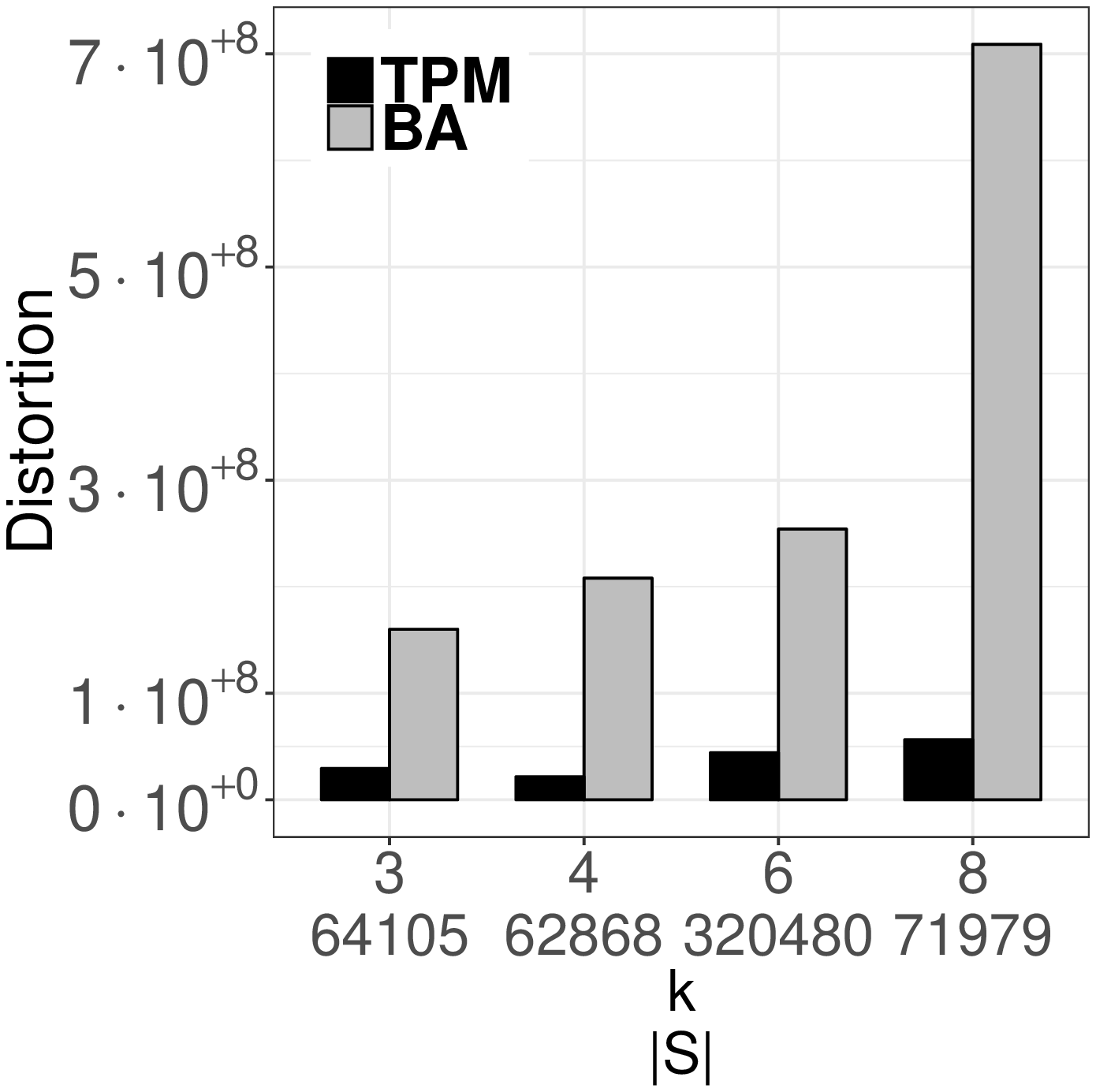}
       \vspace{-7mm}
        \caption{\MSN}
        \label{utility_distortion_kc}
    \end{subfigure}
    \hspace{+2mm}
      \begin{subfigure}[b]{0.23\textwidth}
       \includegraphics[width=1.05\textwidth]{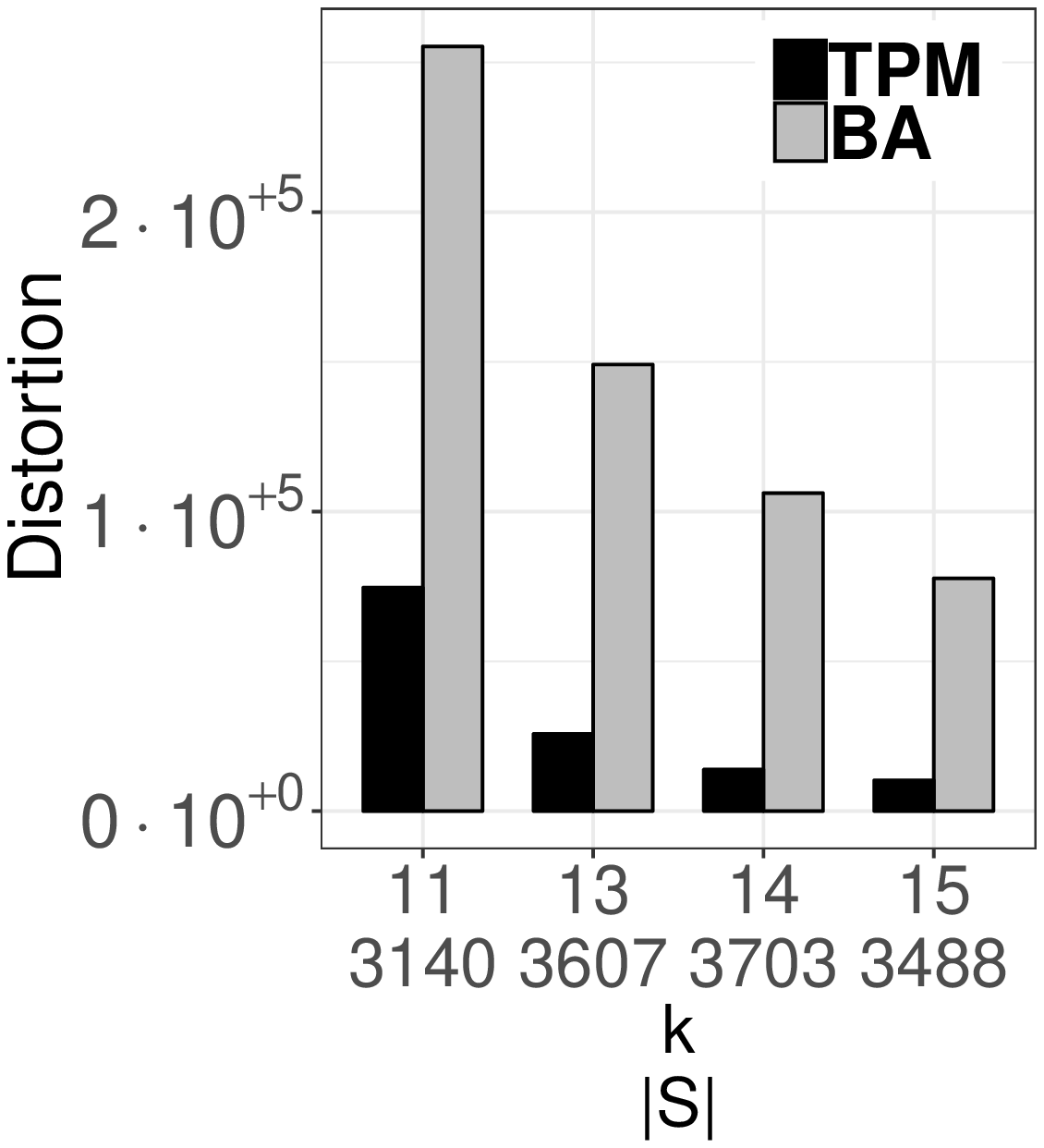}
       \vspace{-7mm}
        \caption{\DNA}
        \label{utility_distortion_kd}
    \end{subfigure}
\caption{Distortion vs.~length of sensitive patterns $k$ (and $|\mathcal{S}|$).}\label{utility_distortion_k}
\end{figure}

Next, we demonstrate that {\TPM} permits \emph{accurate frequent pattern mining}: Fig. \ref{utility_lostghost_sens} shows that  {\TPM} led to no $\tau$-lost or $\tau$-ghost patterns for the {\TRU} and {\MSN} datasets. This implies no utility loss for mining frequent length-$k$ substrings with threshold $\tau$. In all other cases, the number of $\tau$-ghosts was on average $6$ (and up to $12$) times smaller than the total number of $\tau$-lost and $\tau$-ghost patterns for {\BA}. {\BA} performed poorly (\eg up to $44\%$ of frequent patterns became $\tau$-lost for {\TRU} and $27\%$ for {\DNA}). Fig. \ref{utility_lostghost_k} shows that, for varying $k$, {\TPM} led to on average $5.8$ (and up to $19$) times fewer $\tau$-lost/ghost patterns than {\BA}.  {\BA} performed poorly (\eg up to $98\%$ of frequent patterns became $\tau$-lost for {\DNA}).
\begin{figure}\hspace{-5mm}
    \centering
    \begin{subfigure}[b]{0.23\textwidth}
       \includegraphics[width=1.15\textwidth]{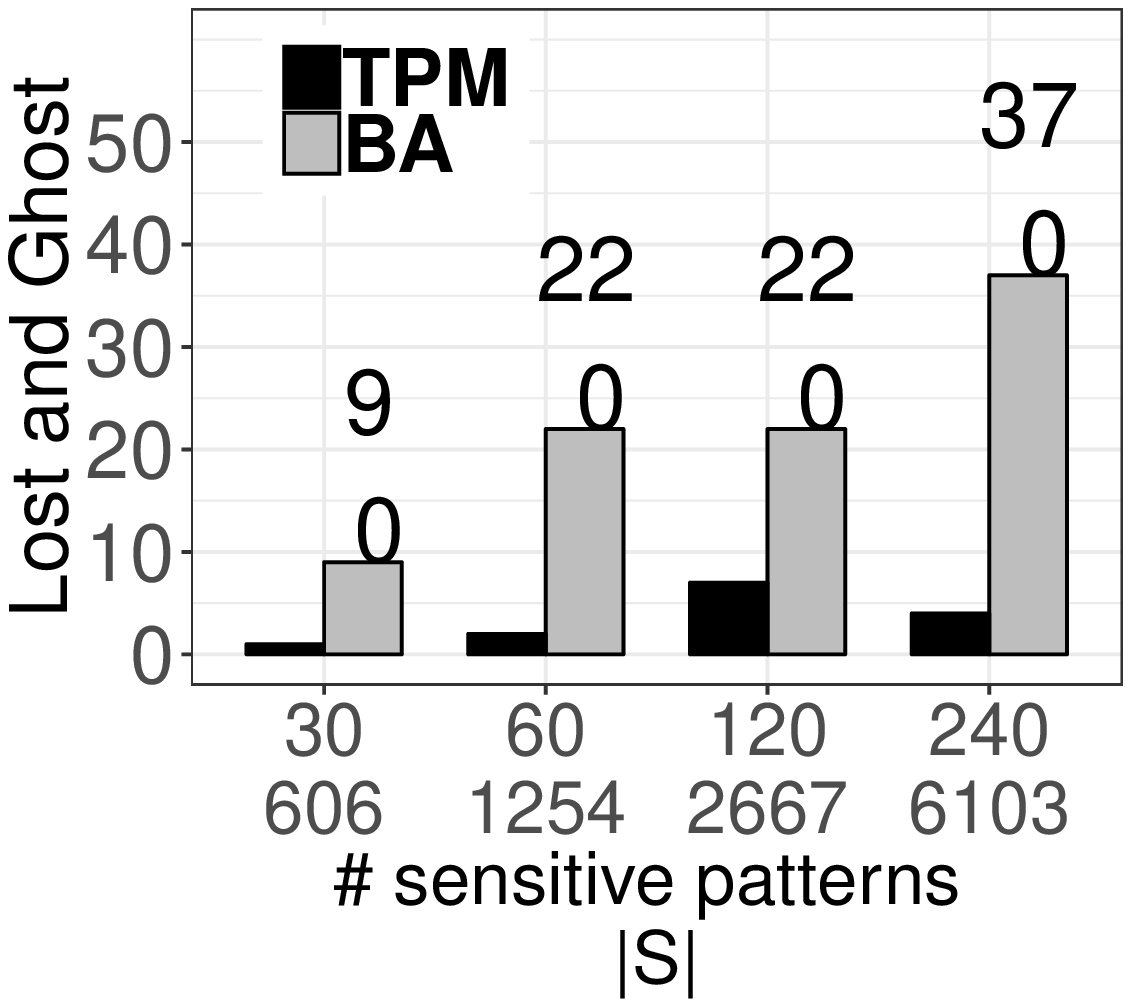}
       \vspace{-7mm}
        \caption{\OLDEN}
        \label{utility_lostghost_sensa}
    \end{subfigure}
    \hspace{+3mm} 
    \begin{subfigure}[b]{0.23\textwidth}
       \includegraphics[width=1.04\textwidth]{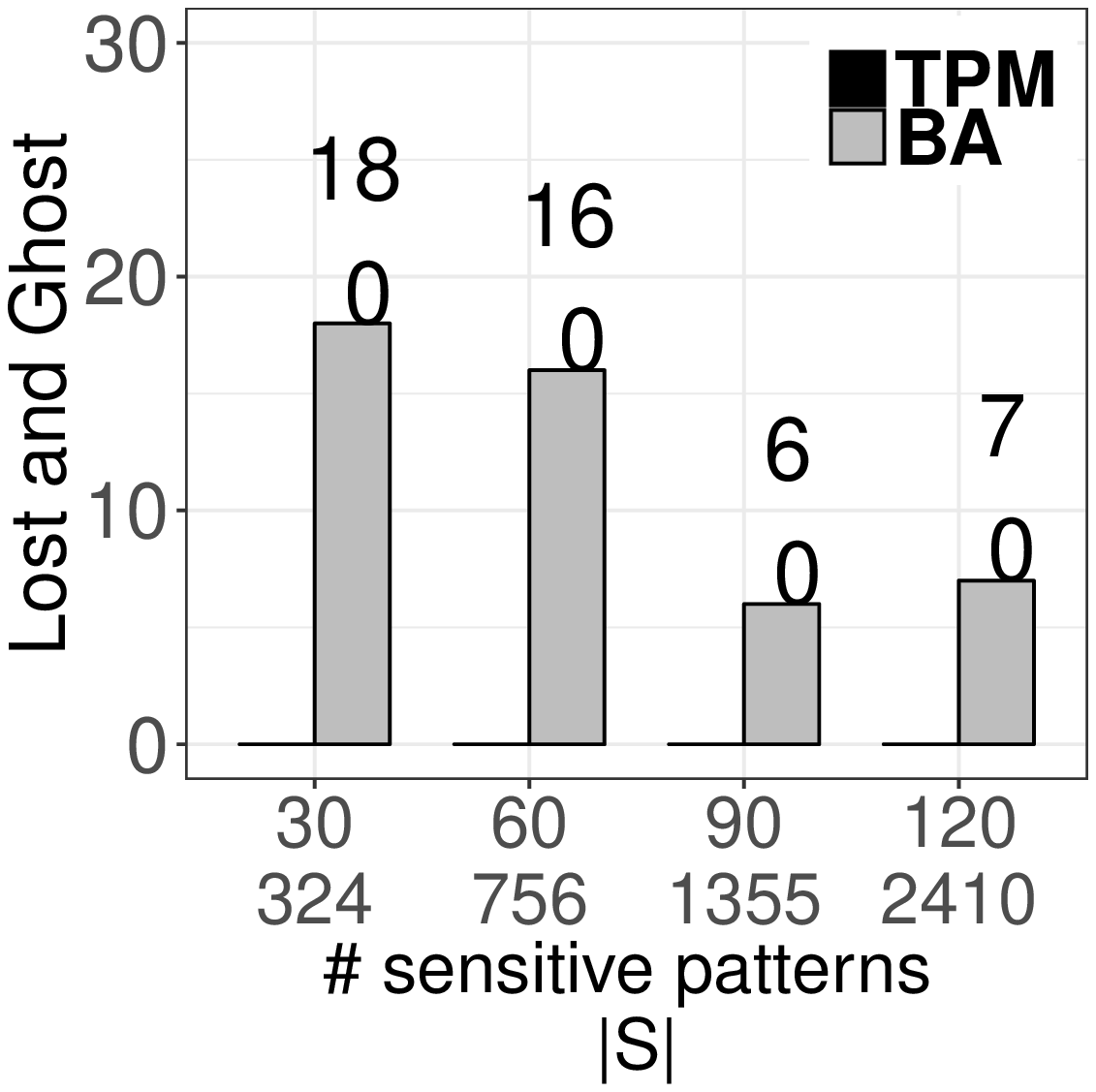}
       \vspace{-7mm}
        \caption{\TRU}
        \label{utility_lostghost_sensb}
    \end{subfigure}
    \hspace{+2mm} 
    \begin{subfigure}[b]{0.23\textwidth}
       \includegraphics[width=1.05\textwidth]{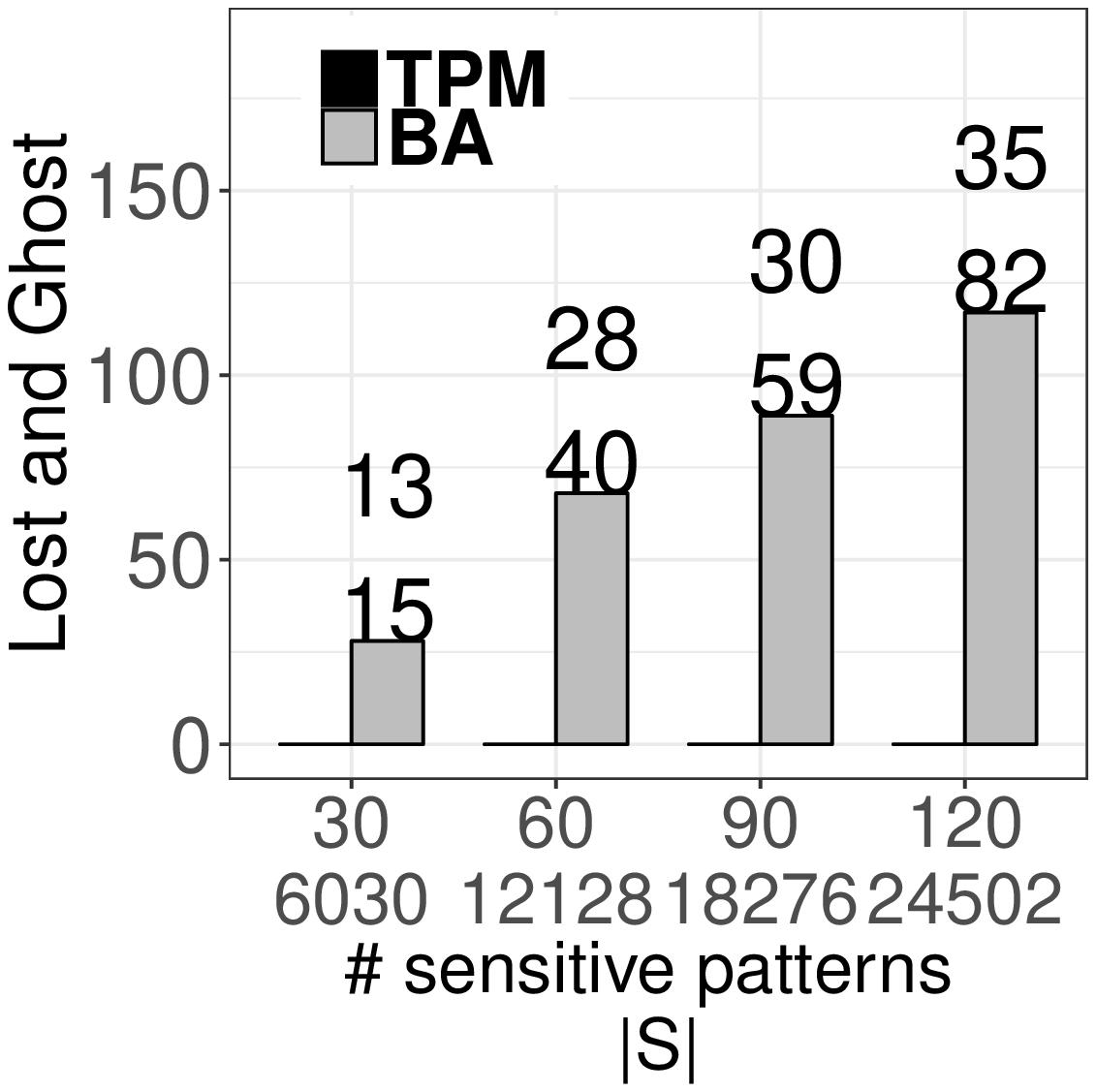}
       \vspace{-7mm}
        \caption{\MSN}
        \label{utility_lostghost_sensc}
    \end{subfigure}
    \hspace{+2mm} 
      \begin{subfigure}[b]{0.23\textwidth}
       \includegraphics[width=1.05\textwidth]{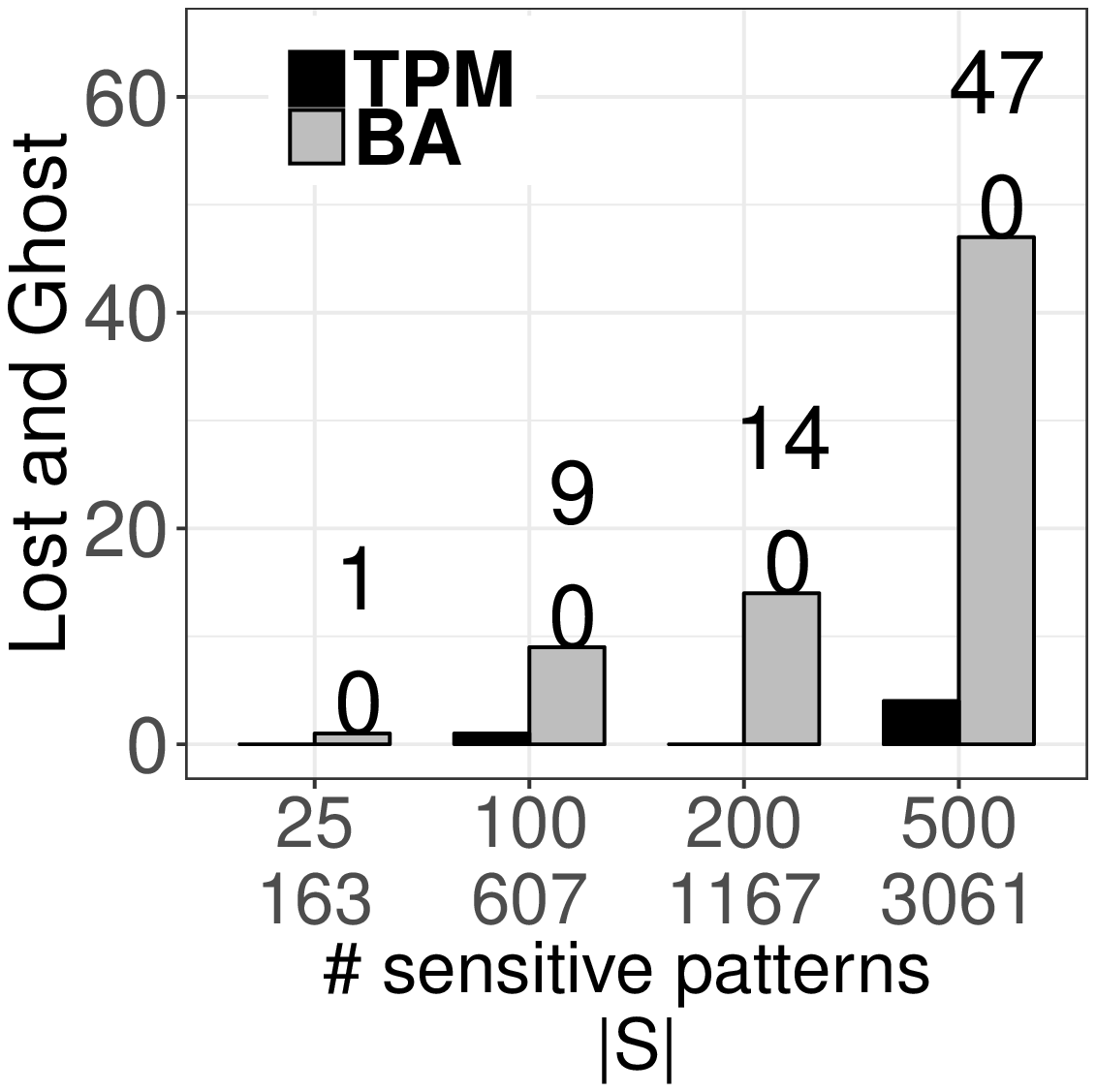}
       \vspace{-7mm}
        \caption{\DNA}
        \label{utility_lostghost_sensd}
    \end{subfigure}
\caption{Total number of $\tau$-lost and $\tau$-ghost patterns vs.~number of sensitive patterns (and $|\mathcal{S}|$). $_{y}^{x}$ on the top of each bar for BA denotes $x$ $\tau$-lost and $y$ $\tau$-ghost patterns.}\label{utility_lostghost_sens}
\end{figure}
\begin{figure*}[ht]\hspace{-0.6cm}
    \centering
    \begin{subfigure}[b]{0.23\textwidth}
       \includegraphics[width=1.13\textwidth]{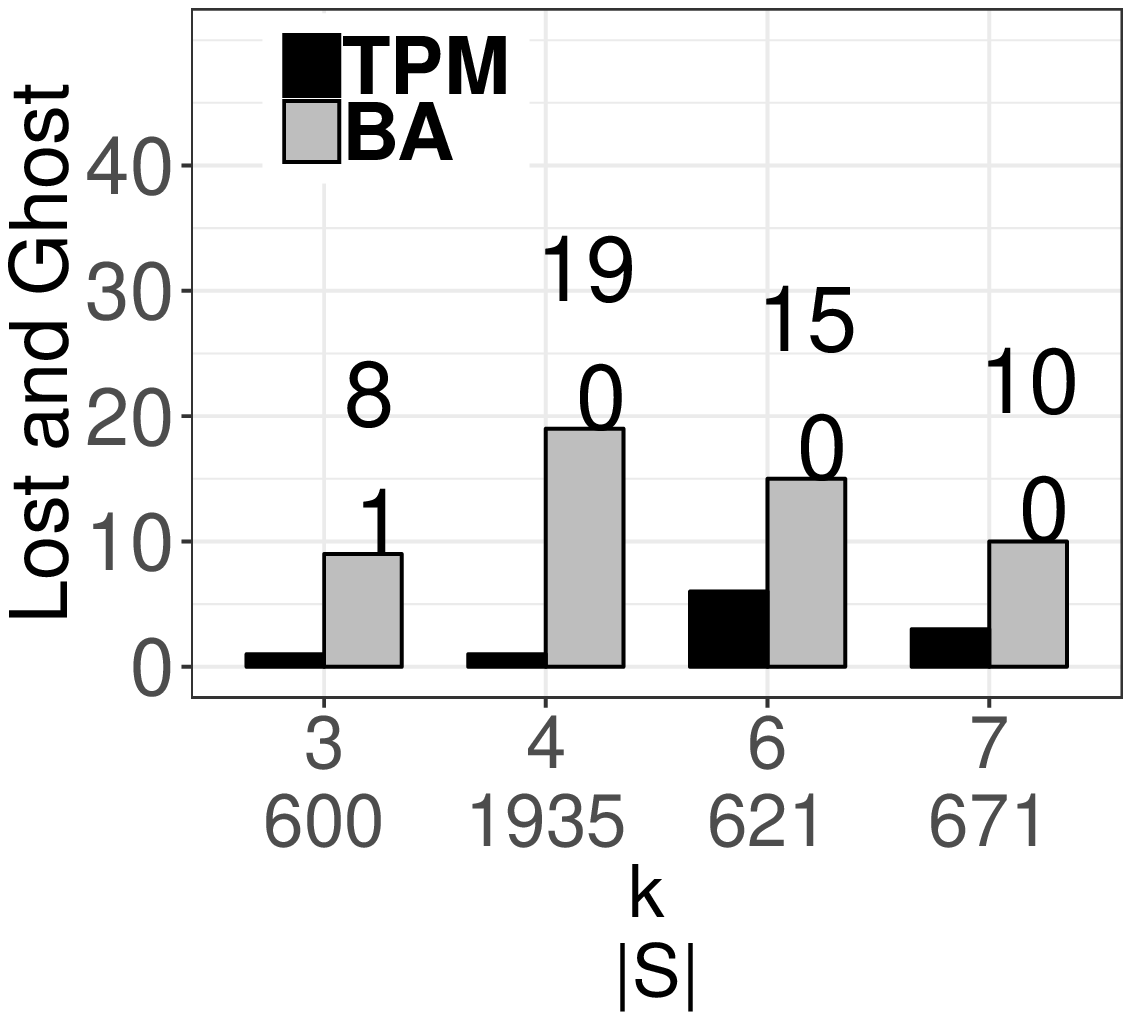}
       \vspace{-7mm}
        \caption{\OLDEN}
        \label{utility_lostghost_ka}
    \end{subfigure}
    \hspace{+2mm} 
    \begin{subfigure}[b]{0.23\textwidth}
       \includegraphics[width=1.13\textwidth]{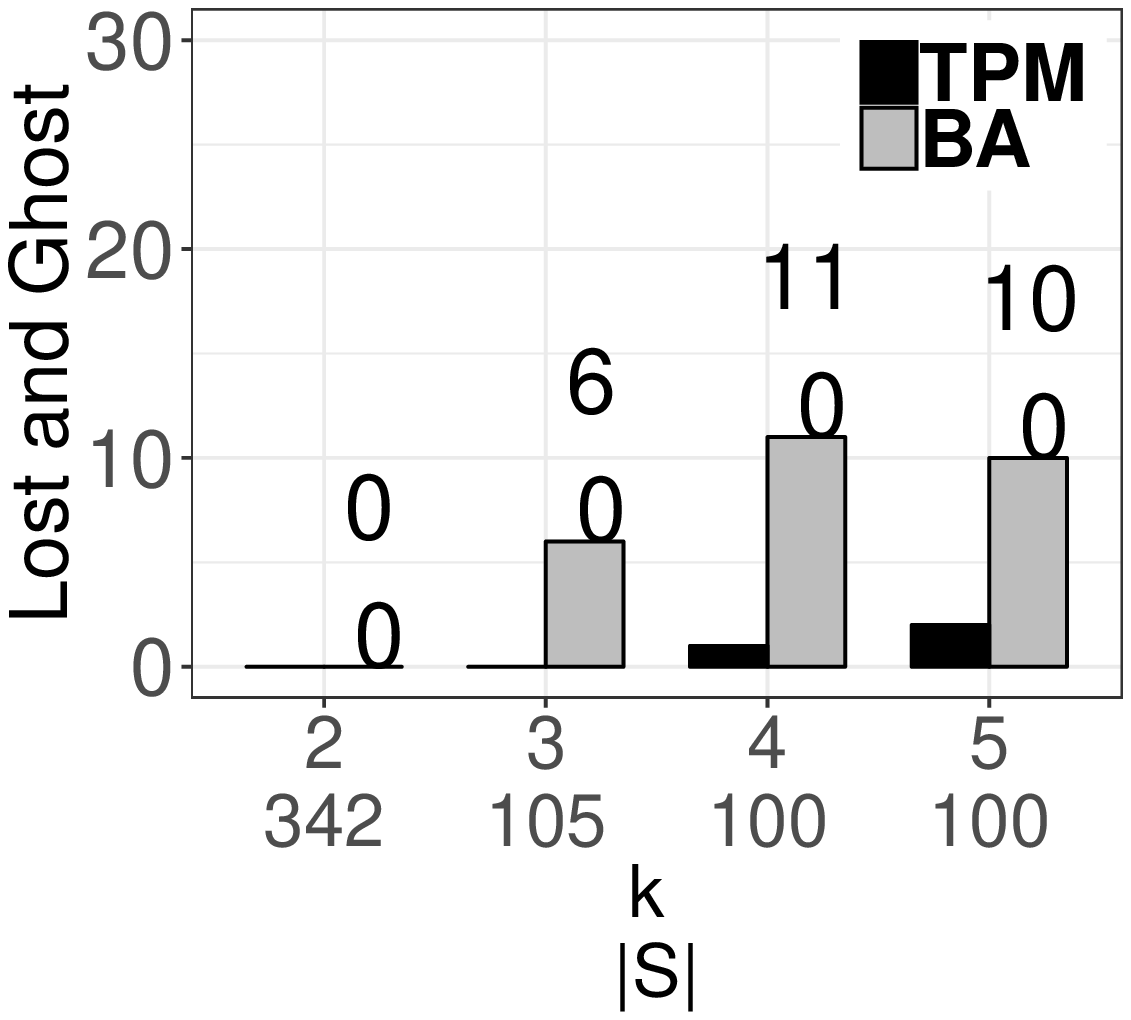}
       \vspace{-7mm}
        \caption{\TRU}
        \label{utility_lostghost_kb}
    \end{subfigure}
    \hspace{+4mm} 
    \begin{subfigure}[b]{0.23\textwidth}
       \includegraphics[width=1.13\textwidth]{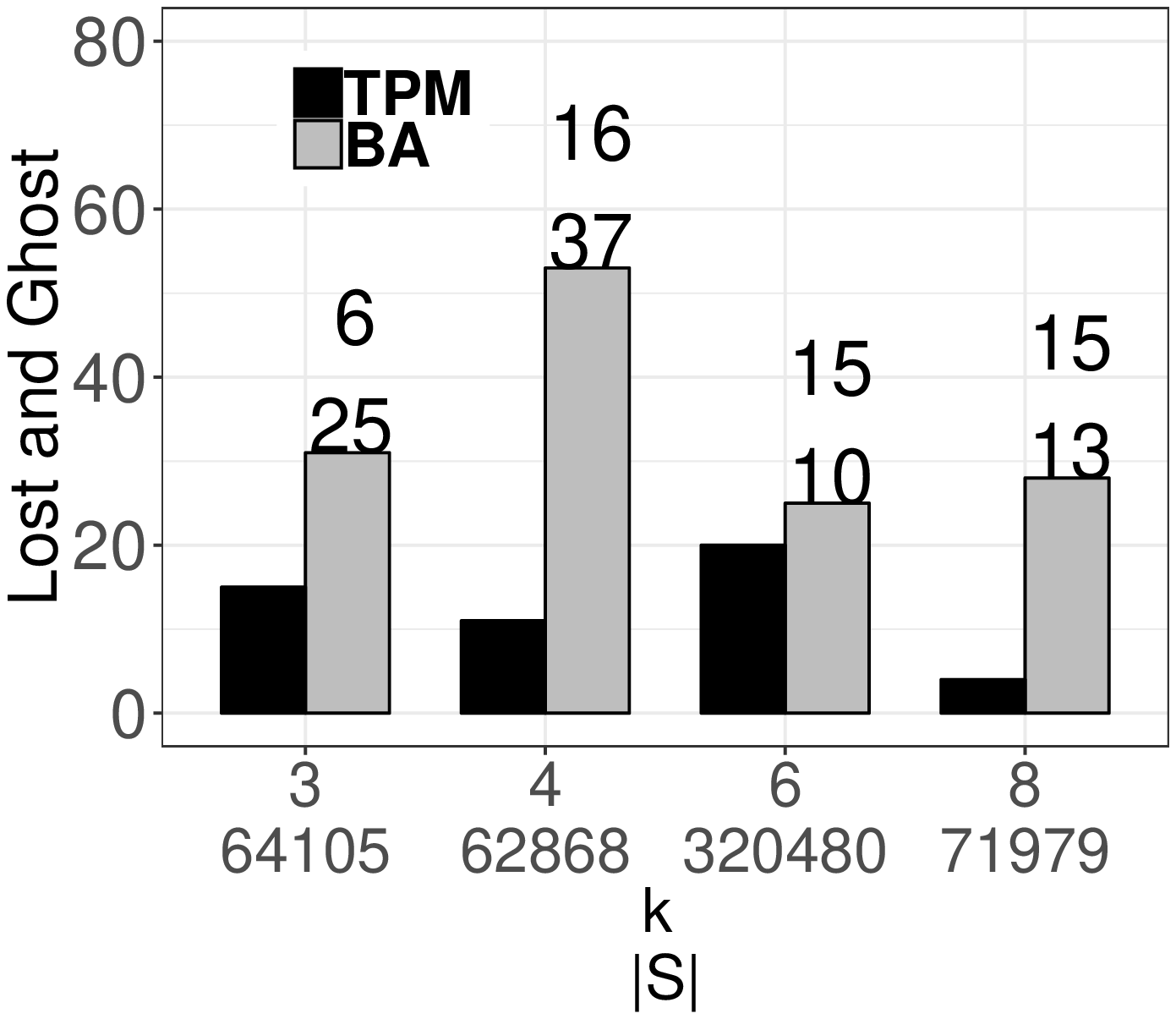}
       \vspace{-5mm}
        \caption{\MSN}
        \label{utility_lostghost_kc}
    \end{subfigure}
    \hspace{+3.5mm} 
      \begin{subfigure}[b]{0.23\textwidth}
       \includegraphics[width=1.15\textwidth]{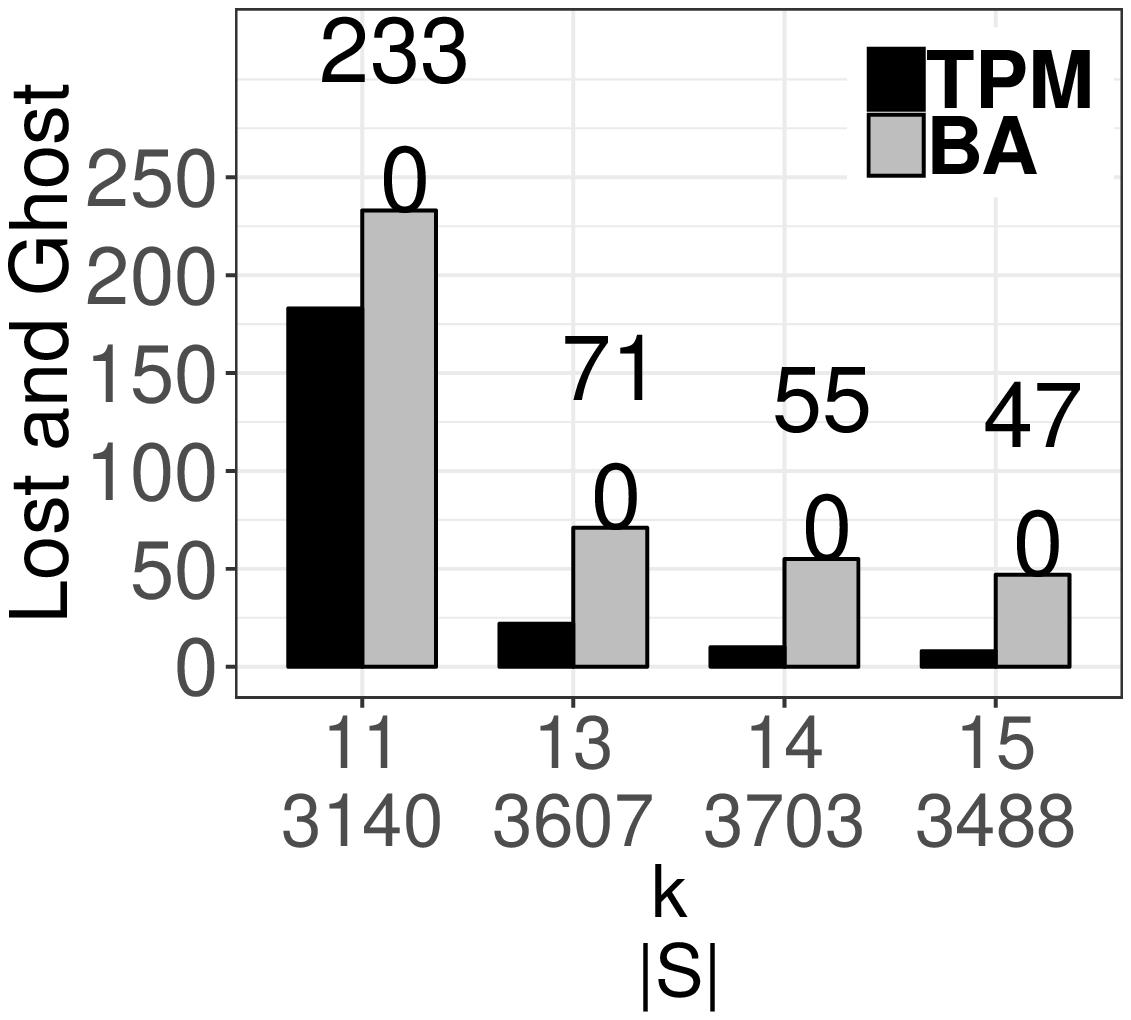}
       \vspace{-7mm}
        \caption{\DNA}
        \label{utility_lostghost_kd}
    \end{subfigure}
\caption{Total number of $\tau$-lost and $\tau$-ghost patterns vs.~length of sensitive patterns $k$ (and $|\mathcal{S}|$). $_{y}^{x}$ on the top of each bar for BA denotes $x$ $\tau$-lost and $y$ $\tau$-ghost patterns.}\label{utility_lostghost_k} \end{figure*}

We also demonstrate that {\PFSA} reduces the length of the output string $X$ of {\TFSA} substantially, creating a string $Y$ that contains \emph{less redundant information} and allows for more efficient analysis. Fig. \ref{utility_shortena} shows the length of  $X$ and of $Y$ and their difference for $k=5$. $Y$ was much shorter than $X$ and its length decreased with the number of sensitive patterns, since more substrings had a suffix-prefix overlap of length $k-1=4$ and were removed (see Section \ref{sec:ps}). Interestingly, the length of $Y$ was close to that of $W$ (the string before  sanitization). A larger $k$ led to less substantial length reduction as shown in Fig. \ref{utility_shortenb} (but still few thousand letters were removed), since it is less likely for long substrings of sensitive patterns to have an overlap and be removed.

\begin{figure}\hspace{-0.55cm}
    \centering
    \begin{subfigure}[b]{0.22\textwidth}
       \includegraphics[width=1.15\textwidth]{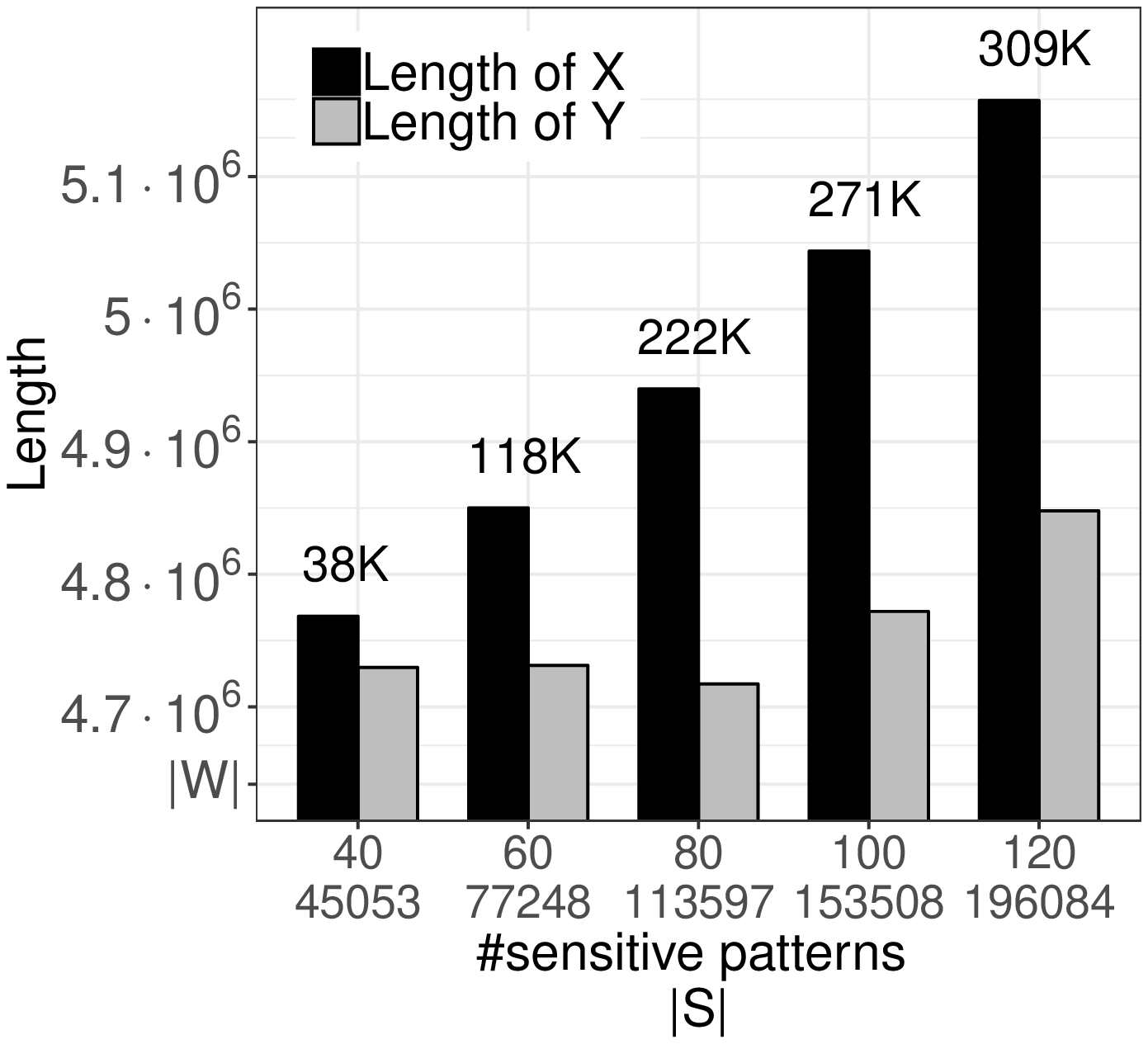}
       \vspace{-7mm}
        \caption{\DNA}
        \label{utility_shortena}
    \end{subfigure}
    \hspace{+2mm} 
    \begin{subfigure}[b]{0.22\textwidth}
       \includegraphics[width=1.15\textwidth]{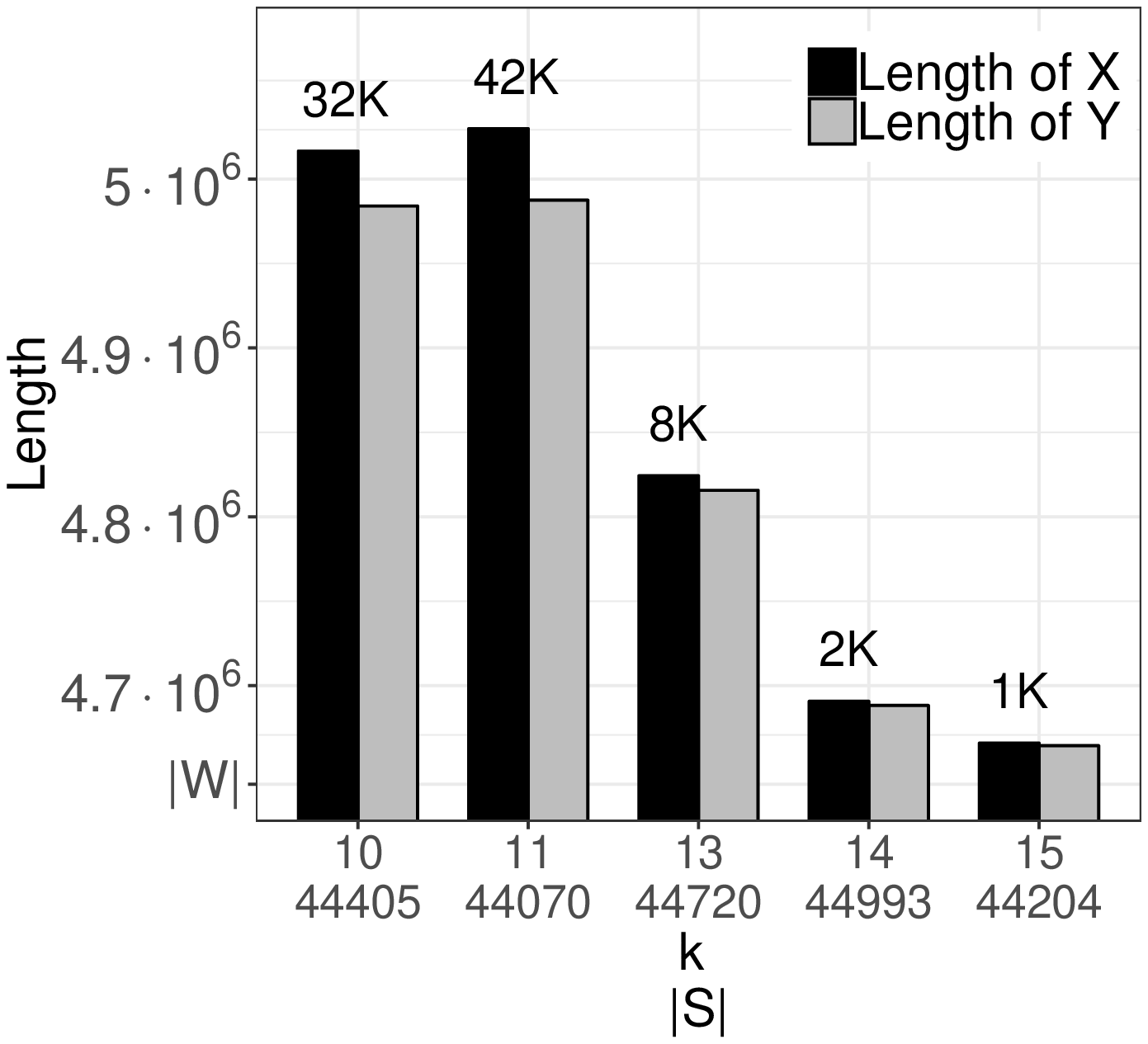}
       \vspace{-7mm}
        \caption{\DNA}
        \label{utility_shortenb}
    \end{subfigure}
    \hspace{+2.5mm} 
    \begin{subfigure}[b]{0.23\textwidth}
       \includegraphics[width=1.08\textwidth]{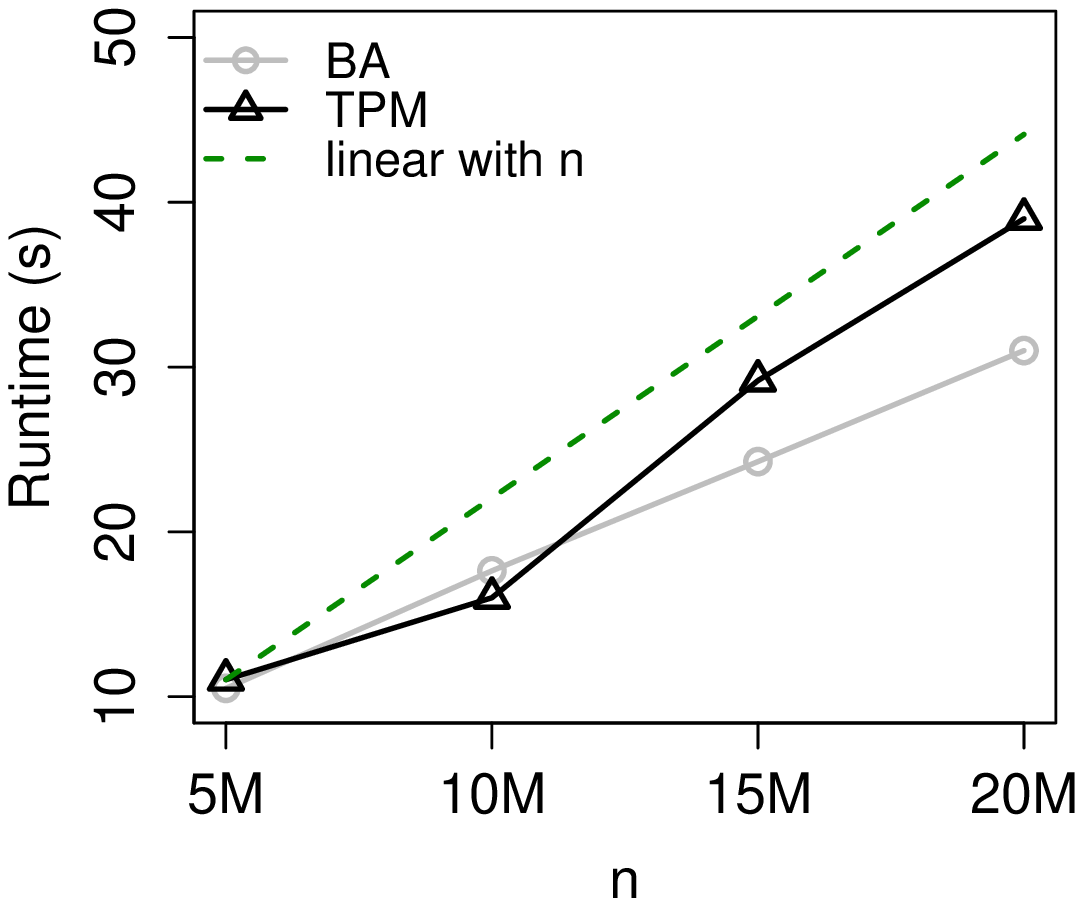}
       \vspace{-5mm}
        \caption{Substr. of {\SYN}}
        \label{runtime_n}
    \end{subfigure}
    \hspace{+0.1mm} 
      \begin{subfigure}[b]{0.23\textwidth}
       \includegraphics[width=1.08\textwidth]{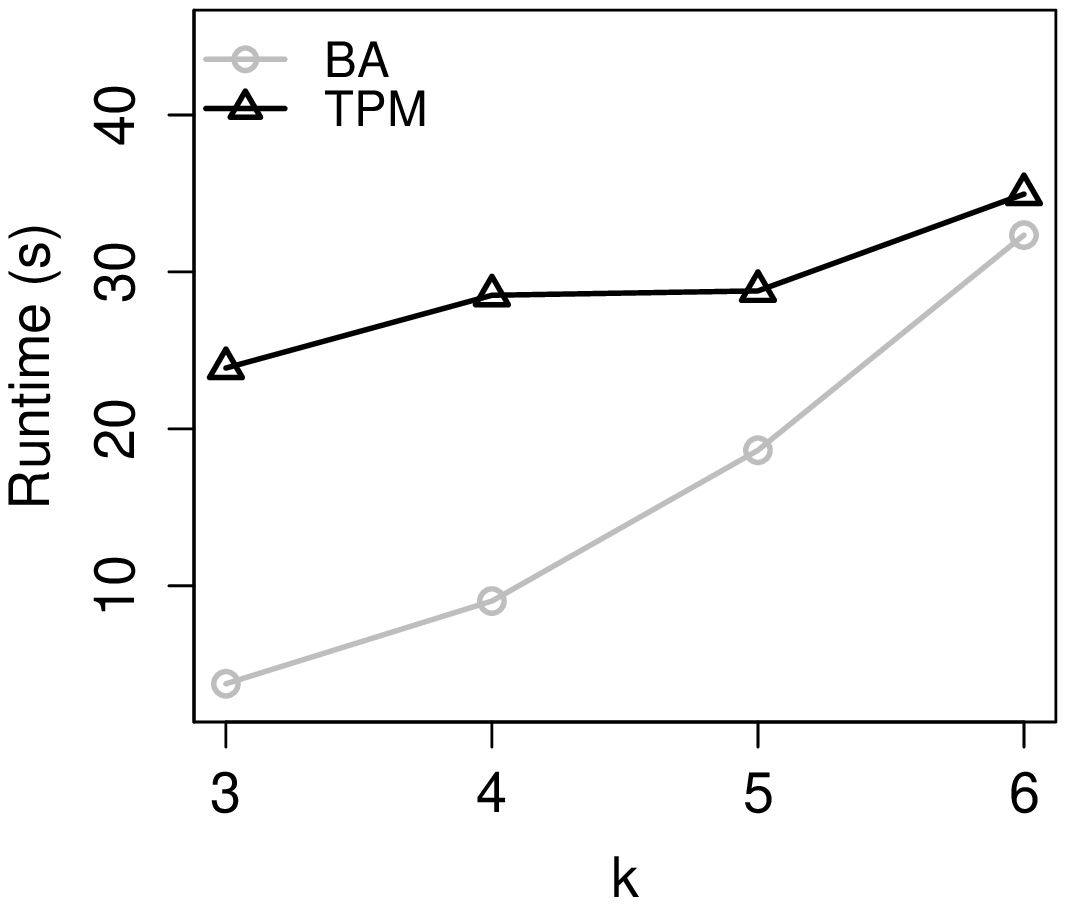}
       \vspace{-5.5mm}
        \caption{\SYN}
        \label{runtime_k}
    \end{subfigure}
\caption{Length of $X$ and $Y$ (output of {\TFSA} and {\PFSA}, resp.) for varying: (a) number of sensitive patterns (and  $|\mathcal{S}|$), (b) length of sensitive patterns $k$ (and $|\mathcal{S}|$). On the top of each pair of bars we plot $|X|-|Y|$. Runtime on synthetic data for varying: (c) length $n$ of string and (d) length $k$ of sensitive patterns. Note that $|Y|=|Z|$.} \label{shorten_runtime}
\end{figure}

\paragraph{Efficiency} We finally measured the runtime of {\TPM} using
prefixes of the synthetic string {\SYN} whose length $n$ is $20$ million letters. Fig. \ref{runtime_n} (resp., Fig. \ref{runtime_k}) shows that {\TPM} scaled linearly with $n$ (resp., $k$), as predicted by our analysis in Section \ref{sec:z} ({\TPM} takes $\cO(n+|Y|+k\delta\sigma+\delta\sigma\theta)=\cO(kn+k\delta\sigma+\delta\sigma\theta)$ time, since the algorithm of \cite{Pissinger} was used for MCK instances). In addition, {\TPM} is efficient, with a runtime similar to that of {\BA} and less than $40$ seconds for {\SYN}.

\subsection{{\TM} vs.~{\TMI}}
We compare {\TM} with {\TMI} based on data utility and the number of implausible patterns incurred. The objective of these experiments is to show that {\TMI} is able to produce a string $Z$ that does not contain implausible patterns, while being comparable to {\TM} in terms of the amount of distortion and number of ghost patterns incurred.

We do not report the results of comparing {\TM} with {\TMI} in terms of efficiency, because the runtime of {\TMI} was almost identical to that of {\TM}.

\paragraph{Impact of $|\mathcal{S}|$} We first demonstrate that many implausible patterns may occur as a result of replacing $\#$s with letters, when {\MCSR} is used. This can be seen from Figs. \ref{TM_TMI_vs_s_unprotect_old}, \ref{TM_TMI_vs_s_unprotect_tru}, and \ref{TM_TMI_vs_s_unprotect_msn}, which show the percentage of implausible patterns incurred by {\TM}, for varying $|\mathcal{S}|$ in {\OLDEN}, {\TRU}, and {\MSN}, respectively.  The percentage is on average $33.08\%$ (and up to $35.63\%$). The percentage for {\DNA} is $0\%$ (omitted), because this dataset has a  very small alphabet size. Thus,  in this experiment, {\MCSRA} and {\MCSRAI} are essentially the same algorithm.  Since {\TMI} is guaranteed to eliminate  implausible patterns, its corresponding percentages are zero (omitted).

\begin{figure}[!ht]\hspace{-5mm}
    \centering
    \hspace{+2mm}
    \begin{subfigure}[b]{0.23\textwidth}
      \includegraphics[width=1.01\textwidth]{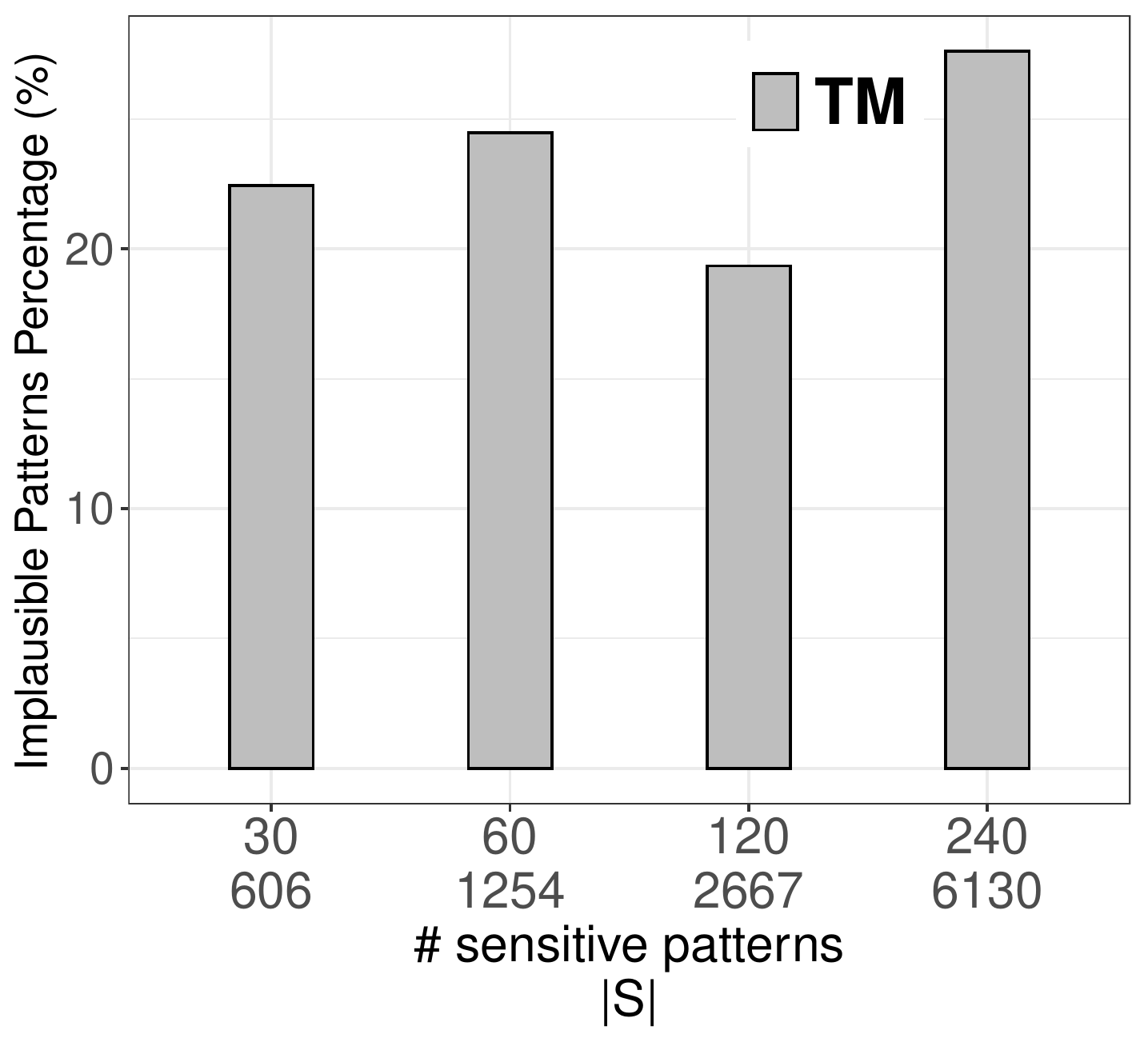}
      \vspace{-7mm}
        \caption{\OLDEN}
        \label{TM_TMI_vs_s_unprotect_old}
    \end{subfigure}
    \hspace{+2mm}
    \begin{subfigure}[b]{0.23\textwidth}
      \includegraphics[width=1.01\textwidth]{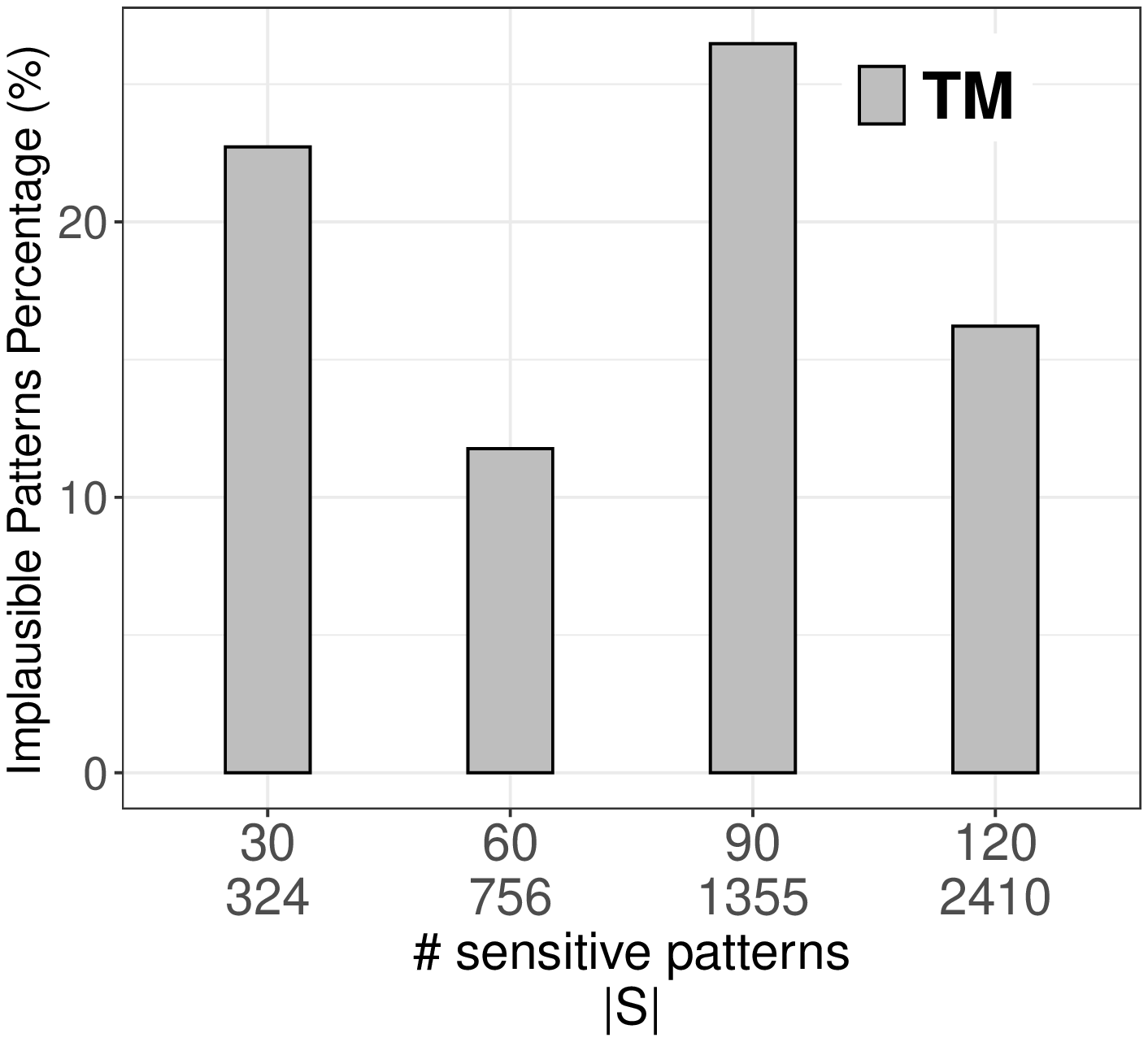}
      \vspace{-7mm}
        \caption{\TRU{}}
        \label{TM_TMI_vs_s_unprotect_tru}
    \end{subfigure}
      \hspace{+2mm}
    \begin{subfigure}[b]{0.23\textwidth}
      \includegraphics[width=1.01\textwidth]{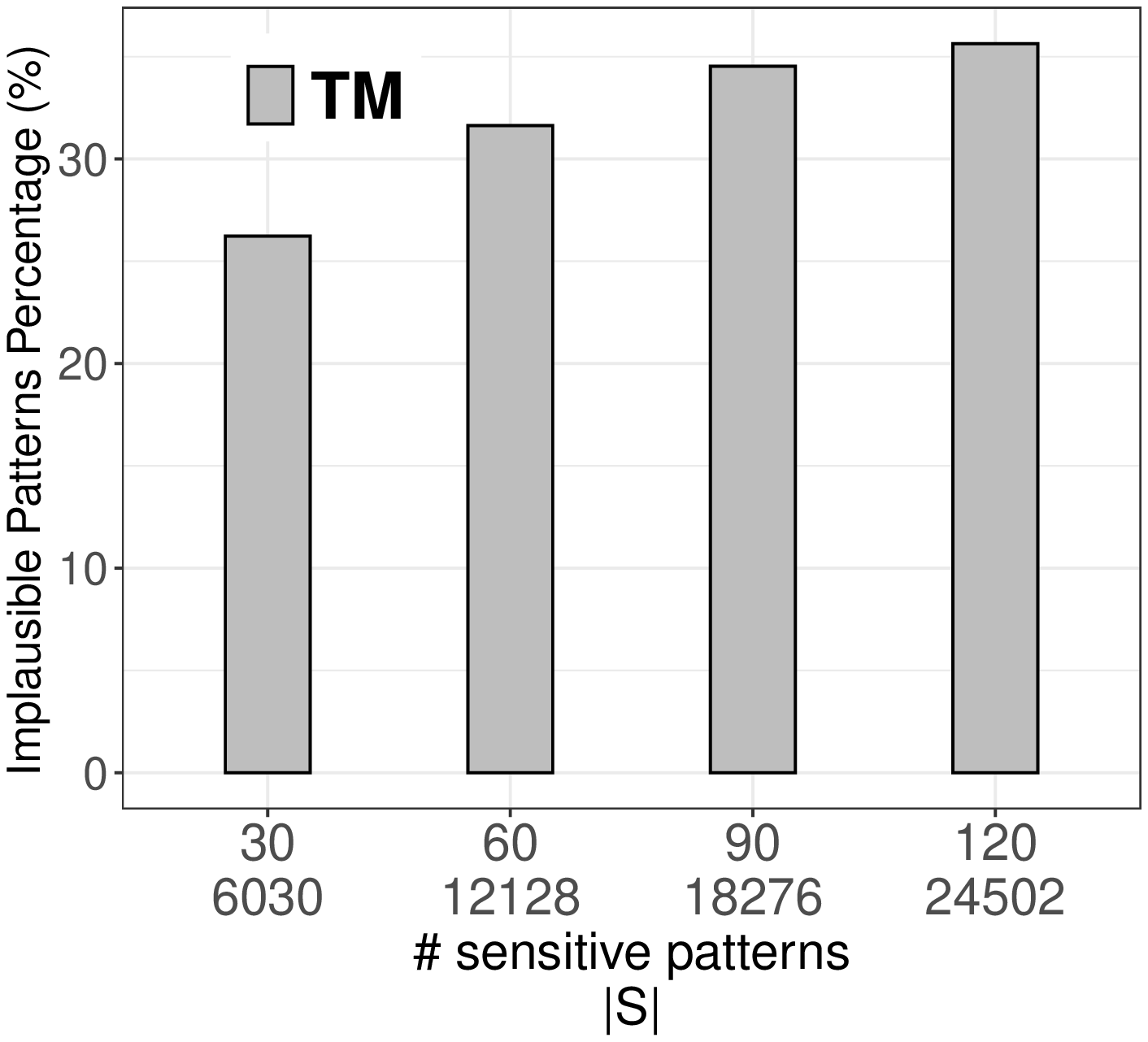}
      \vspace{-7mm}
        \caption{\MSN}
        \label{TM_TMI_vs_s_unprotect_msn}
    \end{subfigure}
    \ 

\caption{Percentage of implausible patterns vs.~number of sensitive patterns (and $|\mathcal{S}|$). The percentages of implausible patterns for {\DNA} are all $0\%$. }\label{TM_TMI_vs_s_unprotect}
\end{figure}

We then demonstrate that {\TMI} eliminates implausible patterns without incurring substantial  utility loss compared to {\TM}. Figs. \ref{TM_TMI_vs_s_disto} and \ref{TM_TMI_vs_s_ghost} show that {\TMI} incurred a comparable amount of distortion to {\TM}. Specifically, {\TMI} incurred $8\%$ and $1\%$ less distortion in the case of {\OLDEN} and {\TRU} datasets and $37\%$ more distortion in the case of {\MSN}.
{\TMI} also incurred a similar number of ghosts than {\TM}. Specifically, {\TMI} incurred $7.1\%$ fewer ghosts in the case of {\TRU} and $54\%$ more ghosts in the case of {\MSN}. Note that no $\tau$-ghost patterns were incurred in the case of {\OLDEN} (for both {\TM} and {\TMI}). The worse performance of {\TMI} in the case of the {\MSN} dataset is attributed to its relatively small alphabet size, which makes it more difficult to select a letter replacement that does not incur implausible patterns.

\begin{figure}[!ht]\hspace{-5mm}
    \centering
    \begin{subfigure}[b]{0.23\textwidth}
      \includegraphics[width=1.05\textwidth]{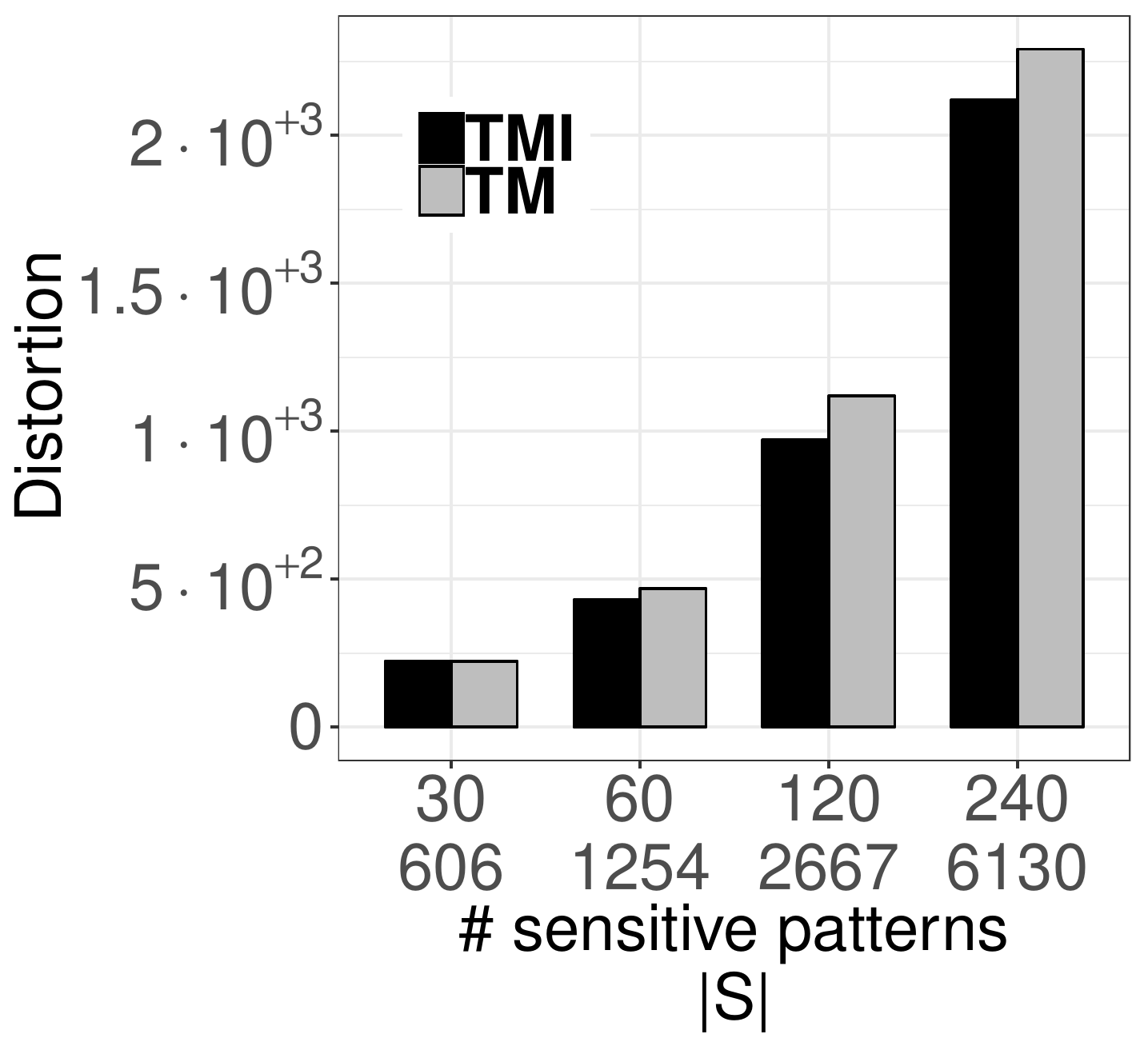}
      \vspace{-7mm}
        \caption{\OLDEN}
        \label{TM_TMI_vs_s_disto_old}
    \end{subfigure}
    \hspace{+3mm} 
     \begin{subfigure}[b]{0.23\textwidth}
      \includegraphics[width=1.05\textwidth]{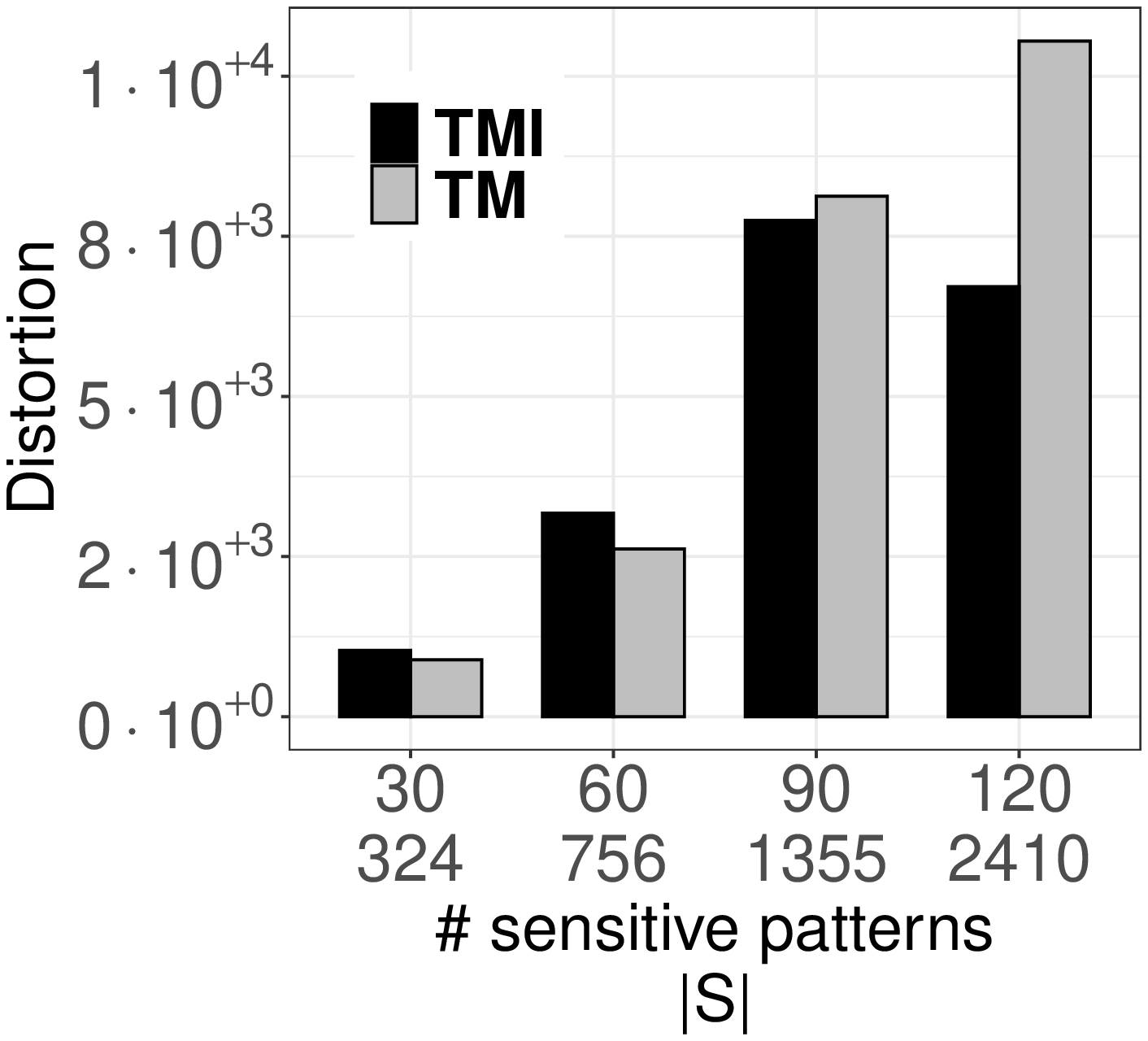}
      \vspace{-7mm}
        \caption{\TRU{}}
        \label{TM_TMI_vs_s_disto_tru}
    \end{subfigure}
     \hspace{+2mm} 
     \begin{subfigure}[b]{0.22\textwidth}
      \includegraphics[width=1.09\textwidth]{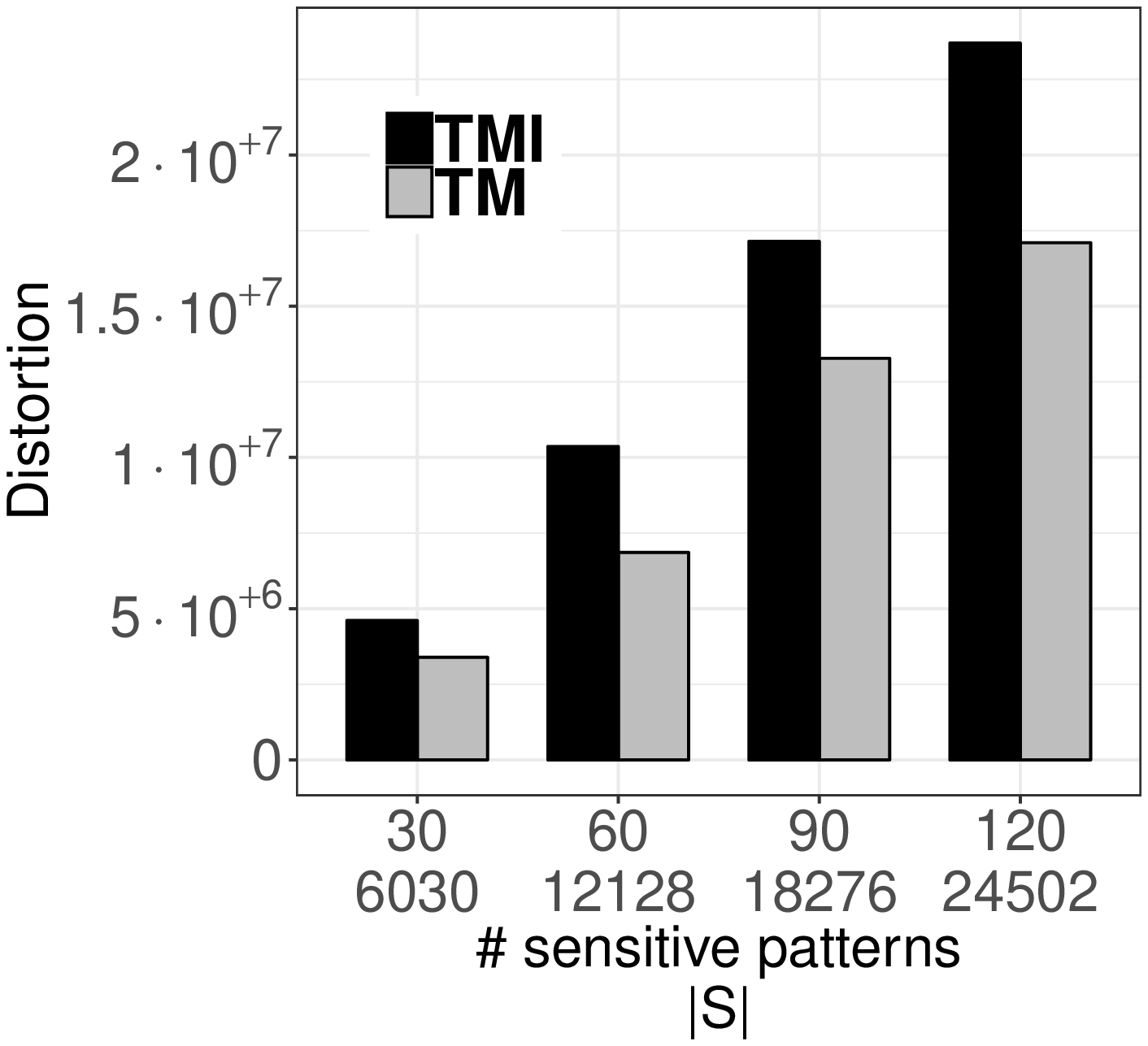}
      \vspace{-7mm}
        \caption{\MSN}
        \label{TM_TMI_vs_s_disto_msn}
    \end{subfigure}
    \hspace{+3mm} 
\caption{Distortion vs.~number of sensitive patterns and their total number $|\mathcal{S}|$ of occurrences in $W$ ~(first two lines on the $X$ axis).}\label{TM_TMI_vs_s_disto}
\end{figure}

\begin{figure}[!ht]\hspace{-5mm}
    \centering
    \begin{subfigure}[b]{0.23\textwidth}
      \includegraphics[width=1.04\textwidth]{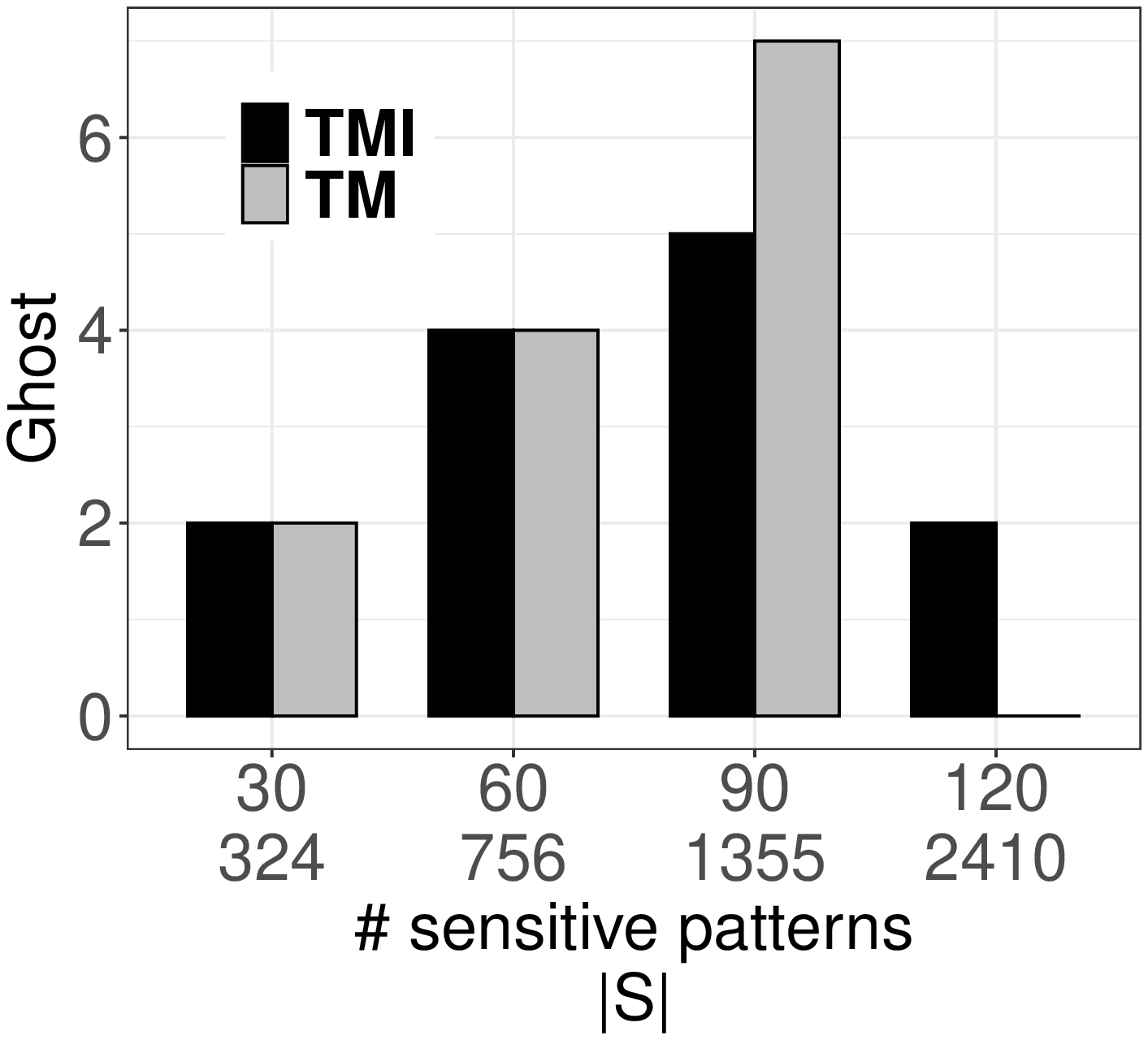}
      \vspace{-7mm}
        \caption{\TRU{}}
        \label{TM_TMI_vs_s_ghost_tru}
    \end{subfigure}
    \hspace{+3mm}
    \begin{subfigure}[b]{0.23\textwidth}
      \includegraphics[width=1.04\textwidth]{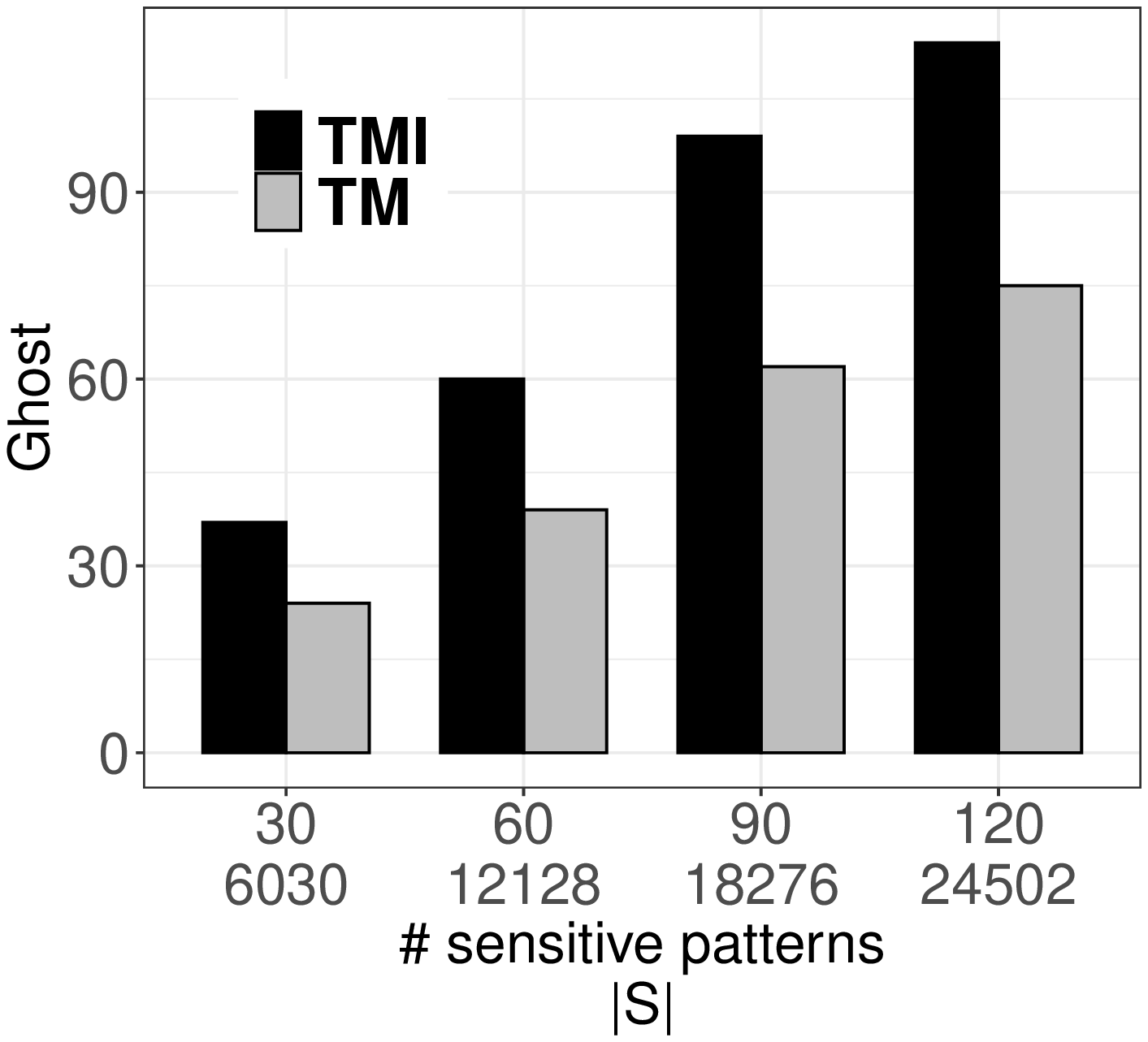}
      \vspace{-7mm}
        \caption{\MSN}
        \label{TM_TMI_vs_s_ghost_msn}
    \end{subfigure}
\caption{Number of $\tau$-ghost patterns (the number of $\tau$-lost patterns is zero by design) vs.~number of sensitive patterns (and $|\mathcal{S}|$). The number of $\tau$-ghost patterns for {\OLDEN} is  $0$.}\label{TM_TMI_vs_s_ghost}
\end{figure}

\begin{figure}[!ht]\hspace{-5mm}
    \centering
     \begin{subfigure}[b]{0.235\textwidth}
      \includegraphics[width=1.01\textwidth]{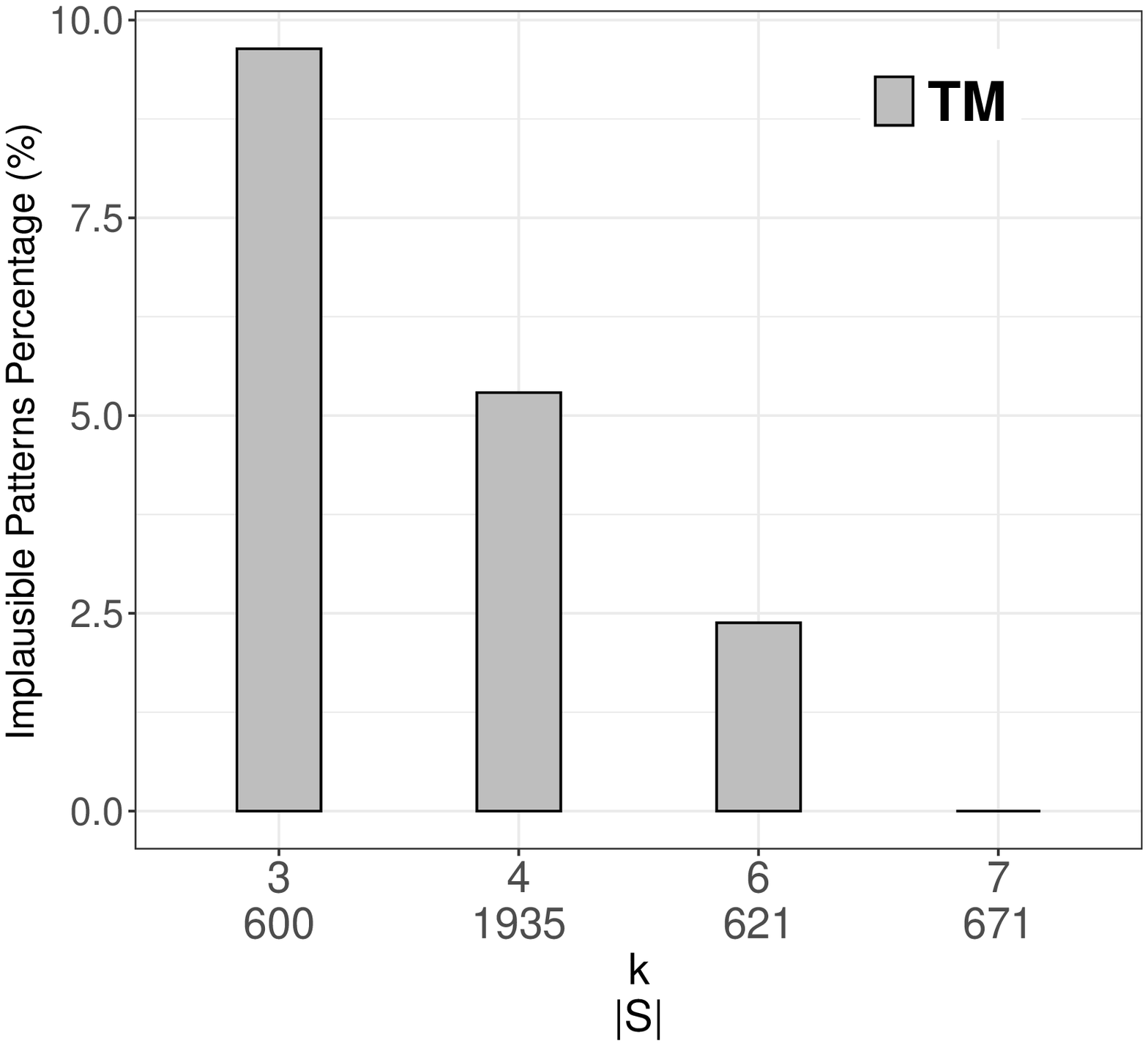}
      \vspace{-7mm}
        \caption{\OLDEN}
        \label{TM_TMI_vs_k_unprotect_old}
    \end{subfigure}
     \hspace{+2mm} 
    \begin{subfigure}[b]{0.23\textwidth}
      \includegraphics[width=1.03\textwidth]{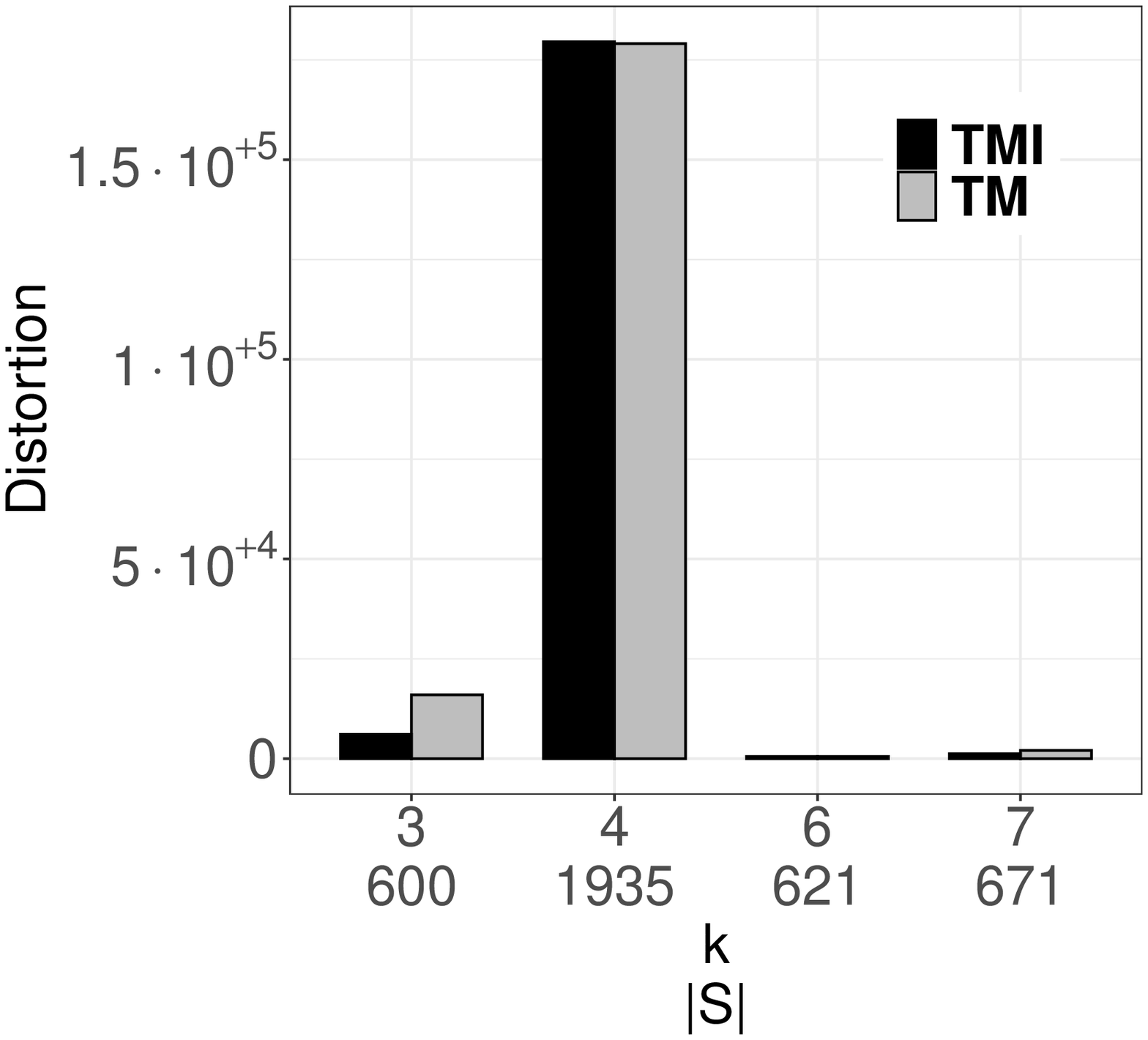}
      \vspace{-7mm}
        \caption{\OLDEN}
        \label{TM_TMI_vs_k_disto_old}
    \end{subfigure}
\hspace{+2mm}
\begin{subfigure}[b]{0.23\textwidth}
      \includegraphics[width=1.05\textwidth]{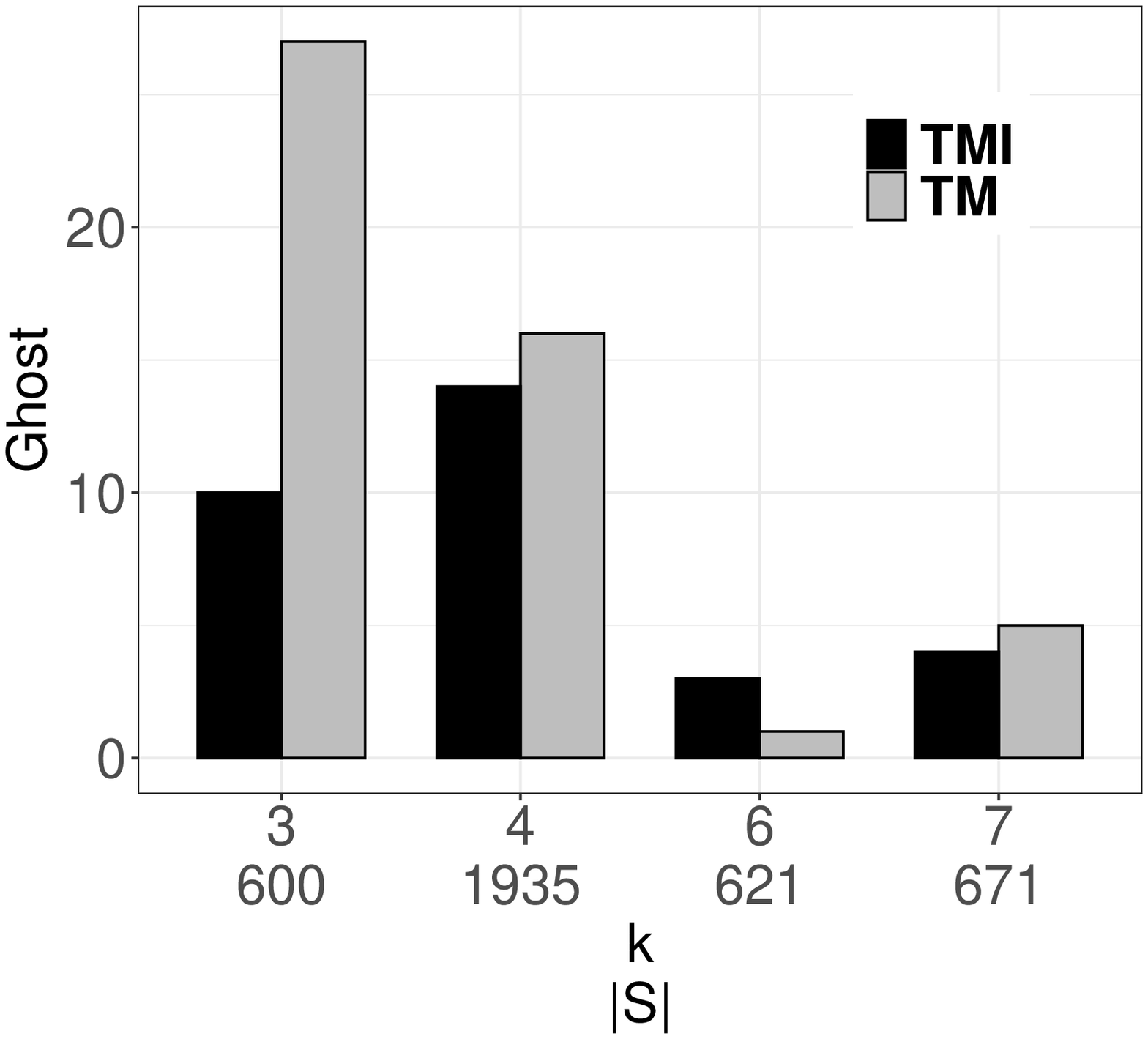}
      \vspace{-7mm}
        \caption{\OLDEN}
        \label{TM_TMI_vs_k_ghost_old}
    \end{subfigure}

\caption{(a) Percentage of implausible patterns vs.~$k$ (and $|\mathcal{S}|$). (b) Distortion vs.~$k$ (and $|\mathcal{S}|$). (c) Number of $\tau$-ghost patterns vs.~$k$ (and $|\mathcal{S}|$).}\label{TM_TMI_vs_s_unprotect}
\end{figure}

\paragraph{Impact of $k$} Fig. \ref{TM_TMI_vs_k_unprotect_old} shows that  the percentage of implausible patterns incurred by {\TM} for the {\OLDEN} dataset was on average $4.3\%$ (and up to $9.6\%$).  Again, this confirms the need to eliminate implausible patterns in practice. The results for {\TRU}, {\MSN}, and {\DNA} are qualitatively similar and omitted from all  remaining experiments.

We now demonstrate that {\TMI} eliminates implausible patterns,  while incurring a comparable amount of distortion and ghosts (on average) compared to {\TM}.  Specifically, the distortion for  {\TMI} was $17\%$ lower than {\TM} on average (see  Fig. \ref{TM_TMI_vs_k_disto_old}), and the number of $\tau$-ghost patterns for {\TMI} was $16.2\%$ lower on average (see Fig. \ref{TM_TMI_vs_k_ghost_old}).

\paragraph{Impact of $\rho$} We demonstrate that {\TMI} can eliminate implausible patterns, while preserving data utility as well as {\TM} does. This can be seen from Fig.~\ref{TM_TMI_vs_rho_unprotect_olden}, which  shows that the percentage of implausible patterns incurred by {\TM} was $4.1\%$ on average (and up to $5.3\%$), and from
Figs.~\ref{TM_TMI_vs_rho_disto_olden} and \ref{TM_TMI_vs_rho_ghost_olden}, which show that {\TMI} caused on average $19.5\%$ lower distortion and $9.4\%$ fewer $\tau$-ghosts, respectively, compared to {\TM}.

\begin{figure}[!ht]
    \centering
     \begin{subfigure}[b]{0.235\textwidth}
      \includegraphics[width=1.03\textwidth]{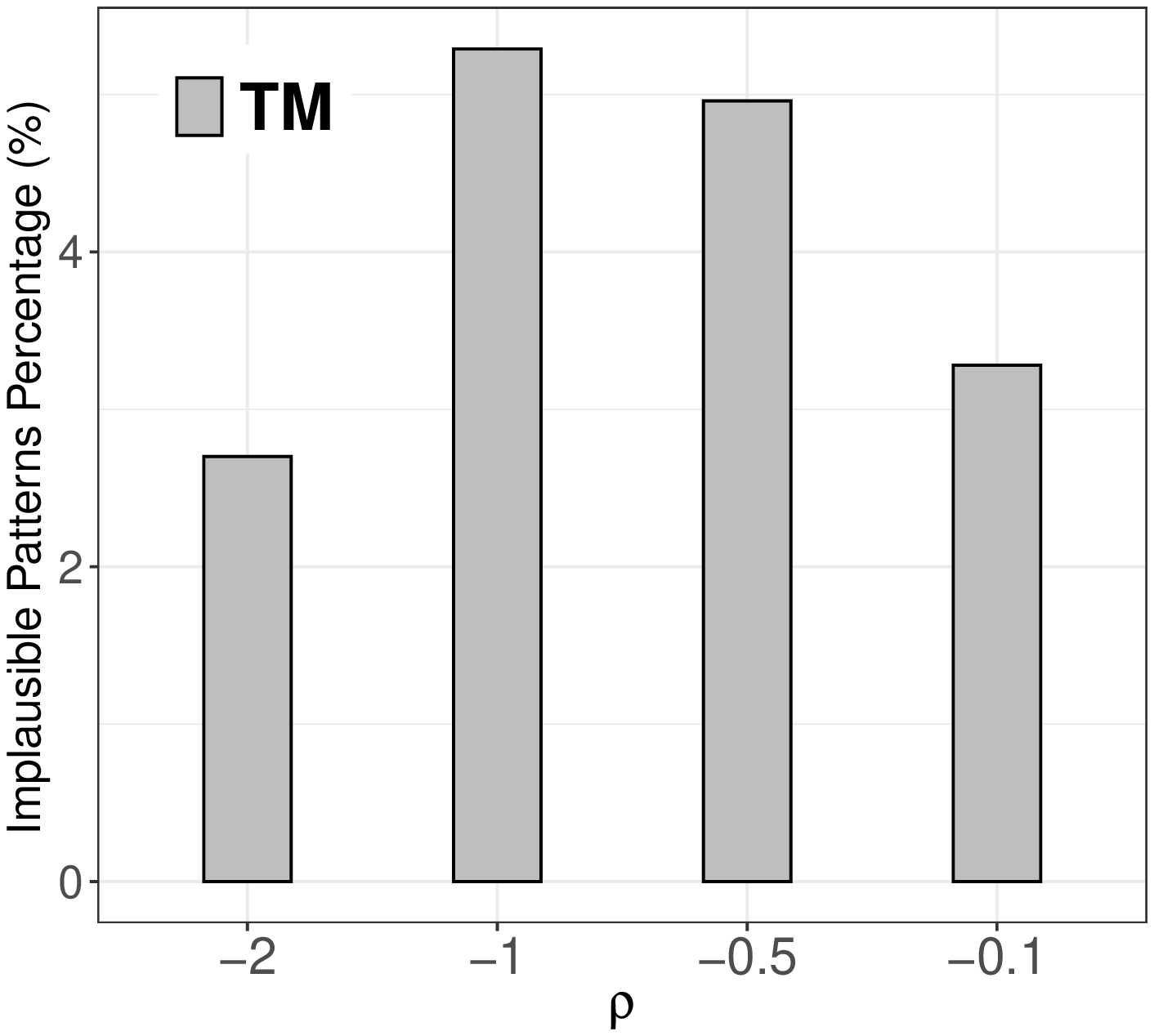}
      \vspace{-7mm}
        \caption{\OLDEN}
        \label{TM_TMI_vs_rho_unprotect_olden}
    \end{subfigure}
    \hspace{+3mm} 
   \begin{subfigure}[b]{0.23\textwidth}
      \includegraphics[width=1.05\textwidth]{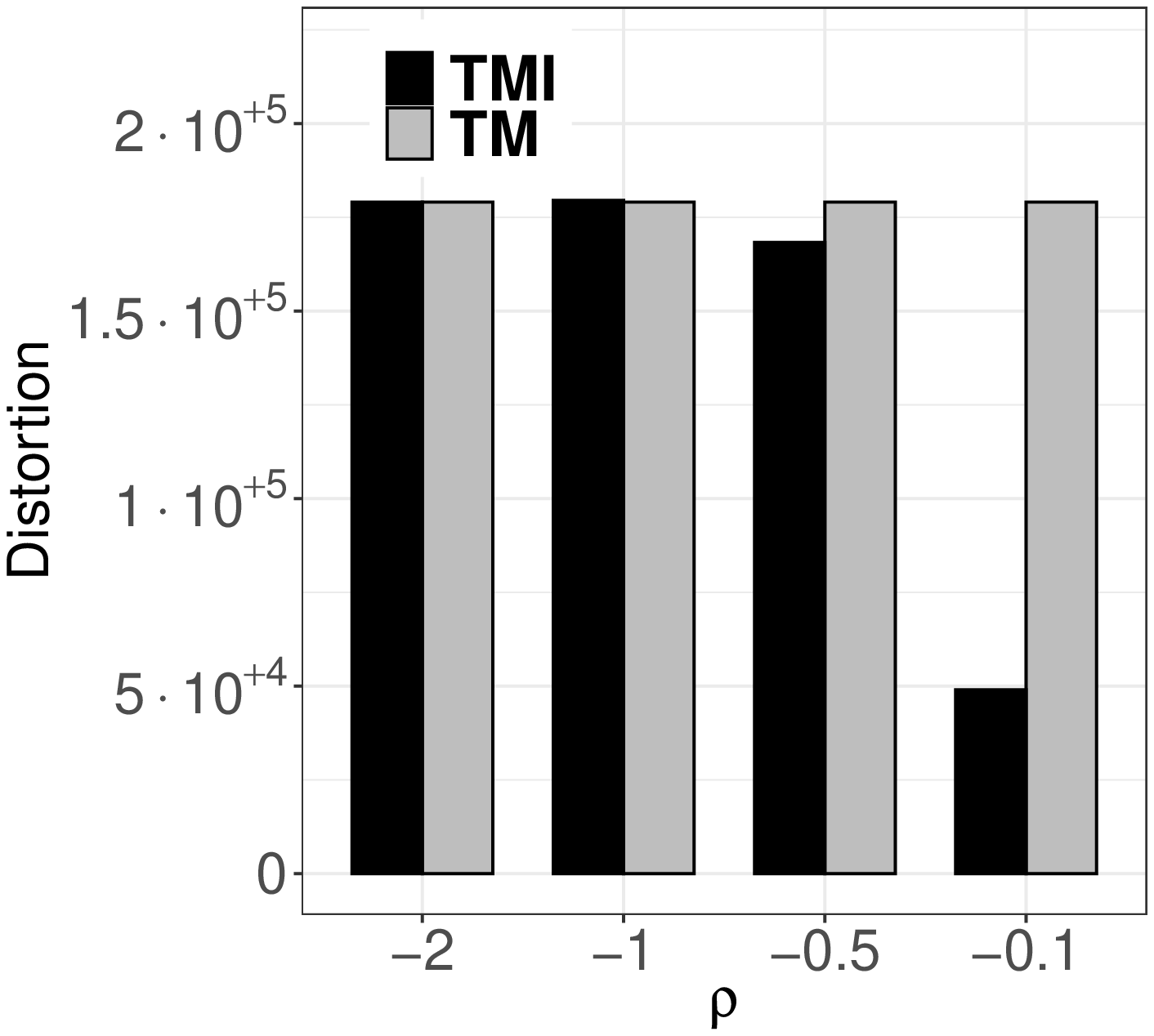}
      \vspace{-7mm}
        \caption{\OLDEN}
        \label{TM_TMI_vs_rho_disto_olden}
    \end{subfigure}
    \hspace{+2mm} 
    \begin{subfigure}[b]{0.23\textwidth}
      \includegraphics[width=1.05\textwidth]{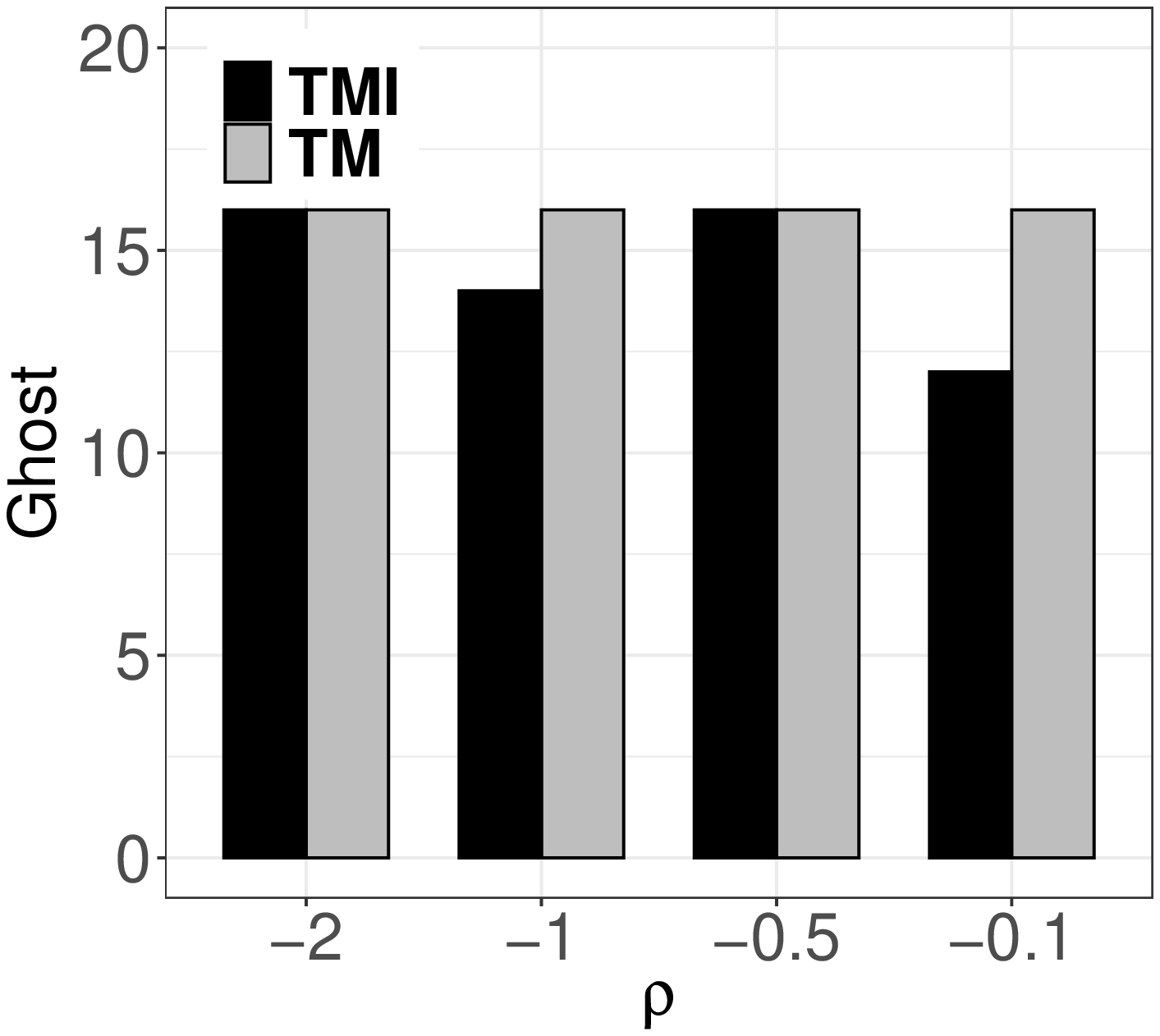}
      \vspace{-7mm}
        \caption{\OLDEN}
        \label{TM_TMI_vs_rho_ghost_olden}
    \end{subfigure}
     \hspace{+2mm} 

\caption{(a) Distortion, (b) number of $\tau$-ghost patterns, and (c) percentage of implausible patterns vs.~$\rho$.}\label{TM_TMI_vs_rho}
\end{figure}

\subsection{{\TFSA} vs.~{\ETFSA}}

We demonstrate that {\TFSA} is a very effective heuristic for the {\ETFS} problem. Specifically, it constructs a string $X$ that is either an optimal solution to the problem or it is at slightly larger edit distance from $W$ compared to the exact solution string $X_{\text{ED}}$ that is constructed by {\ETFSA}. This can be seen from Fig. \ref{etfs_vs_k_syn} (resp., \ref{etfs_vs_s_syn}), which shows that {\TFSA} constructed optimal solutions (\ie{} Edit Distance Relative Error was $0$) in $98\%$ (resp., $93\%$) of the tested strings, on average. These strings are uniformly random and have the same length and alphabet as {\SYNB}. Qualitatively similar results were obtained for uniformly random strings of different lengths and alphabet sizes (omitted).  In addition, the effectiveness of {\TFSA} can be seen from Figs. \ref{etfs_vs_k_tru} and \ref{etfs_vs_s_tru}, which show that the Edit Distance Relative Error in {\TRU}  was no more than $2.8\%$. These results are encouraging because, unlike {\ETFSA},  {\TFSA} is applicable to large strings such as {\OLDEN}, {\MSN}, and {\DNA} (recall that its time complexity is linear instead of quadratic in $|W|$).

\begin{figure}[!ht]\hspace{-5mm}
    \centering
    \begin{subfigure}[b]{0.23\textwidth}
      \includegraphics[width=1.1\textwidth]{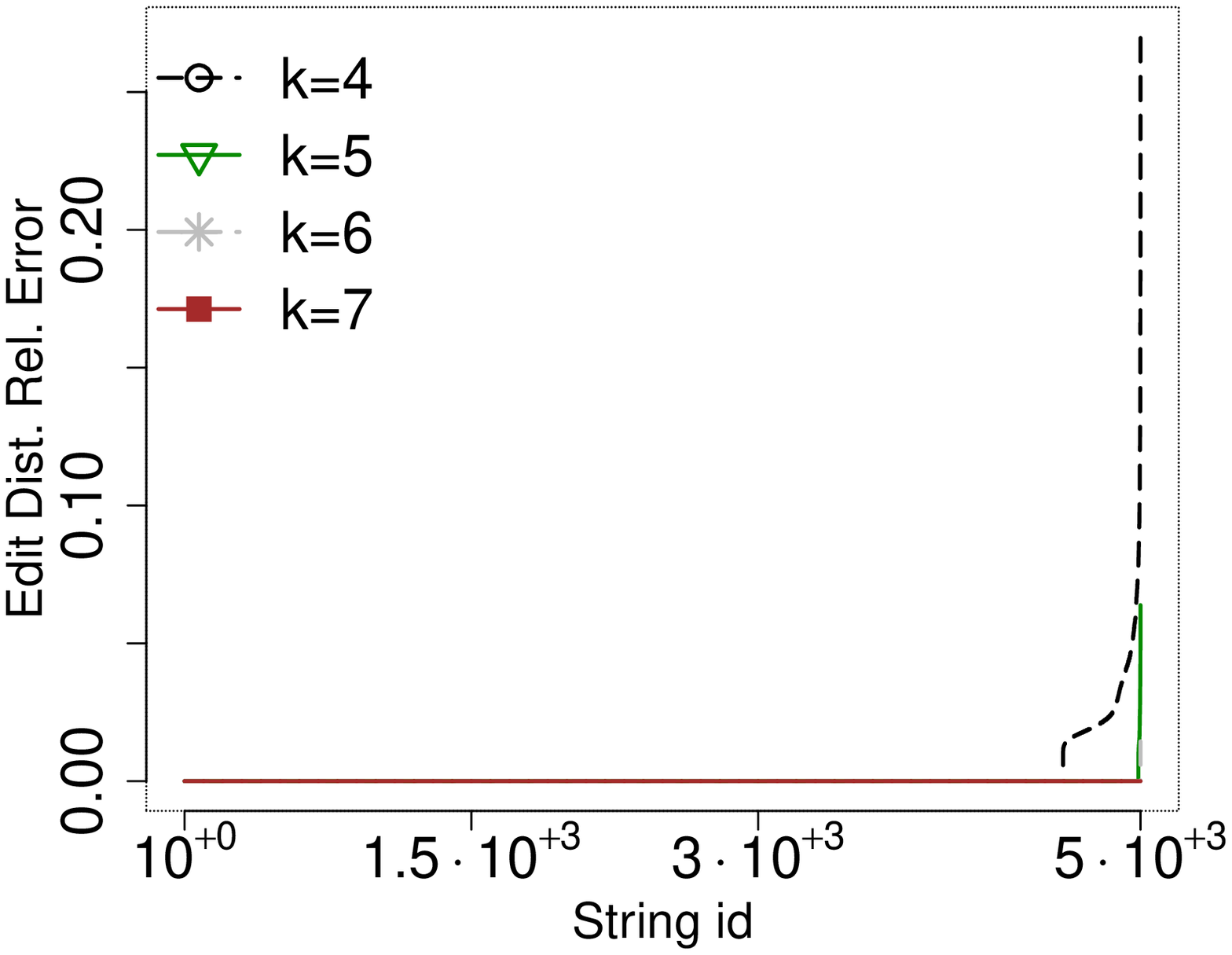}
      \vspace{-5mm}
        \caption{\SYNB}
        \label{etfs_vs_k_syn}
    \end{subfigure}
    \hspace{+3mm} 
    \begin{subfigure}[b]{0.23\textwidth}
      \includegraphics[width=1.1\textwidth]{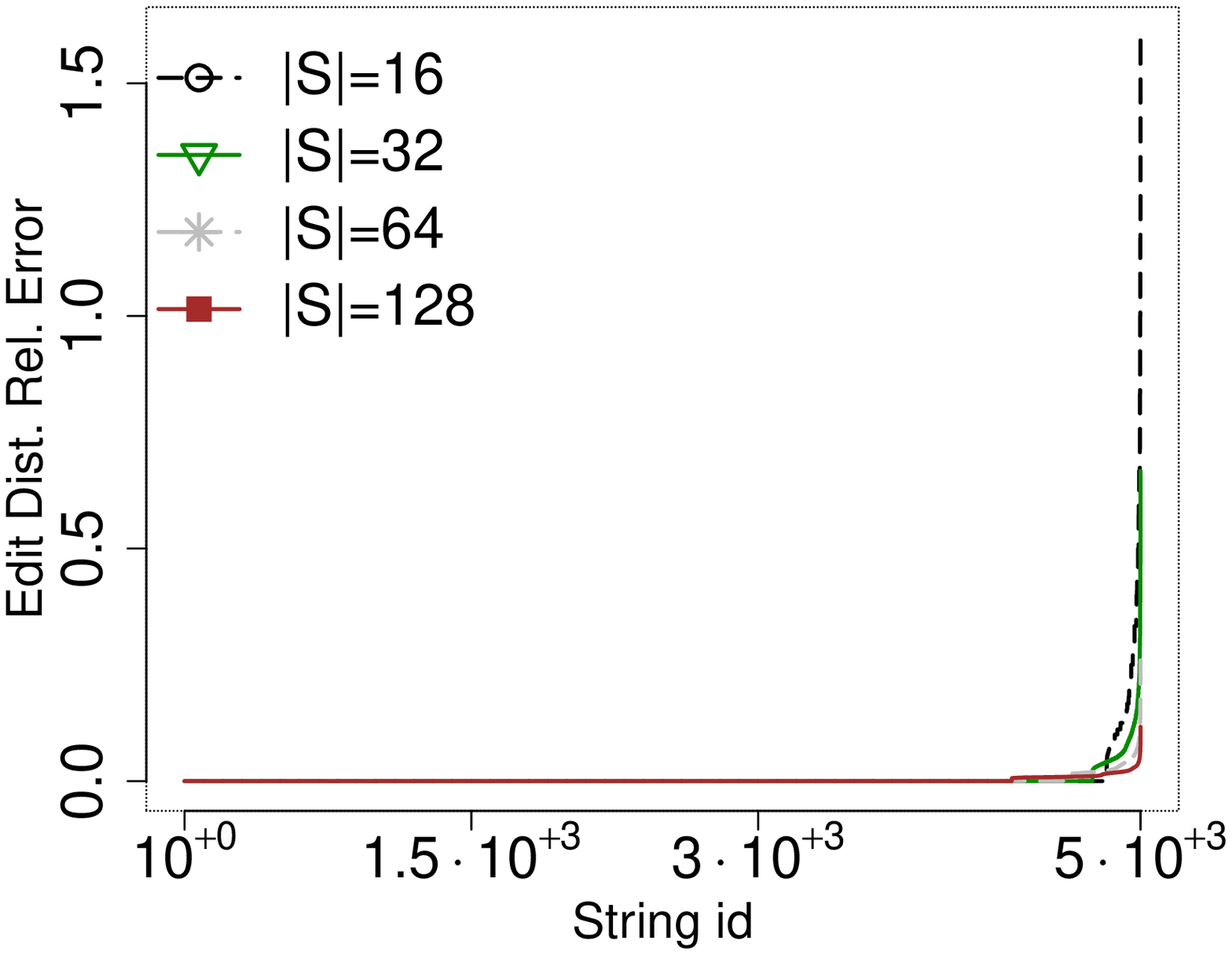}
      \vspace{-5mm}
        \caption{\SYNB{}}
        \label{etfs_vs_s_syn}
    \end{subfigure}
    \hspace{+2mm} 
    \begin{subfigure}[b]{0.23\textwidth}
      \includegraphics[width=1.04\textwidth]{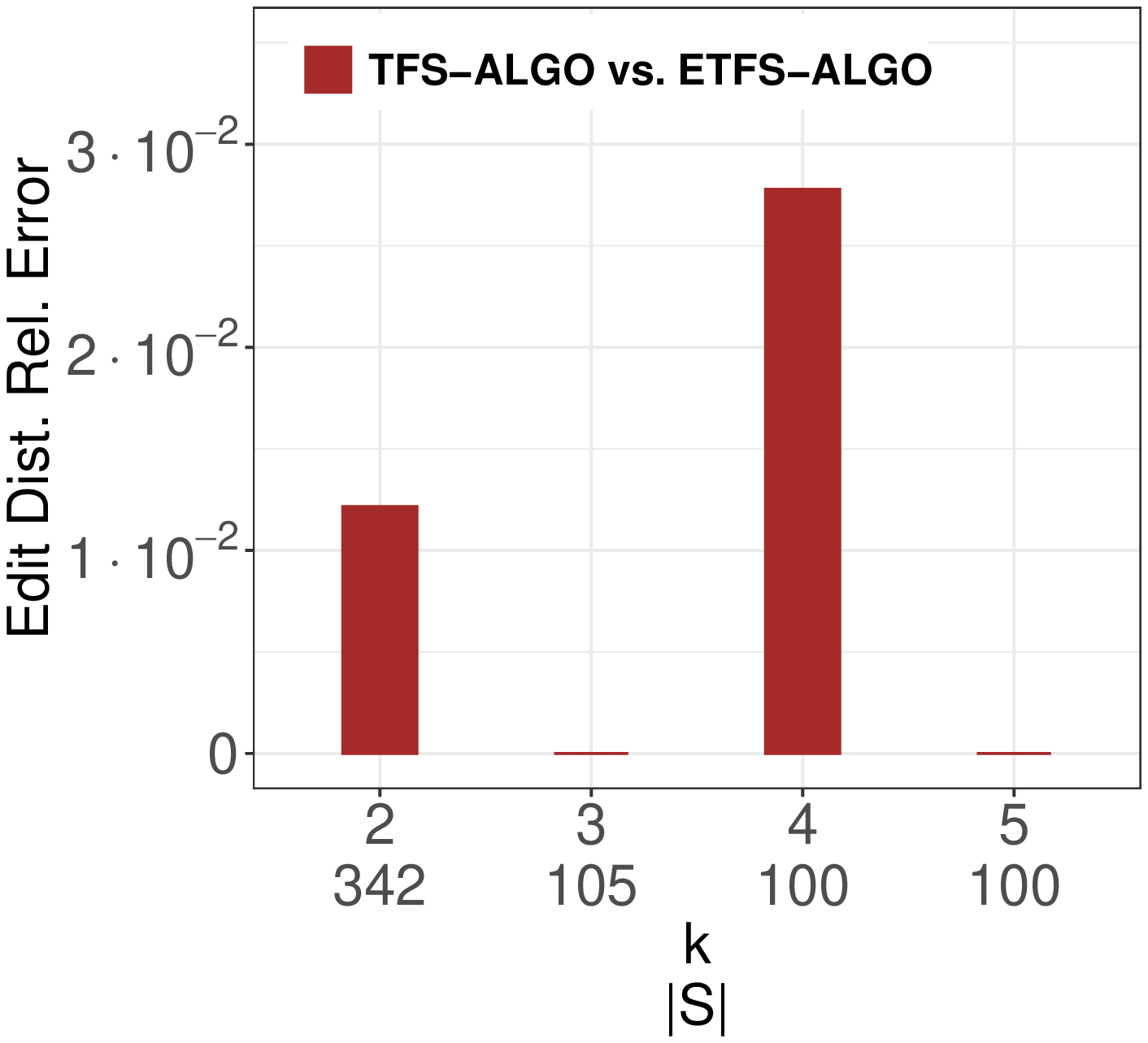}
      \vspace{-7mm}
        \caption{\TRU{}}
        \label{etfs_vs_k_tru}
    \end{subfigure}
     \hspace{+2mm} 
    \begin{subfigure}[b]{0.23\textwidth}
      \includegraphics[width=1.04\textwidth]{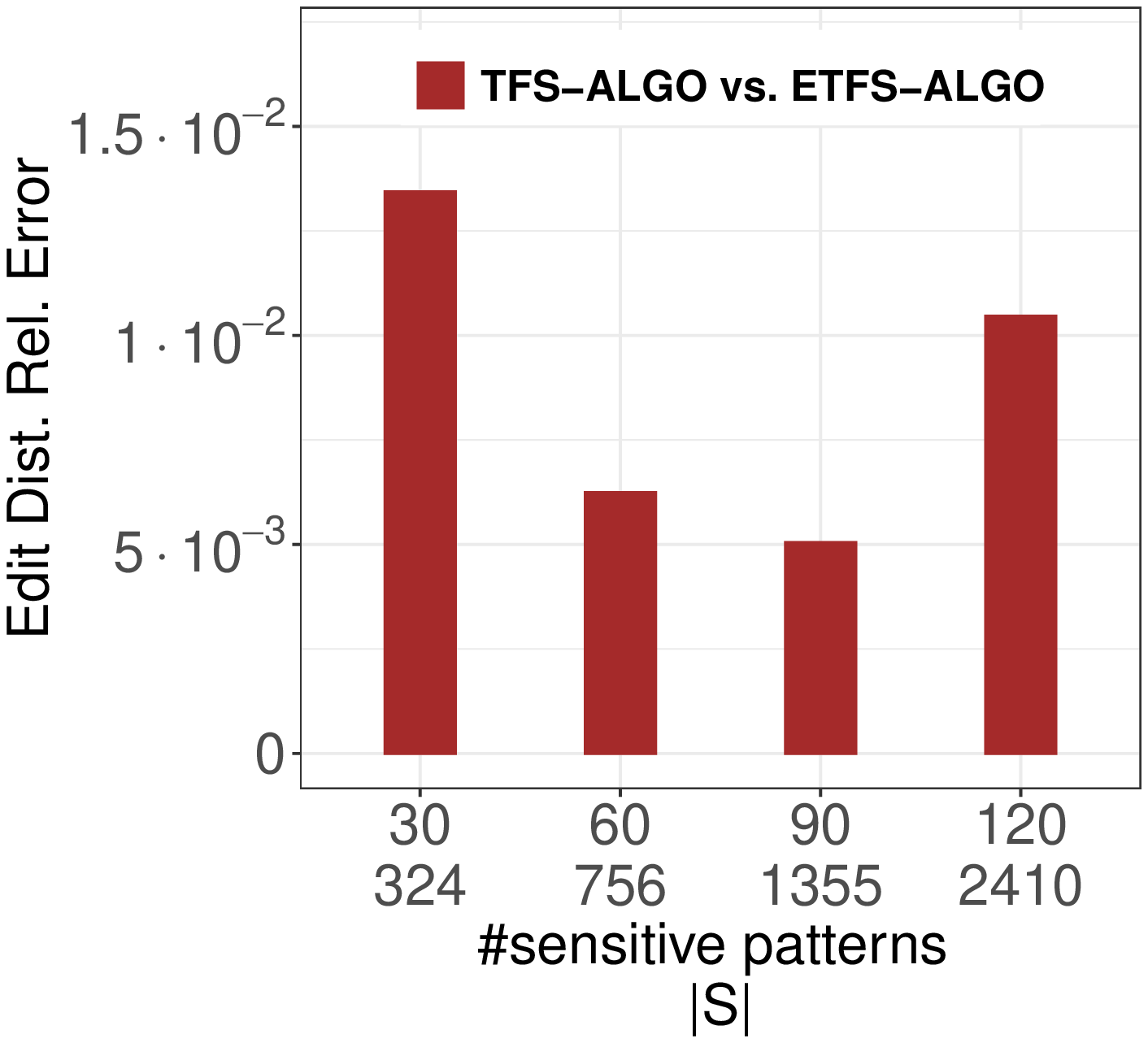}
      \vspace{-7mm}
        \caption{\TRU{}}
        \label{etfs_vs_s_tru}
    \end{subfigure}

\caption{Edit Distance Relative Error vs.~(a) $k$ (and $|\mathcal{S}|$), and (b) number of sensitive patterns (and  $|\mathcal{S}|$) for each of the $50,000$ random strings. Edit Distance Relative Error vs.~(c) $k$ (and $|\mathcal{S}|$), and (d) number of sensitive patterns (and  $|\mathcal{S}|$) for {\TRU}. }\label{etfs_syn_tru}
\end{figure}

\section{Related Work}\label{sec:finale}

{\em Data sanitization} aims at concealing confidential information from a dataset prior to its dissemination. In privacy-preserving data mining, data sanitization (\textit{a.k.a.} knowledge hiding) aims at concealing patterns modeling confidential knowledge by limiting their frequency, so that they are not easily mined from the data. Existing methods are applied to: (I) a \emph{collection} of set-valued data (transactions)~\cite{verykios} or spatiotemporal data (trajectories)~\cite{abul}; (II) a \emph{collection} of sequences~\cite{sbsh,icdm13}; or (III) a \emph{single} sequence~\cite{liyue,odesa,streamevent}. Yet, none of these methods follows our CSD setting: Methods in category I are not applicable to string data, and those in categories II and III do not have guarantees on privacy-related constraints~\cite{streamevent} or on utility-related properties~\cite{sbsh,icdm13,liyue,odesa}. Specifically, unlike our approach,~\cite{streamevent} cannot guarantee that all sensitive patterns are concealed (constraint {\bf C1}),  while~\cite{sbsh,icdm13,liyue,odesa} do not guarantee the satisfaction of utility properties (\eg \textbf{$\boldsymbol{\Pi} 1$} and {\bf P2}).

{\em Data anonymization} is a different direction in privacy-preserving data mining which is applied to individual-specific data and aims to prevent the disclosure of individuals' identity and/or information that individuals are not willing to be associated with  \cite{condensation,anna,differentialprivacy}. On the other hand, our approach is applied to a string modeling information that does not necessarily refer to specific individuals and aims to protect sensitive patterns that model confidential knowledge rather than values individuals do not want to be associated with. For example, our approach may be applied to a string comprised of letters corresponding to orders of different products by a business. In this case, subsequences of ordered products that provide competitive advantage to the business~\cite{icdm13}  are treated as sensitive patterns and should be concealed from the disseminated string. The fact that anonymization methods deal with individual-specific data and aim to prevent privacy threats other than confidential knowledge exposure leads to fundamentally different protection principles and methods than ours. For instance, differential privacy \cite{differentialprivacy} is a well-known anonymization principle and anonymization methods based on condensation \cite{condensation} have been proposed for strings \cite{condensation,charusdm}. Our work is related to anonymization approaches in that it shares the general objective of protecting string data with~\cite{condensation,charusdm} and that of protecting data while supporting string mining with the works of~\cite{bonomi} and~\cite{chen}. However, our work considers different input data  and has a fundamentally different privacy objective than \cite{condensation,charusdm,bonomi,chen}. Specifically, these works consider a collection of strings instead of a single long string and employ privacy objectives which do not aim to reduce the frequency of sensitive length-$k$ substrings to zero. Therefore, they cannot be applied to address the problems considered in this paper.

\section{Conclusion}\label{sec:conclusion}

In this paper, we introduced the Combinatorial String Dissemination model. The focus of this model is on {\em guaranteeing} privacy-utility trade-offs in sequential data (\eg {\bf C1} {\em vs.} {\bf $\boldsymbol{\Pi} 1$} and {\bf P2}).

Under this model, we considered two different settings. The common privacy constraint in both settings is that the output string must not contain any sensitive pattern. In the first setting, we aim to generate the minimal-length string that preserves the order of appearance and the frequency of all non-sensitive patterns. We defined a problem, \TFS, to capture these requirements, and a variant of it, \PFS, that preserves a partial order and the frequency of the non-sensitive patterns but generally produces a shorter string. We developed two time-optimal algorithms, {\TFSA} and {\PFSA}, for {\TFS} and {\PFS}, respectively. We also developed {\MCSRA}, a heuristic that prevents the disclosure of the location of sensitive patterns, ensuring that sensitive patterns are not reinstated, implausible
patterns are not introduced, and occurrences of spurious patterns are prevented from the outputs of {\TFSA} and {\PFSA}.  In the second setting, we aim to generate a string that is at minimal edit distance from the original string, in addition to preserving the order of appearance and the frequency of all non-sensitive patterns. We defined a problem, {\ETFS}, to capture these requirements, and proposed {\ETFSA}, an algorithm, which is based on solving specific instances of approximate regular expression matching, to construct such a string.

Our experiments show that string sanitization by {\TFSA}, {\PFSA} and then {\MCSRA} is both effective and efficient. They also demonstrate that {\TFSA} can be employed as an effective heuristic to the {\ETFS} problem producing optimal or near-optimal solutions  in practice.

\bibliographystyle{splncs04}
\bibliography{references}

\end{document}